\newif\ifarxiv\arxivtrue  % for arXiv which disallows line numbers
\newif\ifseechange\seechangetrue
\seechangefalse  % uncomment iff only want to see new version of \changed stuff 

\newif\ifseepoplchange\seepoplchangetrue
\seepoplchangefalse

\documentclass[acmsmall,screen]{acmart}

\setcopyright{rightsretained}
\acmPrice{}
\acmDOI{10.1145/3571213}
\acmYear{2023}
\copyrightyear{2023}
\acmSubmissionID{popl23main-p114-p}
\acmJournal{PACMPL}
\acmVolume{7}
\acmNumber{POPL}
\acmArticle{20}
\acmMonth{1}
\received{2022-07-07}
\received[accepted]{2022-11-07}

\usepackage{color}

\usepackage{latexsym} 
\usepackage{amsmath}
\usepackage[lightning,llbracket,rrbracket,fatsemi]{stmaryrd} 
\usepackage{amstext} % for \text macro
\usepackage{mathpartir} 
\usepackage{stackengine} 
\usepackage{diagrams}
\usepackage{relsize}
\usepackage{xcolor}
\usepackage{listings}
\usepackage{wrapfig}

\definecolor{codegreen}{rgb}{0,0.6,0}
\definecolor{codegray}{rgb}{0.5,0.5,0.5}
\definecolor{codepurple}{rgb}{0.58,0,0.82}
\definecolor{backcolour}{rgb}{1,1,1}
\lstdefinestyle{pstyle}{
backgroundcolor=\color{backcolour},   
    commentstyle=\color{codegreen},
    keywordstyle=\color{blue},    
    stringstyle=\color{codepurple},
    basicstyle=\ttfamily\footnotesize,
    morekeywords={do,od,fi,then},
    breakatwhitespace=false,
    xleftmargin=2pt,
	numbers=none,
	numbersep=4pt,
	rulecolor=\color{codepurple},
    breaklines=true,                 
    captionpos=b,                    
    keepspaces=true,
    showspaces=false,                
    showstringspaces=false,
    showtabs=false,                  
    tabsize=2,
	mathescape=true,
	framexleftmargin=0mm,
	escapechar=|,
}

\newcommand\kcode[1]{\text{\lstinline|#1|}}
\newenvironment{theoremApp}[1]{{\noindent\sc Theorem \ref{#1} (restated).} \it}{}

\begin{document}

%\ifarxiv
%\pagestyle{plain} % to get page numbers and eliminate header - must remove in camera ready
%\fi

\ifarxiv
\title{An Algebra of Alignment for Relational Verification (Extended Version)}
\else
\title{An Algebra of Alignment for Relational Verification}
\fi

\author{Timos Antonopoulos}
\affiliation{\institution{Yale University}\country{USA}}
\author{Eric Koskinen}
\author{Ton Chanh Le}
\author{Ramana Nagasamudram}
\author{David A. Naumann}
\author{Minh Ngo}
\affiliation{\institution{Stevens Institute of Technology}\country{USA}}
%\email{webmaster@marysville-ohio.com}
%\affiliation{%
%  \institution{Institute for Clarity in Documentation}
%  \streetaddress{P.O. Box 1212}
%  \country{USA}


% Dave's TeX hacking with help from Knuth p357, 359

\def\rightharpoonupfill{$\mathsurround=0pt \mathord- \mkern-6mu
  \cleaders\hbox{$\mkern-2mu \mathord- \mkern-2mu$}\hfill
  \mkern-6mu \mathord\rightharpoonup$}

\def\leftharpoonupfill{$\mathsurround=0pt \mathord\leftharpoonup \mkern-6mu
  \cleaders\hbox{$\mkern-2mu \mathord- \mkern-2mu$}\hfill
  \mkern-6mu \mathord-$}

\def\overleftharpoonup#1{\vbox{\ialign{##\crcr
      \leftharpoonupfill\crcr\noalign{\kern-1pt\nointerlineskip}
      $\hfil\displaystyle{#1}\hfil$\crcr}}}

\def\overrightharpoonup#1{\vbox{\ialign{##\crcr
      \rightharpoonupfill\crcr\noalign{\kern-1pt\nointerlineskip}
      $\hfil\displaystyle{#1}\hfil$\crcr}}}

\def\overleftrightharpoonup#1{\vbox{\ialign{##\crcr
      \leftharpoonupfill\hspace*{-.7em}\rightharpoonupfill\crcr\noalign{\kern-1pt\nointerlineskip}
      $\hfil\displaystyle{#1}\hfil$\crcr}}}
% End Dave's hack

% notes 
\newcommand\timos[1]{{[\color{orange} TA: #1]}}
\newcommand\dave[1]{{[\color{magenta} DN: #1]\marginpar{\color{red}$\checkmark$}}}
\newcommand\eric[1]{{[\color{olive} EK: #1]}}
\newcommand\ramana[1]{{[\color{violet} RN: #1]}}

% \changed{old}{new} 
\ifseechange
\newcommand\changed[2]{\textcolor{red}{#1}\textcolor{blue}{#2}}
\else 
\newcommand\changed[2]{\textcolor{blue}{#2}}
\fi

\newcommand\added[1]{\textcolor{blue}{#1}}

\newcommand\popladded[1]{\poplchanged{}{#1}}
\newcommand\poplremoved[1]{\poplchanged{#1}{}} 

\definecolor{aqua}{rgb}{0.0, 1.0, 1.0}
\definecolor{bleudefrance}{rgb}{0.19, 0.55, 0.91} % maybe
\definecolor{amethyst}{rgb}{0.6, 0.4, 0.8} % no
\definecolor{ceruleanblue}{rgb}{0.16, 0.32, 0.75} % maybe 
\definecolor{deepskyblue}{rgb}{0.0, 0.75, 1.0}
\definecolor{majorelleblue}{rgb}{0.38, 0.31, 0.86}
\definecolor{mediumblue}{rgb}{0.0, 0.0, 0.8} %no 

\definecolor{deepskyblue}{rgb}{0.0, 0.55, 1.0}

\ifseepoplchange
\newcommand\poplchanged[2]{\textcolor{red}{#1}\textcolor{bleudefrance}{#2}}
\else
\newcommand\poplchanged[2]{#2}
\fi

% gray box to highlight defined notations 
\definecolor{light-gray}{gray}{0.88}
%\definecolor{dark-gray}{gray}{0.25}
\newcommand{\graybox}[1]{\colorbox{light-gray}{#1}} 

% defined term well highlighted
\newcommand\dt[1]{\textbf{\emph{#1}}} 

\newcommand\nat{\mathbb{N}} % natural numbers 
\newcommand\kA{\mathbb{A}}
\newcommand\kB{\mathbb{B}}
\newcommand\kdot{\cdot}
\newcommand\ksemi{\kcode{;}}
\newcommand\kstar{*}
\newcommand\kplus{+}
\newcommand\kneg{\neg}
\newcommand\kNeg[1]{\overline{#1}}
\newcommand\kone{1}
\newcommand\kzero{0}
\newcommand\bdots[1]{\ddot{#1}}
\newcommand\bA{\bdots{\mathbb{A}}}
\newcommand\bB{\bdots{\mathbb{B}}}
\newcommand\bdot{\fatsemi} % Dave likes this better than \odot; it's also known as \zcmp 
\newcommand\bplus{\oplus}
\newcommand\bstar{\circledast}
\newcommand\bneg{\ddot{\neg}}
\newcommand\beq{\ddot{=}}
\newcommand\bone{\bdots{1}}
\newcommand\bzero{\bdots{0}}
\newcommand\bR{\mathcal{R}} % relation on state pairs 
\newcommand\bS{\mathcal{S}} % relation on state pairs 
% alert: following are upper case, to avoid conflict with standard latex macros
\newcommand{\Left}[1]{\overleftharpoonup{#1}} % a left embedding of unary action/test
\newcommand{\Right}[1]{\overrightharpoonup{#1}} % a right embedding of unary action/test
\newcommand{\lproj}{\mathord{\swarrow}} % projection operator 
\newcommand{\mproj}{\mathord{\downarrow}}
\newcommand{\rproj}{\mathord{\searrow}}
\newcommand{\proj}[1]{\mathord{\downarrow_{#1}}}

\newcommand{\hav}{\mathbf{hav}} % unary avoc

\newcommand{\fst}{\mathit{fst}} % first projection function 
\newcommand{\snd}{\mathit{snd}} % second projection 
\newcommand{\Dom}{\mathord{\triangleleft}} % domain of a relation 
\newcommand{\Cod}{\mathord{\triangleright}} % codomain of a relation
\newcommand{\sDom}{\mathord{\mbox{\footnotesize$\triangleleft$}}} % domain of a trace
\newcommand{\sCod}{\mathord{\mbox{\footnotesize$\triangleright$}}} % codomain of trace 

\newcommand{\luproj}{\hat{\mathord{\sswarrow}}}
\newcommand{\ruproj}{\mathord{\ssearrow}}
\newcommand{\llproj}{\mathord{\sswarrow}\!\!\!\mathord{\sswarrow}}
\newcommand{\rlproj}{\mathord{\ssearrow}\!\!\!\mathord{\ssearrow}}
\newcommand{\auxlproj}{\mathord{\circ}\!\;\!\!\!\!\mathord{\sswarrow}}

\newcommand\eml[1]{\langle #1 ]} % left embedding
\newcommand\emr[1]{[ #1 \rangle} % right embedding
\newcommand\emb[2]{\langle #1 \,|\, #2 \rangle} % joint embedding 
\newcommand\embRepeat[1]{\emb{#1}{#1}}
\newcommand\Embb[2]{\Big\langle \begin{matrix} #1 \\ #2 \end{matrix} \Big\rangle} % joint embedding 
\newcommand\Emb[2]{\langle #1 \mid #2 \rangle} % \emb with a bit more space 
\newcommand{\eqdef}{\mathrel{ \hat{=} }} % definition
\newcommand\defeq{\eqdef}

\newcommand{\aftqua}{.\:}
\newcommand{\all}[2]{\forall #1 \aftqua #2} % forall quantified formulas
\newcommand{\some}[2]{\exists #1 \aftqua #2} % exists quant

% specs
\newcommand{\specSym}{\leadsto} 
\newcommand{\rspecSym}{\ensuremath{\mathrel{\mbox{\footnotesize$\thickapprox\hspace*{-.4ex}>$}}}}
\newcommand{\aespecSym}{\ensuremath{\mathrel{\mbox{%
\footnotesize$\stackrel{\exists}{\thickapprox\hspace*{-.4ex}>}$}}}}
\newcommand{\bespecSym}{\ensuremath{\mathrel{\mbox{%
\footnotesize$\stackrel{\exists\shortleftarrow}{\thickapprox\hspace*{-.4ex}>}$}}}}

\newcommand{\spec}[2]{#1\specSym #2} 
\newcommand{\rspec}[2]{#1\rspecSym #2} 
\newcommand{\aespec}[2]{#1\aespecSym #2} 
\newcommand{\bespec}[2]{#1\bespecSym #2} 

\newcommand{\proves}{\vdash}

\newcommand{\subst}[3]{{#1}^{#2}_{#3}} % substitution 

% left-right separator in RHL judgments 
\newcommand{\sep}{\mathbin{\mid}} % larger version of | 
\newcommand{\Sep}{\mathbin{\:\mid\:}} % still larger version of | 

% rule names (consistent with mathpartir)
\newcommand{\rn}[1]{\textsc{#1}} 

% code in math font 
\newcommand{\keyw}[1]{\ensuremath{\mathsf{#1}}} 
\newcommand{\skipc}{\keyw{skip}}
\newcommand{\ifc}[3]{\keyw{if}\ {#1}\ \keyw{then}\ {#2}\ \keyw{else}\ {#3}}
\newcommand{\whilec}[2]{\keyw{while}\ {#1}\ \keyw{do}\ {#2}}
%\newcommand{\ifc}[3]{\keyw{if}\ {#1}\ \keyw{then}\ {#2}\ \keyw{else}\ {#3}\ \keyw{fi}}
%\newcommand{\whilec}[2]{\keyw{while}\ {#1}\ \keyw{do}\ {#2}\ \keyw{od}}
%\newcommand{\varblock}[2]{\keyw{var}~ #1 ~\keyw{in}~ #2} %\var{x:T}{body}
%\newcommand{\choice}[2]{#1\mathrel{\sqcup}#2} % 

% hint in proofs with between-lines format 
\newcommand{\hint}[1]{\quad\mbox{#1}} 

% left/right embedding of formulas in RHL, to discuss 
\newcommand{\leftex}[1]{ \raisebox{.25ex}{\relsize{-1}$\langle\hspace*{-2.1pt}[$} #1 \raisebox{.25ex}{\relsize{-1}$\langle\hspace*{-2.0pt}]$} }
\newcommand{\rightex}[1]{ \raisebox{.25ex}{\relsize{-1}$[\hspace*{-2.0pt}\rangle$} #1 \raisebox{.25ex}{\relsize{-1}$]\hspace*{-2.1pt}\rangle$} }
\newcommand{\leftF}[1]{\eml{#1}}
\newcommand{\rightF}[1]{\emr{#1}}

\newcommand{\eqbib}[2]{#1\beq#2}

\newcommand{\imp}{\Rightarrow}
\newcommand{\impby}{\Leftarrow}
\renewcommand{\iff}{\Leftrightarrow}

% for diagrams 
\newcommand{\Bforall}{\mbox{\Large$\forall$\quad}}
\newcommand{\Bexists}{\mbox{\Large\quad$\exists$\quad}}
\newcommand{\BforallX}{\mbox{\Large$\forall$}}
\newcommand{\BexistsX}{\mbox{\Large$\exists$}}
\newcommand{\dRel}{\dEmbed}

% TriKAT 
\newcommand{\triEmb}[2]{\langle #1 \talloblong #2 \rangle} % embed from bi to tri 
\newcommand{\tricom}[3]{\langle #1 \mid #2 \mid #3\rangle} 
\newcommand{\Lproj}{\mathord{\Swarrow}}
\newcommand{\Mproj}{\mathord{\Downarrow}} % middle projection of TriKAT

\newcommand{\eqbi}{\mathrel{\ddot{=}}}

\newcommand\mkEqBi[1]{[\kcode{#1} \eqbi \kcode{#1}]}
\newcommand\etal{\emph{et al.}}

\begin{abstract} % Kent Beck guidelines 
% Context/Importance
Relational verification encompasses information flow security,
regression verification, translation validation for compilers, 
and more.
Effective alignment of the programs and computations to be related
facilitates use of simpler relational invariants 
and relational procedure specs, 
which in turn 
enables automation and modular reasoning.
% Problem
Alignment has been explored in terms of trace pairs, deductive rules of relational Hoare logics (RHL),
and several forms of product automata. 
% Startling sentence
This article shows how a simple extension of Kleene Algebra with Tests (KAT), called BiKAT, 
subsumes prior formulations,
including alignment witnesses for forall-exists properties,
which brings to light new RHL-style rules for such properties.
% Implications
Alignments can be discovered algorithmically or devised
manually but, in either case, their adequacy with respect to the original programs must be proved;
an explicit algebra enables constructive proof by equational reasoning.
Furthermore our approach inherits algorithmic benefits from existing KAT-based techniques and tools, which are applicable to a range of semantic models.
\end{abstract}

% HACK: omit classification/keywords to avoid breaking later formatting,
% since the arxiv version's longer title adds vertical space
\ifarxiv 
\else


\fi

\maketitle

\section{Introduction}\label{sec:intro}

A number of important program requirements are not trace properties but can be
defined as 2-properties.  For example, secure information flow says that any two executions 
with the same ``low'' (non-secret) inputs have the same low outputs.   
Continuity says two executions from very close inputs produce very close outputs.
Pairs of executions are also the basis for relations between programs,
such as extensional equivalence, refinement, or simulation.

One way to prove such a relational property is to prove a strong functional property
of each program, and show that the relational property is a consequence.
However, from early work on refinement~\cite{deRdataref,Concur:deRoever}
through explicit studies of relational program logic~\cite{Francez83,Benton:popl04}
to current work on automated relational reasoning (e.g.\ see~\cite{BeckertU18short})
it has been clear that one should align intermediate points in programs and their executions,
in order to decompose the reasoning.
In many cases it is much better to reason in terms of a well chosen alignment, 
which can make it possible to reason using logically simple
relational invariants such as conjunctions of equalities between variables.  
This is especially important for automated reasoning, since restricted fragments such as linear arithmetic facilitate techniques like Horn clause solving to find invariants.
It is also important in the many situations where a functional specification is not available.
Moreover, alignment enables use of relational specifications and summaries for procedures.

For these reasons, alignment appears in various guises such as the same-structure (``diagonal'') proof 
rules of Benton's relational Hoare logic~\cite{Benton:popl04} and program product constructions based on syntax~\cite{BartheCK-FM11,SousaD2016} and on product automata~\cite{BartheCK13,ChurchillP0A19}.
For the moment let us use the term \dt{aligned product} for such constructions.
The proof of a given relational property 
has two parts:
the chosen aligned product satisfies the property, and the aligned product is \dt{adequate}, in the sense that it represents all computations of the program(s) of interest.  
Verification can be automated by searching for possible alignments and checking whether they can be proved to satisfy the property, or by synthesizing an alignment in tandem with attempting to prove the property~\cite{FarzanV2019,ChurchillP0A19,AntonopoulosKL2019,ShemerGSV19,UnnoTK21}.

The current state of the art offers a number of variations on these ideas, in \poplremoved{seemingly} disparate formulations\poplchanged{,
and not all providing complete solutions with rigorous adequacy proofs.}{.}
In this paper we introduce an algebra
for alignment products which makes it possible to account for adequacy by equational reasoning.   Our algebra,
called BiKAT,  is a simple extension of Kleene Algebra with Tests (KAT)~\cite{Kozen1997},
itself an algebra of programs.  Our representation of alignment products is thus immediately amenable to 
various program analysis tools and may serve as a framework for procedures that search for alignments.

In KAT, a partial correctness assertion $\{p\}c\{q\}$ is expressed by an equation:
\poplremoved{one form is} $p;c;\neg q = 0$.  
Now consider
what we call the $\forall\forall$ \dt{relational judgment}, written $c|c':\rspec{R}{S}$,
meaning that: from a pair of states 
related by $R$, terminated executions of $c$ and $c'$ respectively, end in states related by $S$.
The relational judgment can be expressed by the similar equation 
$R;\emb{c}{c'};\neg S = 0$, using the BiKAT form $\emb{c}{c'}$ that represents pairs of executions.
We show that rules of relational Hoare logics can be derived 
in BiKAT, which applies to many semantic models, 
not just the specific semantic models used in prior works such as those cited above.

Not all relational properties have the termination-insensitive $\forall\forall$ (2-safety) form described above.  For programs which may 
exhibit nondeterminacy, typical requirements involve existential quantification.
For example, possibilistic noninterference says that 
given two low-indistinguishable states,
and a terminated execution from the first, 
there is a terminated execution \popladded{from the second,} with low-indistinguishable final 
state~\cite{Sabelfeld:Myers:JSAC,ClarksonSchneiderHyper10}.
As we explain, attempts to formulate such properties in the algebra lead to a three-run equation which we have found somewhat unwieldly.
So, in addition we provide a usable reduction to equations on two-run alignment products.

With this reduction in hand, we discovered that one can derive rules for reasoning about 
two fundamental and ubiquitous kinds of $\forall\exists$ properties: forward and backward simulation.
We derived inference rules for a logic of forward simulation and a logic of backward simulation.
In retrospect, these logics seem natural and it is surprising that (to our knowledge) such rules have
not appeared in the literature.  
As often happens, devising an algebraic description of the models of interest leads to new insights about existing ideas, and also results far beyond the initial motivations.

\paragraph*{Contributions}
\begin{itemize}

\item We define a novel algebra, BiKAT, that can express the 
  alignment of programs in both relational and trace semantic
  models, building in a simple way on the well understood KAT.

\item Using BiKAT, we derive rules of relational Hoare 
  logic for $\forall\forall$ properties, just as Kozen's seminal work derives the rules for
  conventional Hoare logic~\cite{Kozen00}.  

\item Using BiKAT equations, we characterize the $\forall\exists$ forward and backward simulation properties, 
  and use these results to derive new proof rules for these properties. 
  To our knowledge this provides the first deductive system for backward simulation judgments.

\end{itemize}

\paragraph*{Outline}
Sect.~\ref{sec:overview} is an overview of the problem and
shows BiKAT in action.
Sect.~\ref{sec:prelim} briefly reviews KAT.
Sect.~\ref{sec:bikat} defines BiKAT and some standard models.
Sect.~\ref{sec:examples} demonstrates direct use of BiKAT
on illustrative examples of relational reasoning challenges.
Examples are given throughout to show how BiKAT can express and justify alignments used in the literature on relational verification.
Sect.~\ref{sec:RHL} derives rules for the $\forall\forall$ judgment.
Sect.~\ref{sec:beyond} addresses $\forall\exists$ properties:
we develop characterizations in BiKAT and use these to derive proof rules 
for forward and backward simulation. 
We also report on explorations for $\exists\forall$ and $\exists\exists$ properties, 
and sketch a notion of TriKAT that deserves further study. 
Sect.~\ref{sec:discuss} discusses automation,
product programs, expressiveness, 
and future directions, all in the context of related work.
\ifarxiv
\emph{This extended version of the paper adds an Appendix (Sect.~\ref{sec:appendix}) with many proofs and examples.
It also completes the proof of Theorem~\ref{thm:sim:complete},
strengthens that result, and strengthens Theorems~\ref{thm:fsim:rules} and~\ref{thm:bsim:rules}
accordingly.
}
\fi

Accompanying this paper is a Coq formalization that includes the basic definitions and results about BiKATs~\cite{BiKATv}.  
\ifarxiv
\else
There is also an extended version of the paper with an appendix working out many proofs and examples~\cite{BiKATarxiv}.
\fi

\section{Examples of Alignment in BiKAT}\label{sec:overview}

\newcommand\rr{\kcode{r}}
\newcommand\nn{\kcode{n}}

We now introduce our algebra of alignment, through a series of examples that highlight the features of our theory. Recall that Kleene Algebra with Tests (KAT)~\cite{Kozen1997} can be used to represent and reason about the behavior of programs. Atomic program statements can be treated as primitive ``actions'' in the KAT and conditions can be treated as primitive ``tests.'' For example, consider the following program that computes factorial of \kcode{n} and stores it in \rr:
\[
\begin{array}{ll}
\textrm{Code:} & \kcode{C1: i:=n; r:=1; while i!=0 do r:=r*i; i:=i-1 od} \\
\textrm{KAT:} & k_1 : \kcode{i:=n} \kdot \kcode{r:=1} \kdot  ( \kcode{[i$\neq$0]} \kdot \kcode{r:=r*i} \kdot \kcode{i:=i-1})^\kstar
   \kdot \neg \kcode{[i$\neq$0]}
\end{array}
\]

The above KAT expression combines primitives with sequential composition $\kdot$ and Kleene star-iteration $\kstar$. Tests are denoted with square brackets to distinguish from non-test actions 
(in a KAT tests are a subset of actions). 
There are various semantic models of KATs, notably relational models, where actions are state relations, and 
trace models where actions denote sets of traces.
KAT permits one to express properties through terms/equations that involve programs and their specifications.
For example, KAT subsumes Hoare logic: the KAT equation $p \kdot k_1 \kdot \neg q = 0$ expresses validity of the Hoare triple $\{p\} \kcode{C}_1 \{q\}$.

\begin{example}%\label{eg:one}
\emph{Dependency or non-interference.}
Let us now consider a basic $\forall\forall$ property of a single program (2-safety): whether different output results can be observed by varying input values. A basic pre/post 2-safety $\forall\forall$ property, 
written $C1 \mid C1 : [\nn \eqbi \nn] \approx\!> [\rr \eqbi \rr]$, requires that for any two executions of $C1$, if they agree on the initial value of \nn, 
then they will agree on the final value of \rr, regardless of the initial values of the other variables.
Here $[\nn \eqbi \nn]$ means value of \nn\ in the left program is equal to the value of \nn\ in the right program.
Our notation does not require that the given program(s) have
been modified to act on disjoint parts of a single state.

In this paper we will describe an algebra for expressing these and other relational verification problems, as well as the alignments necessary for proving them. We will introduce a relational analog of KAT called \emph{BiKAT}, which allows us to, for example, express this 
2-safety judgment as:
\begin{equation}\label{eq:NIspec}
  [\nn\eqbi \nn] \kdot 
  \emb{ k_1 }{ k_1 }\kdot
  \kneg [\rr\eqbi \rr] = \kzero
  \end{equation}
The connectives above are relational analogs of the KAT connectives:
$\kdot$ is sequential composition, $\kneg$ is negation, and $\kzero$ is the empty relation.
A key feature of BiKAT is to be able to homomorphically \emph{embed} a unary KAT expression $k$ into a relational setting, that acts according to $k$ on one side and acts as the identity on the other side. Embedding $k$ on the ``left'' is denoted $\eml{k}$, on the ``right'' is $\emr{k}$, and we write $\emb{k}{k'}$ to mean $\eml{k}\kdot\emr{k'}$.
This embedding resembles the popular sequential product (a.k.a.\ self-composition~\cite{BartheDArgenioRezk}), but does not require $k$ and $k'$ to act on disjoint variables.
A BiKAT is a (special kind of) KAT and, in addition to embedded unary actions/tests, 
can be equipped with its own relational primitive actions and tests. 
The agreement $[\rr\eqbi \rr]$ above is a primitive ``bitest.''
Embedding and bitests are sequentially composed together in the equation above with $\kdot$ into an overall equation that characterizes validity of the relational judgment, by requiring that post-disagreement on \rr\ is equivalent to the empty set of behaviors.

While the above is merely a BiKAT problem statement of non-interference, the principle benefits of BiKAT arise when we begin to use our equational theory of BiKAT to derive alignments that ease (or make tractable) the task of relational verification. Notice that, thus far, the above embedding would still require one to reason that $k_1$ computes the factorial function, which is more work than truly necessary for non-interference. 
(This issue was highlighted in the seminal work that introduced the term 2-safety~\cite{TerauchiA2005}.)
We can instead exploit
that the left/right embeddings are homomorphisms, and exploit 
the  algebra of the embedded KAT expressions, to find another equation that implies the above one. In Sect.~\ref{sec:bikat} we will discuss the details that lead to the following:
\[\begin{array}{l}
    [\nn \eqbi \nn] \kdot 
\emb{ \kcode{i:=n} }{ \kcode{i:=n} } \kdot
\emb{ \kcode{r:=1} }{ \kcode{r:=1} }\\
\;\;\kdot\left(
\emb{ \kcode{[i!=0]} \kdot \kcode{r:=r*i} \kdot \kcode{i:=i-1} }{ \kcode{[i!=0]} \kdot \kcode{r:=r*i} \kdot \kcode{i:=i-1} } 
\right)^\kstar
\kdot \emb{ \kcode{[i==0]} }{ \kcode{[i==0]} } \kdot
  \kneg [\rr\eqbi\rr] \; = \; \kzero  % \bzero hasn't been introduced yet
\end{array}\]
This equation is derived from (\ref{eq:NIspec}) using BiKAT laws, 
the fact that embedding distributes through the KAT operators, and the fact that left and right embeddings commute: $\eml{k}\kdot\emr{k'}=\emr{k'}\kdot\eml{k}$.
We have now aligned the programs so that their loops iterate in ``lock step''. (We will see more complicated alignments below.)  For this example, this alignment is sufficient because, using a few 
axioms that embody the semantics of primitive tests and assignments
we can introduce relations within the \emph{body} of the BiKAT term, using bitests $[\kcode{i}\eqbi\kcode{i}]$, $[\nn\eqbi\nn]$, $[\rr\eqbi\rr]$ expressing agreement on \kcode{i}, \nn, \rr. Overall this lets us conclude the original property without having to resort to reasoning about the full functional behavior of the programs involved.
\end{example}

\newcommand\Axx{[\kcode{x}\eqbi\kcode{x}]}
\newcommand\Ayy{[\kcode{x}\eqbi\kcode{x}]}
\newcommand\AzTz{[\kcode{z}\eqbi 2\kcode{z}]}
\newcommand\ATyy{[2\kcode{y}\eqbi\kcode{y}]}

\begin{example}\label{eg:two}
\emph{Aligning two different programs.} Consider the programs from~\citet{ShemerGSV19}: %saving space
\[\begin{array}{ll}
D_1: & \kcode{y:=0}\kdot\kcode{z:=2*x}\kdot(\kcode{[z>0]}\kdot\kcode{z--}\kdot\kcode{y:=y+x})^\kstar 
\kdot\kcode{[z<=0]}\\
D_2: & \kcode{y:=0}\kdot\kcode{z:=x}\kdot(\kcode{[z>0]}\kdot\kcode{z--}\kdot\kcode{y:=y+x})^\kstar 
\kdot\kcode{[z<=0]}\kdot\kcode{y:=y*2}\\
\end{array}\]
These programs compute $2\kcode{x}^2$, but in  different ways: $D_1$ iterates $2\kcode{x}$ times, each time adding \kcode{x}  to the result \kcode{y}, whereas $D_2$ iterates only \kcode{x} times, but then multiplies by 2 at the end. This example has a flavor of a compiler optimization, where a program may be transformed to save on loop iterations.
Here too we can formulate the property of equal outputs from equal inputs and, using the laws of BiKAT to rewrite the problem into this aligned form:
\[\begin{array}{l}
    \Axx \kdot 
\emb{ \kcode{y:=0} }{ \kcode{y:=0} } \kdot
\emb{ \kcode{z:=2*x} }{ \kcode{z:=x} }\kdot\AzTz\\
\;\;\kdot\left(
  \AzTz\ATyy 
\emb{ \mathit{body}\kdot\mathit{body} }{ \mathit{body} } 
\right)^\kstar\\
\;\;\kdot \ATyy \kdot \emb{ \kcode{[z==0]} }{ \kcode{[z==0]}\kdot \kcode{y:=y*2} } \kdot
\Ayy
\kdot
  \kneg \Ayy  = \kzero % \bzero hasn't been introduced yet
\end{array}\]
where $\mathit{body} \defeq \kcode{[z>0] z--; y:=y+x}$.
(Throughout this paper, we often use juxtaposition or \kcode{;} to mean the $\kdot$ KAT operator.)
This example uses an alignment where two iterations of the loop in the first program are related to one iteration of the second one, to preserve the relational invariant that \kcode{z} in the left program is twice the value of \kcode{z} in the right program and that \kcode{y} in the right program is twice the value of \kcode{y} in the left program.
To prove this we introduce agreement bitests in the body of the loop and use KAT-based inductive reasoning to show those bitests, including $\Axx$, are indeed invariant. 
With these simple relational bitests, we can prove the equivalence between $D_1$ and $D_2$.
\end{example}

\begin{example}\label{eg:three}
\emph{Mixing program algebra with alignment algebra.} Now consider two programs that iterate over C-style arrays:

\begin{lstlisting}
$L_1$: x := 0; while (x < N * M){ a[x] := f(x); x++;}  
$L_2$: i := 0; while (i < N){ j := 0; while (j < M){ A[i,j] := f(i*M+j); j++; i++;}}
\end{lstlisting}

In $L_1$ there is a one-dimensional array, \kcode{a}, of size \kcode{N*M}, while $L_2$ uses an \kcode{N}-by-\kcode{M} two-dimensional array \kcode{A}.
We prove that the two programs are equivalent modulo this change in data
representation.
This example is used in~\citet{BartheCK13} to argue for alignment based on a control flow graph representation.
A custom rewriting relation, with KAT-like rules, is used in~\citet{BanerjeeNN16} to handle the example
using syntactic RHL rules.

Aligning these programs is challenging because they have fundamentally different shapes to them (single versus nested loops).
To address this example we exploit the fact that the programs can be represented as KAT terms  which themselves can be algebraically manipulated, whilst embedded within the BiKAT.
Our alignment (detailed in Sect.~\ref{sec:examples}) first uses unary KAT laws to transform $L_1$ so that it consists of nested loops. We then align those nested loops and introduce loop alignment invariants including an agreement on indexing $[\kcode{x} \eqbi \kcode{i}\times\kcode{M}+\kcode{j}]$ and an agreement on array values
$[\kcode{a[x]} \eqbi \kcode{A[i,j]}]$. 
In this way, the laws of BiKAT (and underlying KAT laws) allow one to use equational reasoning to simplify the task of deriving useful alignments, and then introduce simple relational invariants that suffice to prove the overall property.

Beyond the above examples, BiKAT can also be used to construct \emph{data-dependent} alignment (e.g.,~\cite{ShemerGSV19}).
A single iteration in one program may need to be aligned with 
a varying number of iterations in a second program, dependent on the data. 
Example~\ref{eg:caWh} will demonstrate BiKAT's capability to support such alignments.
\end{example}

\begin{example}\label{eg:four}
\emph{Expressing $\forall\exists$ properties.}
Consider programs
$E_1: \kcode{x:=any; y:=x}$ and
$E_2: \kcode{t:=any; z:=t+1}$.
These simple programs choose a nondeterministic value (keyword \kcode{any}) and then use that value to make an assignment.  
A common specification pattern is that for any pair of pre-related states, and every execution of $E_1$, there are choices that $E_2$ can make to ensure that the post-states are related. The pattern is useful for program equivalence and 
also for information flow where it is known as \emph{possibilistic non-interference}~\cite{ClarksonSchneiderHyper10}.
For this example, let us use \kcode{true} as the pre-relation and $[\kcode{y} \eqbi \kcode{z}]$ as the post-relation.

In Sect.~\ref{sec:exists} we describe how BiKAT can be used to support $\forall\exists$ properties such as the above.
We again use BiKAT to derive alignments but for $\forall\exists$ we also introduce bitests that intentionally restrict some of the behaviors of the programs. Let $W$ be this ``witness,'' consisting of the aligned programs, along with bitest restrictions.
We reduce the problem to three conditions on $W$ (Theorem~\ref{thm:fsim}):
\begin{itemize}
\item[(WC)] The corresponding $\forall\forall$ property holds of $W$.
\item[(WO)] $W$ over-approximates the behavior of the left program.
\item[(WU)] $W$ under-approximates the behavior of the right program.
\end{itemize}
For the above example, we can align the programs in lock-step, and use the bitest $[\kcode{x-1} \eqbi \kcode{t}]$ so that, after the assignment \kcode{t:=any}, executions are then restricted to those where \kcode{t} is \kcode{x-1}.  Doing so will ensure that the post-relation $[\kcode{y} \eqbi \kcode{z}]$ holds (WC). Moreover, this bitest does not restrict the behaviors of the left program (WO), nor does it add behaviors on the right that are not already allowed by the right program (WU). In Sect.~\ref{sec:forwardBack} we also provide deductive rules for such $\forall\exists$ properties.
\end{example}

\section{Preliminaries}\label{sec:prelim}

A Kleene Algebra with Tests (KAT)~\cite{Kozen1996,Kozen1997}
is a two-sorted structure
$(\kA,\kB,\kplus,\kdot,\kstar,\kneg,\kone,\kzero)$, where
$\kB\subseteq \kA$, such that 
$\kB$ is closed under the operations $\kplus,\kdot,\kneg$ 
and these satisfy the laws of Boolean algebra.
Moreover $(\kA,\kplus,\kdot,\kstar,\kone,\kzero)$ is a Kleene algebra.

We \poplremoved{will} use $a,b,c,\ldots$ for elements of $\kA$ (the ``actions'') and $p,q\ldots$ for elements of $\kB$ (the ``tests'').
We sometimes write $a \kdot b$ as $a;b$ or $ab$ and let it bind tighter than $+$.
The axioms of Kleene algebra are:

\[\begin{array}{rcl@{\hspace{5em}}rcl}
a \kplus (b \kplus c) &=& (a \kplus b) \kplus c 
& 
\kone \kplus a\kdot a^\kstar &=& a^\kstar \\
a \kplus b &=& b \kplus a 
& 
\kone \kplus a^\kstar\kdot a &=& a^\kstar \\
a \kplus \kzero &=& a 
&
b \kplus a\kdot c \leq c &\imp& a^\kstar\kdot b \leq c \\
a \kplus a &=& a 
&
b \kplus c\kdot a \leq c &\imp& b\kdot a^\kstar \leq c \\
a \kdot (b \kdot c) &=& (a \kdot b) \kdot c \\
\kone \kdot a &=& a \\
a \kdot \kone &=& a \\
a \kdot (b \kplus c) &=& a \kdot b \kplus a \kdot c \\
(a \kplus b)\kdot c &=& a \kdot c \kplus b \kdot c \\
\kzero \kdot a &=& \kzero \\
a \kdot \kzero &=& \kzero 
\end{array}
\]
where $\leq$ is the  partial order defined by $a \leq b \iff a \kplus b = b$.
Some useful consequences are the \dt{sliding} law 
$a\kdot (b\kdot a)^* = (a\kdot b)^*\kdot a$
and the \dt{invariance} law 
$p\kdot a \leq p\kdot a\kdot p \imp p\kdot a^* \leq p\kdot a^*\kdot p$
(where $p$ is a test).
The correctness equation $p\kdot c \kdot \neg q = 0$ 
has useful equivalent forms:
$p\kdot c = p\kdot c\kdot q$ and 
$p\kdot c \leq p\kdot c\kdot q$.
Every test $p$ satisfies $p\leq \kone$.

Given some set $\Sigma$, a \dt{relational model} 
is $(\kA,\kB,\kplus,\kdot,\kstar,\kneg,\kone,\kzero)$
where $\kA$ is a set of relations\footnote{We always mean binary relations.}
on $\Sigma$,
where $1$ is the identity relation $id_\Sigma$,
$0$ is the empty relation, 
$\cdot$ is relational composition, $*$ is reflexive-transitive closure, $+$ is union of relations, $\kB$ is a set of \dt{sub-identities}, i.e., subsets of the identity relation $1$,
and $\kneg b$ is the complement $1\setminus b$.
Note that $\leq$ is set inclusion.\footnote{Also, $\kA$ is closed under these operations, but need not be the full set $\wp(\Sigma\times\Sigma)$ of all relations;
and $\kB$ need not include all sub-identities.} 
Relational models can be obtained from big-step or denotational semantics of programs,
with $\Sigma$ the set of states.  

A \dt{full} relational model is one where $\kA$ is all relations on $\Sigma$
and $\kB$ is all sub-identities.
A full relational model has 
a \dt{top} element, which we call \dt{havoc} and write $\hav$.  
Being top means $a\leq \hav$ for all $a$ in $\kA$.

A \dt{trace model}~\cite{Kozen2003,Kozen04} 
$(\kA,\kB,\kplus,\kdot,\kstar,\kneg,\kone,\kzero)$ is given in terms of some set of primitive actions and some set $\Sigma$ of states.
An element of $\kA$ is a set of traces, where a trace is a nonempty alternating sequence 
of states and primitive actions, beginning and ending with a state.  
(The definitions are with respect to a given set of admissible traces, which one can
consider as possible observations for a programming language of interest.
For example, if each primitive action has an associated relation on states,
then traces may be required to be consecutive with respect to those relations.)
An element of $\kB$ is a set of singleton traces (essentially a set of states) and
negation is defined as complement with respect to $\Sigma$.    
The $+$ operation is union.  
For traces $t$ and $u$, the \dt{coalesced catenation} $t\diamond u$ is defined only if the last state of $t$ is the first of $u$, in which case the sequences are catenated but omitting one copy of that state.
Otherwise $t\diamond u$ is undefined.
For $a$ and $b$ in $\kA$, the set $a\cdot b$
is defined to be 
$\{ t \diamond u \mid t\in a \land u \in b \land t\diamond u \mbox{ defined} \}$.
The star operation is given by iterated catenation.
A \dt{full} trace model is one that contains all sets of admissible traces
(so it has a top).

\popladded{A KAT is \dt{*-continuous}~\cite{Kozen1997} provided that for 
all $x,y,z$ the supremum $\sup_{n\in\nat}  (x\kdot y^n \kdot z)$ (with respect to $\leq$)
exists and $x\kdot y^\kstar \kdot z = \sup_{n\in\nat}  (x\kdot y^n \kdot z)$.
Here $y^n$ is the iterate defined by $y^0 = \kone$ and $y^{n+1} = y\kdot y^n$.
Relational and trace models are *-continuous.
}

\pagebreak
\section{BiKAT}\label{sec:bikat}

\subsection{Definitions and Basic Results}

\begin{definition}\label{def:bikat}
A \dt{BiKAT} over the KAT $(\kA,\kB,\kplus,\kdot,\kstar,\kneg,\kone,\kzero)$ is a KAT
$(\bA,\bB,\bplus,\bdot,\bstar,\bneg,\bone,\bzero)$
with
KAT homomorphisms $\eml{\_} : \kA\to\bA$ and $\emr{\_} : \kA\to\bA$, 
which we call the \emph{left} and \emph{right embeddings},
that satisfy \emph{left-right commutativity}: 
\[ (LRC) \qquad \eml{x}\bdot\emr{y}=\emr{y}\bdot\eml{x} \quad\mbox{for all $x,y$ in $\kA$.}\]
We call $\kA$ the \dt{underlying} KAT.
\popladded{A BiKAT is \dt{*-continuous} if both $\kA$ and $\bA$ are.}
%  and $\bA$ the \dt{2d-level KAT} of the BiKAT.    DN: doesn't seem needed
\end{definition}
\popladded{
For $\eml{\_}$ to be a KAT homomorphism means the following:
\begin{itemize} 
\item $\eml{x} \in \bB$ for all $x\in\kB$
\item $\eml{\kzero} = \bzero$, $\eml{\kone} = \bone$, and $\eml{\kneg x} = \kneg\eml{x}$ for all $x\in\kB$
\item 
$\eml{x\kplus y} = \eml{x}\bplus\eml{y}$, 
$\eml{x\kdot y} = \eml{x}\bdot\eml{y}$, 
and $\eml{x^\kstar} = \eml{x}^\bstar$  
for all $x,y$ in $\kA$
\end{itemize}
\emph{Mutatis mutandis} for $\emr{\_}$.
}

An obvious generalization is to consider two different underlying KATs
(e.g., modeling source and target semantics, for compiler correctness).  
It is also possible to generalize to $k$-KATs;
we return to this in Sect.~\ref{sec:trikat}.
But our focus is on 2-trace properties.
Specializing to 2 lets us streamline notation and use convenient ``left/right'' terminology.

Define \graybox{$\emb{\_}{\_}$} by
\[ \emb{a}{b} \eqdef \eml{a}\bdot\emr{b} \]
We call this the two-argument embedding.
Note that $\eml{a} = \emb{a}{1}$ and $\emr{a}=\emb{1}{a}$
because $\emr{1}=\bone$ and $\bone$ is the identity element of $\bdot$.
Another immediate consequence is that $\emb{\_}{\_}$ is a homomorphism 
to $\kA$ from the product KAT $\kA\times\kA$,
with the additional property that $\emb{0}{a}=0=\emb{a}{0}$ for all $a\in \kA$.\footnote{An equivalent 
way to define BiKAT is to take $\emb{\_}{\_}$ as primitive
and define $\eml{a}=\emb{a}{1}$ and $\emr{a}=\emb{1}{a}$.
Then we get left-right commutativity because 
$\eml{a};\emr{b}=
\emb{a}{1};\emb{1}{b}=  
\emb{a;1}{1;b}= 
\emb{1;a}{b;1}= 
\emb{1}{b};\emb{a}{1}=  
\emr{b};\eml{a}$.
For $\eml{\_}$ and $\emr{\_}$ defined this way to be homomorphisms requires 
the property $\emb{0}{a}=0=\emb{a}{0}$
since in $\kA\times\kA$ the values 
$(a,0)$, $(0,0)$, and $(0,a)$ are different in general.
Since left-right commutativity is the key property we care about, we choose a 
formulation that highlights it.  
}
In case $\kA$ is a relational model of KAT,
the elements of $\kA\times\kA$ are pairs of relations on some set $\Sigma$.
By contrast, the relational BiKAT over $\kA$ will
comprise relations on $\Sigma\times\Sigma$ (Sect.~\ref{sec:models}).

Being KAT homomorphisms, the embeddings send tests to tests.  
For tests $p,q$ in $\kA$ we have $\neg\emb{p}{q}=\neg(\eml{p}\bdot\emr{q})=\eml{\neg p}\bplus\emr{\neg q}$
by homomorphism and de Morgan.

We use identifiers $A,B,C\ldots$ for elements of $\bA$ 
and $P,Q,R,S\ldots$ for \dt{bitests}, i.e., elements of $\bB$.
While we use fancy symbols $\bone,\bplus,\bdot,\ldots$ in Def.~\ref{def:bikat}
to distinguish between the BiKAT's operations and those of the underlying KAT,
we usually use the simpler $+$, $\kdot$ or $;$, $\kneg$, etc., for both levels
of a BiKAT since the types can be inferred from context.

The  $\forall\forall$ judgment \graybox{$c\sep c' : \rspec{R}{S}$} can be interpreted 
in any BiKAT, provided $R$ and $S$ are bitests and $c$ and $c'$ are actions in 
the underlying KAT.  It is like the Hoare triple equation:
\begin{equation}\label{eq:allallbikat}
R ; \emb{c}{c'} ; \neg S = 0 
\end{equation}
Defining the judgment this way is justified in terms of the standard models (see Sect.~\ref{sec:models} \popladded{and Theorem~\ref{thm:allall}}).

In applications of KAT, one reasons under hypotheses that specify the interpretation of 
primitive tests and actions.  For many applications it suffices to use hypotheses in the form of Hoare triples.
In applications of BiKAT we can use relational Hoare triple hypotheses to specify the interpretation
of primitive bitests, usually with respect to embedded unary actions.  
The key point is to leverage the use of such hypotheses by rewriting the given verification task
into a conveniently aligned form.  This can be done in ad hoc ways but several general patterns
can be identified.

The following general law is useful for reasoning about two loops
$\whilec{e}{c}$ and $\whilec{e'}{c'}$.
They are aligned lockstep, under a common star, with a remainder to account for whichever may iterate longer.\footnote{This is the most flexible loop alignment 
available in some systems, e.g., Cartesian Hoare logic~\cite{SousaD2016}.}
Here $e,e'$ are tests in the underlying KAT.
\begin{equation}\label{eq:expand}
\begin{array}{l}
\emb{(e;c)^*}{(e';c')^*};\emb{\neg e}{\neg e'} 
 = 
\emb{e;c}{e';c'}^*;(\emb{e;c}{\neg e'}^* + \emb{\neg e}{e';c'}^*) ;\emb{\neg e}{\neg e'} 
\end{array}
\end{equation}
This holds in any BiKAT, i.e., it follows by equational reasoning from Def.~\ref{def:bikat}.
(see appendix of~\citet{BiKATarxiv}).
Note that 
$\emb{(e;c)^*}{(e';c')^*};\emb{\neg e}{\neg e'}
= \emb{(e;c)^*;\neg e}{(e';c')^*;\neg e'}$ by LRC.

The standard models of BiKAT have additional structure that we formalize as follows.

\begin{definition}\label{def:bikatproj}
A \dt{BiKAT with projections} is a BiKAT together with 
total functions $\lproj : \bA\to\kA$ and $\rproj : \bA\to\kA$
to its underlying KAT, 
that satisfy the following for all $a$ in $\kA$ and $A,B$ in $\bA$.
\[
\begin{array}{lll}
\mbox{(Inversion)} &
\lproj \eml{a} = a & \rproj \emr{a} = a \\
\mbox{(Disjointness)} &
a\neq 0 \imp \lproj \emr{a} = \bone &
a\neq 0 \imp \rproj \eml{a} = \bone \\
\mbox{(Disjunctivity)} 
\quad & 
\lproj(A \bplus B) = \lproj A + \lproj B 
\quad
&
\rproj(A \bplus B) = \rproj A + \rproj B
\end{array}
\]
\end{definition}
\poplchanged{Please note that although projection distributes over sum, it is}{Projections are} not required to distribute over sequence or star.
For most of our development, the projections are not needed.  However, they exist 
in the standard models and they do serve a purpose that raises open problems discussed in Sect.~\ref{sec:trikat}.

For BiKATs with projections, we get the following easy consequences.
Only the first condition is expressed using projection, but all of them are proved using projection properties.
\begin{lemma}%\label{lem:bikat}
In any BiKAT with projections, we have 
\begin{center}
{\upshape
\begin{tabular}{ll}
(Unit) & $\lproj\bone = 1 = \rproj \bone$
\\ 
(Separation) & $\eml{a} = \emr{b} \land (a \neq 0 \lor b \neq 0) \imp a = b = 1$
\\
Order Separation) & $\eml{a} \geq \emr{b} \land a\neq 0 \neq b \imp a \geq 1 \land 1 \geq b $
\\ 
(Injectivity) & $\eml{a}=\eml{b}  \imp a=b$ and same for $\emr{\_}$ 
\\
(Order-Injectivity) & $\eml{a}\geq\eml{b}  \imp a\geq b$ and same for $\emr{\_}$ 
\end{tabular}
}
\end{center}
\end{lemma}

An attractive aspect of KAT is that implications $hypoth \imp c = d$ are decidable provided the hypotheses
are of the form $b=0$ which includes the KAT encoding of Hoare triples.
On the other hand, implications with general commutativity hypotheses are undecidable~\cite{Kozen1996}.
One proof goes by reduction to the Post correspondence problem, and the proof idea
can be adapted to BiKAT embeddings in such a way that the commutativity conditions 
are expressed by LRC (see appendix of~\citet{BiKATarxiv}).
It should be noted that the original argument is presented for $*$-continuous KATs and as such our proof is over $*$-continuous BiKATs. 
\poplremoved{We remind the reader that *-continuity means that $x\kdot y^\kstar \kdot z = \sup_{n\in\nat}  (x\kdot y^n \kdot z)$ holds for all $x,y,z$~\cite{Kozen1997}.}

\begin{theorem}\label{thm:undecidability}
	It is undecidable whether a given identity holds in all $*$-continuous BiKATs.
\end{theorem}

We do not need decidability in order to derive useful alignments 
for specific programs, and to derive general laws like (\ref{eq:expand}).
Consider a given verification problem $R ; \emb{c}{c'} ; \neg S = 0$.
Alignment reasoning, using LRC and its consequences, transforms it to 
a problem of the form $R ; B ; \neg S = 0$ that is proved using unary and relational Hoare-triple hypotheses.
For that matter, 
having used BiKAT to obtain well-aligned $B$, verification subtasks in $B$ can be solved 
by whatever method you like.

\subsection{Models of BiKAT}\label{sec:models}

Elements of a relational or trace KAT can be seen as sets of 
observations, where an observation may be a proper trace or just a pair (initial,final)
of states.
Elements of a BiKAT are meant to be sets of observation pairs.  
We will describe trace models of BiKAT this way.
For relational BiKATs, however, we use a slightly different representation that facilitates spelling out connections with relational Hoare logics etc.

% FUTURE: may be better to use the other choice for relation models, and uniformly use
% first/last operators so defs look same as trace models.

First, some definitions. 
Let $R$ and $S$ be relations on some set $\Sigma$.  
Define $R\otimes S$ to be the relation on $\Sigma\times\Sigma$
such that $(\sigma,\sigma')(R\otimes S)(\tau,\tau')$ iff $\sigma R\tau$ and $\sigma' S\tau'$.
If $R$ and $S$ are the denotations of two programs, $R\otimes S$ relates 
a pair $(\sigma,\sigma')$ interpreted as initial states to a final pair $(\tau,\tau')$.
(To strictly model the idea of observation pair, one would use $((\sigma,\tau),(\sigma',\tau'))$
instead of $((\sigma,\sigma'),(\tau,\tau'))$.)

%((start left prog,start right prog), (end left prog, end right prog)).

For relation $R$ on $\Sigma$, the relations $\eml{R}$ and $\emr{R}$ on $\Sigma\times\Sigma$
are defined by $\eml{R} = R \otimes id_{\Sigma}$ and $\emr{R} = id_{\Sigma} \otimes R$.
(We write $id$ or $id_\Sigma$ for the identity relation on $\Sigma$.)
In terms of states, we have
\[ 
(\sigma,\sigma') \eml{R} (\tau,\tau') \;\iff\; \sigma R \tau \land \sigma'=\tau' 
\qquad\quad
   (\sigma,\sigma') \emr{R} (\tau,\tau') \;\iff\; \sigma' R \tau' \land \sigma=\tau
\]
For a relation $\bR$ on $\Sigma\times\Sigma$, define relations $\lproj\bR$ and $\rproj\bR$ on $\Sigma$ by
\begin{equation}\label{eq:BiKATproj}
 \begin{array}{l}
\sigma(\lproj\bR)\tau \;\iff\; \some{\sigma',\tau'}{ (\sigma,\sigma')\,\bR\,(\tau,\tau')} 
\qquad\quad
\sigma'(\rproj\bR)\tau' \;\iff\; \some{\sigma,\tau}{ (\sigma,\sigma')\,\bR\,(\tau,\tau')}  
\end{array}
\end{equation}
Equivalently, $\lproj\bR \eqdef \fst^o;\bR;\fst$ and $\rproj\bR \eqdef \snd^o;\bR;\snd$
where
$\fst:\Sigma\times\Sigma\to\Sigma$ and $\snd:\Sigma\times\Sigma\to\Sigma$ are
the projection functions, treated as relations so we can use the converse operation 
(written $^o$ as in~\cite{Freyd:Scedrov}) 		
and relational composition (written $;$).

A \dt{relational BiKAT} over a relational model $(\kA,\kB,+,\cdot,*,\neg,1,0)$ is
$(\bA,\bB,\bplus,\bdot,\bstar,\bneg,\bone,\bzero)$ such that
$\bA$ is a set of relations on $\Sigma\times\Sigma$
and the structure on $\bA$ is a relational model of KAT.
Moreover, both $\eml{R}$ and $\emr{R}$ are in $\bA$, for any $R$ in $\kA$, 
and both $\lproj\bR$ and $\rproj\bR$ are in $\kA$, for any $\bR$ in $\bA$.
Note that a relational BiKAT is a BiKAT with projections.

Given a relation $R$ on states, we sometimes write $\dot{R}$ for the associated 
sub-identity in $\bB$, i.e., $(\sigma,\sigma')\dot{R}(\tau,\tau')$ iff
$\sigma R \sigma' \land \tau=\sigma \land \tau'=\sigma'$.
Please note that $\dot{id_\Sigma}$ is different from $\bone$.

For an example, say a relation $R$ is boundedly nondeterministic if for any $\sigma$ there are finitely many $\tau$ with $\sigma R \tau$.  
Take $\kA$ to be all boundedly nondeterministic relations on $\Sigma$
and $\bA$ to be all boundedly nondeterministic relations on $\Sigma\times\Sigma$.  

A \dt{full} relational BiKAT is one where the BiKAT and its underlying KAT are full relational models. 
The example of boundedly nondeterministic relations is not full;
both the underlying KAT and the BiKAT lack a top.

For programs $c,d$ (interpreted as relations on $\Sigma$) and relations $R,S$ on $\Sigma$ meant to be pre- and post-conditions,
the $\forall\forall$ judgment \graybox{$c\sep d:\rspec{R}{S}$} is pictured
\begin{equation}\label{eq:allall}
\begin{diagram}[h=4ex,w=2.5em]
        & \sigma  &\rMapsto^c &\tau      &                                 & \tau      \\
\Bforall& \dRel^R &           &          &\mbox{\quad\Large$\implies$\quad}& \dRel_S \\
        & \sigma' &\rMapsto^d &\tau'     &                                 & \tau'     \\
\end{diagram}
\end{equation}
This is meant to say 
$\all{\sigma,\sigma',\tau,\tau'}{ \sigma R \sigma' \land \sigma c \tau \land \sigma' d \tau' \imp \tau S \tau'}$. By definitions we have the following.

\begin{lemma}[\popladded{adequacy for relational models}]\label{lem:adequacy}
In a relational BiKAT, for any $c,d$,
the BiKAT term $\emb{c}{d}$ denotes the set of all
pairs of $c,d$ computations.\footnote{To be precise,
the set of all $((\sigma,\sigma'),(\tau,\tau'))$ where
$(\sigma,\tau)\in c$ and $(\sigma',\tau')\in d$.}
\end{lemma}

\popladded{
Now we can confirm that defining the judgment $c\sep d:\rspec{R}{S}$ 
as the BiKAT equation (\ref{eq:allallbikat})
captures the semantic condition of (\ref{eq:allall}).
Moreover, BiKAT equality preserves adequacy.
}

\popladded{
\begin{theorem}[$\forall\forall$ soundness for relational models]\label{thm:allall}
Suppose that $c,d$ are elements of the underlying KAT of a relational BiKAT,
and suppose the sub-identity relations $\dot{R},\dot{S}$ that represent $R,S$ are 
among the BiKAT's tests.   Then
\begin{itemize}{}{}
\item[(a)] The condition (\ref{eq:allall}) holds iff 
$\, \dot{R};\emb{c}{d};\neg\dot{S}= 0$.
\item[(b)] For any $B$, if $\, \dot{R};\emb{c}{d}\leq \dot{R};B$
and $\, \dot{R};B;\neg\dot{S}= 0$ then (\ref{eq:allall}) holds.
\end{itemize}
\end{theorem}
}
\begin{proof}
\popladded{
For (a) the proof is by mutual implication.
If (\ref{eq:allall}) holds then 
$\dot{R};\emb{c}{d} \leq \dot{R};\emb{c}{d};\dot{S}$ by definitions.
And
$\dot{R};\emb{c}{d} \leq \dot{R};\emb{c}{d};\dot{S}$
is equivalent to $\dot{R};\emb{c}{d};\neg\dot{S} = 0$ 
as a fact about KATs.
For the converse, assume
$\dot{R};\emb{c}{d} \leq \dot{R};\emb{c}{d};\dot{S}$.
To show (\ref{eq:allall}), for any states $\sigma,\sigma',\tau,\tau'$
that satisfy the antecedent in (\ref{eq:allall}),
by adequacy Lemma~\ref{lem:adequacy} the executions are in 
$\dot{R};\emb{c}{d}$, 
so by assumption they are in 
$\dot{R};\emb{c}{d};\dot{S}$.
So by definitions the post states are related by $S$.
}

\popladded{For Part (b), 
if $\dot{R};\emb{c}{d}\leq \dot{R};B$ then
we have that $\dot{R};B;\neg\dot{S} \leq 0$ implies 
$\dot{R};\emb{c}{d};\neg\dot{S} \leq 0$ by KAT reasoning.
This yields (\ref{eq:allall}) using Part (a).
}
\end{proof}

\popladded{
Part (b) of Theorem~\ref{thm:allall} says that 
adequacy of $B$ for proving 
$c\sep d:\rspec{R}{S}$ is expressed by the 
equation $\dot{R};\emb{c}{d} \leq \dot{R};B$ 
(called ``$R$-adequacy'' in~\citet{NagasamudramN21}).
An adequacy proof can be interwoven with the correctness proof, in a form like
$\dot{R};\emb{c}{d};\neg\dot{S} \leq \ldots \leq \dot{R};B;\neg\dot{S} = \ldots = 0$.
Some of our examples have this form, using equality not $\leq$.
}

A \dt{trace BiKAT} over trace model $(\kA,\kB,\kplus,\kdot,\kstar,\kneg,\kone,\kzero)$ 
is $(\bA,\bB,\bplus,\bdot,\bstar,\bneg,\bone,\bzero)$ 
where elements of $\bA$ are sets of pairs of traces.  
To define $\bdot$, the coalesced catenation $\diamond$ is lifted to trace pairs and used as in ordinary trace models.
That is, $A\bdot B = \{ (t \diamond u, t' \diamond u')  \mid (t,t')\in A, (u,u') \in B, 
              \mbox{$t\diamond u$ and $t'\diamond u'$ defined} \}
$.
Sum and star are union and iterated $\bdot$.
For trace set $a$, the left embedding $\eml{a}$ is $\{(t,\sigma) \mid t\in a \land \sigma\in\Sigma\}$.
For set $A$ of trace pairs, the left projection $\lproj A$ is
$\{ t \mid (t,u)\in A \}$, i.e., $\lproj$ maps $\fst$ over $A$.  
As with relational BiKATs, we require that $\bA$ contains all the images of $\eml{\_}$ and $\emr{\_}$
on $\kA$, and $\kA$ contains the images of $\lproj$ and $\rproj$ on $\bA$.

To see why LRC holds in a trace BiKAT, let us write 
$\sDom t$ for the first state of trace $t$ and $t \sCod$ for the last.
For trace sets $c,c'$, and any 
$t\in c$ and $t'\in c'$,
we have $(t,t')\in \eml{c}\bdot\emr{c'}$ 
iff $(t,\sDom t')\in\eml{c}$ and
$(\sCod t, t')\in\emr{c'}$.
Note that $(t,\sDom t')\diamond (\sCod t, t') = (t,t')$.
Similarly, we have $(t,t')\in \emr{c'}\bdot\eml{c}$.
The upshot is that 
$\eml{c}\bdot\emr{c'}= \emr{c'}\bdot\eml{c}$.

A \dt{full trace BiKAT} has all trace sets (relative to the given set of admissible traces), for the underlying KAT, and all pairs of traces for the BiKAT. Like full relational models, full trace models have a top.

\popladded{
Lemma~\ref{lem:adequacy} and Theorem~\ref{thm:allall} 
can be straightforwardly adapted to trace models.
}

\section{Using BiKAT for $\forall\forall$ Relational Reasoning} 
\label{sec:examples}

Having introduced BiKAT we now demonstrate how it can be used to 
algebraically derive alignments and verify  $\forall\forall$ properties 
of a variety of examples that necessitate different kinds of alignment.

\paragraph{Simple Example.}
The following two programs compute the sum of integers up to some $N$. 
\[\begin{array}{ll}
C_1 \eqdef & \kcode{i:=0}\ksemi\;
             \left( [\kcode{i} \leq \kcode{N}] \ksemi\; \kcode{x:=x+i} \ksemi\; \kcode{i:=i+1}  \right)^\kstar
             \ksemi\; \neg{[\kcode{i} \leq \kcode{N}]} \\

C_2 \eqdef & \kcode{i:=1}\ksemi\;
             \left( [\kcode{i} \leq \kcode{N}] \ksemi\; \kcode{x:=x+i} \ksemi\; \kcode{i:=i+1}  \right)^\kstar
             \ksemi\; \neg{[\kcode{i} \leq \kcode{N}]}
\end{array}\]
To prove the $\forall\forall$ relational judgment
$C_1\sep C_2 : \rspec{\eml{[\kcode{N}\geq 0]} [\kcode{x} \eqbi \kcode{x}] [\kcode{N} \eqbi \kcode{N}]}{[\kcode{x} \eqbi \kcode{x}]}$
without resorting to functional correctness, we start from 
$
\eml{[\kcode{N}\geq 0]} \kdot [\kcode{x} \eqbi \kcode{x}] \kdot [\kcode{N} \eqbi \kcode{N}] \kdot
\emb{C_1}{C_2} \kdot \kneg [\kcode{x} \eqbi \kcode{x}] = 0
$
and then manipulate the sequential composition $\emb{C_1}{C_2}$ into an alignment, where we can directly relate the programs' variables during loop iterations.
In this case, there is a simple alignment: 
$C_1$ does one	
more iteration than $C_2$, so we unroll its loop once before aligning the two loop
bodies in lockstep.  Then, a simple relational loop invariant,
$[\kcode{x} \eqbi \kcode{x}] [\kcode{i} \eqbi \kcode{i}]$, suffices to establish
equivalence.

Let us now see this step-by-step in the algebra of BiKAT.
Unrolling $C_1$'s loop once, we have that $\emb{C_1}{C_2}$ is equal to:
\[ \emb{\kcode{i:=0}\ksemi \left( 1 + C; C^\kstar\right)\ksemi \neg e }
       {\kcode{i:=1}\ksemi C^\kstar\ksemi \neg e}
\]
where $C$ is the program text 
$(e\ksemi\,\kcode{x:=x+i}\ksemi\kcode{i:=i+1})$
and $e$ is $[\kcode{i} \leq \kcode{N}]$.  Distributing, we get
\[ \emb{\kcode{i:=0}\ksemi\neg e + \kcode{i:=0}\ksemi C\ksemi C^\kstar\ksemi \neg e }
       {\kcode{i:=1}\ksemi C^\kstar\ksemi \neg e}
\]
For semantic reasons, $\left(\kcode{i:=0}\ksemi \neg e \right)$ is infeasible.
So we can assume hypothesis $\left(\kcode{i:=0}\ksemi \neg e \right) = 0$ 
which lets us eliminate that term.
For the other term, we calculate
\[\begin{array}{lll}
    & \emb{\kcode{i:=0}\ksemi C\ksemi C^\kstar\ksemi \neg e }
          {\kcode{i:=1}\ksemi C^\kstar\ksemi \neg e} \\
=   & \emb{\kcode{i:=0}}{\kcode{i:=1}}\ksemi
      \eml{C} \ksemi
      \emb{C^\kstar}{C^\kstar}\ksemi
      \emb{\neg e}{\neg e}
    & \mbox{embedding homomorphic} \\
=   & \eml{\kcode{i:=0}}\ksemi \emr{\kcode{i:=1}}\ksemi
      \eml{C} \ksemi \eml{C^\kstar}\ksemi \emr{C^\kstar}\ksemi
      \emb{\neg e}{\neg e}
    & \mbox{def two-argument embedding} \\
=   & \eml{\kcode{i:=0}}\ksemi \eml{C} \ksemi
      \emr{\kcode{i:=1}}\ksemi \eml{C^\kstar}\ksemi \emr{C^\kstar}\ksemi
      \emb{\neg e}{\neg e}
    & \mbox{LRC}\\
=   & \emb{\kcode{i:=0}\ksemi C}{\kcode{i:=1}}\ksemi
      \emb{C^\kstar}{C^\kstar}\ksemi
      \emb{\neg e}{\neg e}
    & \mbox{embedding homomorphic}
  \end{array}\]
Using the expansion lemma~(\ref{eq:expand}), $\emb{C^\kstar}{C^\kstar}$ can be
rewritten into
\[ \left( \emb{[\kcode{i} \leq \kcode{N}] \ksemi\; \kcode{x:=x+i} \ksemi\; \kcode{i:=i+1}}
              {[\kcode{i} \leq \kcode{N}] \ksemi\; \kcode{x:=x+i} \ksemi\; \kcode{i:=i+1}} \right)^\kstar
\]
(Eliding terms that cancel out once we introduce the loop invariant.)
By systematically using LRC along with the
homomorphism property of embeddings, we  align the loop bodies in
lockstep:
\[ \left( \emb{[\kcode{i} \leq \kcode{N}]}{[\kcode{i} \leq \kcode{N}]} \ksemi\;
          \emb{\kcode{x:=x+i}}{\kcode{x:=x+i}}\ksemi\;
          \emb{\kcode{i:=i+1}}{\kcode{i:=i+1}}
   \right)^\kstar
\]
We next add the assumption that $[\kcode{x} \eqbi \kcode{x}]$ initially, add relational loop invariant 
$[\kcode{i} \eqbi \kcode{i}][\kcode{x} \eqbi \kcode{x}]$, and conclude the post-relation $[\kcode{x} \eqbi \kcode{x}]$ must hold beyond the loop.

\paragraph{Double Square (Example~\ref{eg:two}, Sect.~\ref{sec:overview})}
Recall the KAT expressions of the two programs:
$$
\begin{array}{l}
  k_{D_1} \eqdef \kcode{y:=0} \ksemi \kcode{z:=2*x} \ksemi (\kcode{[z>0]} \ksemi \kcode{z:=z-1} \ksemi \kcode{y:=y+x})^\kstar \ksemi \kNeg{\kcode{[z>0]}}\\
  
  k_{D_2} \eqdef \kcode{y:=0} \ksemi \kcode{z:=x} \ksemi (\kcode{[z>0]} \ksemi \kcode{z:=z-1} \ksemi \kcode{y:=y+x})^\kstar \ksemi \kNeg{\kcode{[z>0]}} \ksemi \kcode{y:=2*y}
\end{array}
$$
Note we are using overline as alternate notation for negation. 

In the preceding example, we did a unary unfolding of one iteration on the left,
and then aligned the loops in lockstep.
For this example we choose to align two iterations on the left with one iteration on the right.
We first use KAT laws to rewrite $k_{D_1}$ as:
$$k_{D_1} = \kcode{y:=0} \ksemi \kcode{z:=2*x} \ksemi (\kcode{[z>0]} \ksemi b \ksemi \kcode{[z>0]} \ksemi b)^\kstar \ksemi (\kcode{[z>0]} \ksemi b \kplus 1) \ksemi \kNeg{\kcode{[z>0]}}$$ where $b \eqdef \kcode{z:=z-1} \ksemi \kcode{y:=y+x}$. We then align the two loops and apply the expansion law (\ref{eq:expand}) as follows:
\[\begin{array}[t]{lll}
	& \emb{(\kcode{[z>0]} \ksemi b \ksemi \kcode{[z>0]} \ksemi b)^\kstar}
	            {(\kcode{[z>0]} \ksemi b)^\kstar} &\\
= & \emb{\kcode{[z>0]} \ksemi b \ksemi \kcode{[z>0]} \ksemi b}{\kcode{[z>0]} \ksemi b}^\kstar \ksemi
(\emb{\kcode{[z>0]} \ksemi b \ksemi \kcode{[z>0]} \ksemi b}{\kNeg{\kcode{[z>0]}}}^\kstar 
\kplus \emb{\kNeg{\kcode{[z>0]}}}{\kcode{[z>0]} \ksemi b}^\kstar) & \\
\end{array}\]

Now it is easier to prove the relational loop invariant $\bdots{\mathcal{I}} \eqdef [y \eqbi 2y] \ksemi [z \eqbi 2z]$ by proving that it is preserved across the above expansion loops. In the derivation, we break the proof into sub-proofs over the embeddings of loop bodies:
\[
\bdots{\mathcal{I}} \ksemi 
\emb{(\kcode{[z>0]} \ksemi b \ksemi \kcode{[z>0]} \ksemi b)^\kstar}
{(\kcode{[z>0]} \ksemi b)^\kstar }
\ksemi \kneg \bdots{\mathcal{I}} = \kzero 
\begin{array}{ll}
  
  \Leftarrow &
  \left\{
  \begin{array}{l}
    \bdots{\mathcal{I}} \ksemi 
    \emb{\kcode{[z>0]} \ksemi b \ksemi \kcode{[z>0]} \ksemi b}
    {\kcode{[z>0]} \ksemi b} \ksemi \kneg \bdots{\mathcal{I}} = 0 \,\wedge\\
    \bdots{\mathcal{I}} \ksemi 
    \emb{\kcode{[z>0]} \ksemi b \ksemi \kcode{[z>0]} \ksemi b}{\kNeg{\kcode{[z>0]}}} \ksemi \kneg \bdots{\mathcal{I}} = 0 \,\wedge\\
    \bdots{\mathcal{I}} \ksemi 
    \emb{\kNeg{\kcode{[z>0]}}}{\kcode{[z>0]} \ksemi b} \ksemi \kneg \bdots{\mathcal{I}} = 0\\
  \end{array}\right.\\
\end{array}
\]

\paragraph{Loop Tiling (Example~\ref{eg:three}, Sect.~\ref{sec:overview})}
Here are the two programs $L_1, L_2$ as KAT expressions:
$$
\begin{array}{l}
  k_{L_1} \defeq \kcode{x:=0} \ksemi (\kcode{[x<N*M]} \ksemi \kcode{a[x]:=f(x)} \ksemi \kcode{x:=x+1})^\kstar \ksemi \kNeg{\kcode{x<N*M}} \\
  k_{L_2} \defeq \kcode{i:=0} \ksemi (\kcode{[i<N]} \ksemi \kcode{j:=0} \ksemi (\kcode{[j<M]} \ksemi \kcode{A[i,j]:=f(i*M+j)} \ksemi \kcode{j:=j+1})^\kstar \ksemi \kNeg{ \kcode{[j<M]}} \ksemi
  \kcode{i:=i+1})^\kstar \ksemi \kNeg{\kcode{[i<N]}}
\end{array}
$$
For the alignment, we then transform these KAT expressions into two equivalent KAT expressions which have the same structure.
$$
\begin{array}{l}
  k_{L_1} = \kcode{x:=0} \ksemi (\kcode{[x<N*M]} \ksemi b_1 \ksemi (\kcode{[x<N*M]} \ksemi \kcode{[x\%M!=0]} \ksemi b_1)^\kstar \,\ksemi
  \kNeg{\kcode{[x<N*M]} \ksemi \kcode{[x\%M!=0]}})^\kstar \ksemi \kNeg{\kcode{x<N*M}} \\
  k_{L_2} = \kcode{i:=0} \ksemi (\kcode{[i<N]} \ksemi \kcode{j:=0} \ksemi (\kcode{[j<M]} \ksemi b_2 \kplus \kNeg{\kcode{[j<M]}}) \ksemi (\kcode{[j<M]} \ksemi b_2)^\kstar \ksemi \kNeg{ \kcode{[j<M]}} \,\ksemi
  \kcode{i:=i+1})^\kstar \ksemi \kNeg{\kcode{[i<N]}}
\end{array}
$$
        where $b_1 \defeq \kcode{a[x]:=f(x)} \ksemi \kcode{x:=x+1}$ and $b_2 \defeq \kcode{A[i,j]:=f(i*M+j)} \ksemi \kcode{j:=j+1}$.
        We next prove the judgment $\bdots{P} \ksemi \emb{k_{L_1}}{k_{L_2}} \ksemi \kneg \bdots{Q} = \kzero$ with the pre-relation $\bdots{P} \defeq [\kcode{N}\eqbi\kcode{N}] \ksemi [\kcode{M}\eqbi\kcode{M}] \ksemi  \kcode{[N > 0]} \ksemi \kcode{[M > 0]}$ and the post-relation $\bdots{Q} \defeq [\mathcal{R}(\kcode{N} \times \kcode{M}, \kcode{N}, \kcode{M})]$, where 
$\mathcal{R}(x,i,j)$        
is a predicate\footnote{$\mathcal{R}(x, i, j) \defeq \forall l, r, c.~ 0 \leq l < x \wedge 0 \leq r < i \wedge (r < i - 1 \wedge 0 \leq c < \kcode{M} \vee r = i - 1 \wedge 0 \leq c < j) \wedge l = r \times \kcode{M} + c \imp [\kcode{a}[l] \beq \kcode{A}[r, c]]$} that says every element of the array $\kcode{a}$ up to the index $x$ is equal to its corresponding element of the two-dimensional array $\kcode{A}$ up to the index $i, j$.
To prove the above judgment, we use the relational invariant $\bdots{\mathcal{I}} \defeq \kcode{[i < N]} \ksemi \kcode{[j} \leq \kcode{M]} \ksemi \kcode{[x} \beq \kcode{i} \times \kcode{M + j]} \ksemi [\mathcal{R}(\kcode{x}, \kcode{i}, \kcode{j})]$ for the inner loops and the relational invariant $\bdots{\mathcal{J}} \defeq \kcode{[i} \leq \kcode{N]} \ksemi \kcode{[x} \beq \kcode{i} \times \kcode{M]} \ksemi [\mathcal{R}(\kcode{x}, \kcode{i}, 0)]$ for the outer loops.

\paragraph{Loop Summaries and Procedure Calls.}

\begin{wrapfigure}{r}{2.5in}
\begin{lstlisting}[deletekeywords={len}]
arrayInsert (A, len, h) {
  i := 0;
  while (i<len && A[i]<h) i++;
  len := shift_array(A, i, 1);
  A[i] := h;
  while (i<len) i++;
  return i;
}
\end{lstlisting}
\end{wrapfigure}

\poplchanged{KAT-based systems do not intrinsically support variables but instead reason about programs under hypotheses that axiomatize relevant facts about primitives.}{KAT is a propositional 
theory of imperative control struture.  In KAT-based systems, first order state variables, 
conditions, and assignments can be handled using hypotheses (typically, Hoare triples) that axiomatize their semantics.}

Here too, the KATs embedded in a BiKAT are parametric over alphabets of actions and tests and hypotheses about those tests. 
Consequently BiKAT  inherently supports reasoning at coarser or finer granularies. 
Moreover, this parameterization allows BiKAT to be used in concert with 
other procedures (e.g.~loop summarization, procedure specs/summaries, ghost states, prophecy variables, etc.) which could be applied beforehand and then incorporated into a BiKAT through the primitive actions and hypothesis.  For example, consider the array insertion 
procedure\footnote{\popladded{The second loop sets \kcode{i} to \kcode{len} in a way that avoids a timing channel but we are not modeling timing here.}}~\cite{ShemerGSV19,Goyal2021} shown to the right that will illustrate the use of externally-provided unary procedure summaries.
The goal is to prove that,
$\kcode{arrayInsert}\sep \kcode{arrayInsert} : \rspec{\mkEqBi{len} \mkEqBi{A}}{\mkEqBi{i}}$
i.e., the noninterference property that there is no leak on \kcode{h}.

Contrary to some examples above, this example does not require loop bodies to be lock-step aligned. It does, however, require alignment between intermediate points between the loops. Specifically, the programs must be aligned in three places:  after both have completed their first loop, after both have completed the call to \kcode{shift\_array} (incorporating that method's post-condition), and after both of completed their second loop. That is, the BiKAT alignment:
\lstset{literate={'}{{'}}1}
\[\begin{array}{l}
\mkEqBi{len};
\embRepeat{ \kcode{while(...)i++} };
\mkEqBi{len};
\embRepeat{ [\kcode{i<=len}] };\\
\;\;\;\;
\embRepeat{ \kcode{len := shift\_array(...)} };
\mkEqBi{len+1};\\
\;\;\;\;
\eml{ \kcode{while(i<len)i++}; [\kcode{i=len}] };
\emr{ \kcode{while(i<len)i++}; [\kcode{i=len}] };
\mkEqBi{i}
\end{array}\]
We first use the pre-relation $\mkEqBi{len}$ and the postcondition of the first loop on both the left and right sides, aligning when both sides have completed
to show that $\mkEqBi{len}$ still holds, while
$\kcode{i}\leq \kcode{len}$ (on the left) and
$\kcode{i}'\leq \kcode{len}'$ (on the right).

At this point we have established the alignment necessary for this example. Completing the proof requires small semantic hypotheses (similar to those in earlier examples) and some strategy for establishing $\mkEqBi{len}$ is preserved across the procedure call to \kcode{shift\_array}.
One option is to introduce relational hypothesis 
$\kcode{shift\_array(A,i,j)}\sep \kcode{shift\_array(A,i,j)} : \rspec{\mkEqBi{len}\mkEqBi{j}}{\mkEqBi{len}}$ which ensures agreement on \kcode{len} after the embedded procedures.
Alternatively,  one could employ a unary specifications of the form
$
\{  \}
\kcode{shift\_array(A,i,j)} 
\{ \kcode{len=old(len)+j} \}
$ through the use of KAT hypotheses of the form $p C \neg q = 0$ embedded on the left and the right. These unary post-conditions can be combined with the pre-call $\mkEqBi{len}$ to add a post-call bitest $\mkEqBi{len}$. Here there are some details that would be needed (\emph{e.g.}~ghost variables) to support post-condition tests that relate variables to pre-conditions.
Finally, we use $\mkEqBi{len}$ with the post-conditions of the last loops
that 
$\kcode{i}=\kcode{len}$  on both sides to conclude that $\mkEqBi{i}$.

\section{Relational Hoare logic in BiKAT}\label{sec:RHL}

In this section we show that relational Hoare logic rules can be 
derived in any BiKAT.
Relational logics involve two programs, thus quadruples,
sometimes written $\{ P \} c \sim c' \{ Q \}$ for commands $c,c'$.  
\citet{Benton:popl04} writes $c\sim c' : P \Rightarrow Q$.
We consider inference rules for the $\forall\forall$ judgment 
form $c \sep c' : \rspec{P}{Q}$ introduced in Sect.~\ref{sec:intro}
and expressed in any BiKAT by the equation (\ref{eq:allallbikat}).

\paragraph{Deriving Rules of RHL}

A number of publications have presented variations on relational Hoare logic.
We consider a number of basic rules that can be found in Benton's influential paper~\cite{Benton:popl04}
and in Francez' less known paper~\cite{Francez83},
and a number of subsequent works.
There is not yet a standard set of rules, in part because 
until recently there was no satisfactory notion of completeness~\cite{NaumannISOLA20,NagasamudramN21}.
We consider a number of representative rules in Fig.~\ref{fig:RHLselected}.

In the rules we use suggestive syntax for formulas and programs,
and we will not belabor the distinction between program syntax and its 
standard representation in KAT.
We lift boolean expression $e$ to a relation formula 
$\leftF{e}$ that says $e$ is true in the left state,
so its representation as a test in BiKAT will look the same.

\begin{figure*}[t]
\begin{small}
\begin{mathpar}

\inferrule*[left=dSeq]{
  c\sep c' : \rspec{P}{R} \\
  d\sep d' : \rspec{R}{Q}
}{
  c;d \Sep c';d' : \rspec{P}{Q}
}

\inferrule*[left=dIf]{
  P \imp \eqbib{e}{e'} \\
  c\sep c' : \rspec{P\land \leftF{e}\land\rightF{e'}}{Q} \\
  d\sep d' : \rspec{P\land \neg\leftF{e}\land\neg\rightF{e'}}{Q} 
}{
  \ifc{e}{c}{d} \Sep \ifc{e'}{c'}{d'} : \rspec{P}{Q}
}

\inferrule*[left=dWh]{
  P \imp \eqbib{e}{e'} \\
  c\sep c' : \rspec{P\land \leftF{e}\land\rightF{e'}}{P} 
}{ 
  \whilec{e}{c} \Sep \whilec{e'}{c'} : \rspec{P}{P\land \neg\leftF{e}\land\neg\rightF{e'}}
}

\inferrule*[left=rDisj]{
  c\sep d : \rspec{P}{Q} \\
  c\sep d : \rspec{R}{Q} 
}{
  c\sep d : \rspec{P\lor R}{Q} \\
}

\inferrule*[left=SeqSk]{
  c \sep\skipc : \rspec{P}{R} \\
  d \sep \skipc : \rspec{R}{Q} 
}{
  c;d \sep \skipc : \rspec{P}{Q}
}

\inferrule*[left=rConseq]{
  R\imp P \\
  c\sep d : \rspec{P}{Q} \\
  Q \imp S 
}{
  c\sep d : \rspec{R}{S} \\
}

\end{mathpar}
\end{small}
\vspace*{-3ex}
\caption{Selected inference rules of $\forall\forall$ logic}
\label{fig:RHLselected}
\end{figure*}

We will show that all the rules are sound in any BiKAT. 
Recall that we interpret the judgment $c \sep c' : \rspec{P}{Q}$ 
as the BiKAT equation $P;\emb{c}{c'};\neg Q = 0$.
In the soundness proofs we use the equivalent form $P ; \emb{c}{c'}  \leq P ; \emb{c}{c'} ; Q$.  This 
form is also used in Kozen's work deriving Hoare logic rules in KAT~\cite{Kozen00}.

Several rules in Fig.~\ref{fig:RHLselected}
infer ``diagonal'' judgments relating two same-structured programs,
e.g., \rn{dSeq},  \rn{dIf}, and \rn{dWh}. 
The latter two cater for alignment whereby the same control path is followed,
with a requirement of agreement on conditional tests.
There is a rule \rn{dIf4}, named after the number of its premises, does not require such agreement.

Some of the rules can be derived from others.
Regarding rule \rn{SeqSk}, from its premises one can use \rn{dSeq} to obtain
$c;d \sep \skipc;\skipc : \rspec{P}{Q}$.
But the inference rules provide no way to replace $\skipc;\skipc$ by the 
equivalent $\skipc$.
One of the benefits of working in KAT is free use of such equivalences,
including more interesting ones like  
\poplchanged{$\ifc{e}{c}{c} = c$.}{loop unrolling in the Tiling example.}

Rule \rn{dIf} can be derived from \rn{dIf4} using rule \rn{FalsePre} and \rn{rConseq},
but we prove \rn{dIf} directly.  
The side condition $P \imp \eqbib{e}{e'}$ is most directly expressed 
as $P \leq \eqbib{e}{e'}$.
The left-right equality  $\eqbib{e}{e'}$ is equivalent to $\emb{e}{e'}+\emb{\neg e}{\neg e'}$
so the side condition yields
\begin{equation}\label{eq:dWhSide}
   P; \emb{e}{\neg e'} = 0
\qquad
   P; \emb{\neg e}{e'} = 0
\end{equation}
To prove soundness of \rn{dIf} we calculate:
\[\begin{array}{lll}
  &  P; \emb{e;c+\neg e; d}{e';c'+\neg e';d'} \\
= &  P;\emb{e}{e'};\emb{c}{c'} + P;\emb{e}{\neg e'};\emb{c}{d'} & \mbox{emb homo, distrib} \\
  & + \; P;\emb{\neg e}{e'};\emb{d}{c'} + P;\emb{\neg e}{\neg e'};\emb{d}{d'}  \\

= &  P;\emb{e}{e'};\emb{c}{c'} + P;\emb{\neg e}{\neg e'};\emb{d}{d'} & \mbox{using (\ref{eq:dWhSide})} \\

\leq &  P;\emb{e}{e'};\emb{c}{c'};Q + P;\emb{\neg e}{\neg e'};\emb{d}{d'};Q & \mbox{premises of \rn{dIf}} \\

= & P;\emb{e;c+\neg e; d}{e';c'+\neg e';d'};Q & \mbox{reverse steps}
\end{array}\]

To prove \rn{dWh}, first observe
\[\begin{array}{lll}
  & P ; (\emb{e;c}{\neg e'}^* + \emb{\neg e}{e';c'}^*) \\
= 
& P ;(1+ \emb{e;c}{\neg e'};\emb{e;c}{\neg e'}^*) + 
    P ;(1+ \emb{\neg e}{\neg e';c'};\emb{\neg e}{e';c'}^*) & \hint{distrib, star unfold} \\
= & P + P;\emb{e;c}{\neg e'};\emb{e;c}{\neg e'}^* + 
    P + P;\emb{\neg e}{\neg e';c'};\emb{\neg e}{e';c'}^* & \hint{distrib} \\
=  & P &\hint{using (\ref{eq:dWhSide})}
\end{array}\]
The last step uses that embedding is homomorphic, and the side condition (\ref{eq:dWhSide})
whence $P+0+P+0=P$.
Using the premise 
$P ; \emb{e}{e'}; \emb{c}{c'} \leq P ; \emb{e}{e'};\emb{c}{c'};P $,
we get the conclusion of \rn{dWh} by
\[\begin{array}{lll}
     & P ; \emb{(e;c)^*;\neg e}{(e';c')^*;\neg e'} \\ 
=    & P ; \emb{(e;c)^*}{(e';c')^*}; \emb{\neg e}{\neg e'} & \hint{emb homo} \\
=    & P ; \emb{e;c}{e';c'}^*;(\emb{e;c}{\neg e'}^* + \emb{\neg e}{e';c'}^*) ;\emb{\neg e}{\neg e'}     & \hint{expansion (\ref{eq:expand})} \\
\leq & P ; \emb{e;c}{e';c'}^*;P; (\emb{e;c}{\neg e'}^* + \emb{\neg e}{e';c'}^*) ;\emb{\neg e}{\neg e'} & \hint{premise, invariance}\\
=   & P ; \emb{e;c}{e';c'}^*;P;\emb{\neg e}{\neg e'}  & \hint{observation above} \\
= & P ; \emb{e;c}{e';c'}^*;P;\emb{\neg e}{\neg e'};P 
        & \hint{tests idem, tests commute} \\
=  & P ; \emb{e;c}{e';c'}^*;P; (\emb{e;c}{\neg e'}^* + \emb{\neg e}{e';c'}^*) ;\emb{\neg e}{\neg e'};P & \hint{observation, in reverse} \\
\leq & P ; \emb{e;c}{e';c'}^*; (\emb{e;c}{\neg e'}^* + \emb{\neg e}{e';c'}^*) ;\emb{\neg e}{\neg e'};P & \hint{$P\leq 1$ since $P$ is a test}\\
=    
&  P ; \emb{(e;c)^*;\neg e}{(e';c')^*;\neg e'}; P;\emb{\neg e}{\neg e'} 
&\hint{reverse steps using  (\ref{eq:expand})} 
\end{array}
\]

\begin{theorem}\label{thm:RHL}
The rules of Fig.~\ref{fig:RHLselected} 
%(also appendix Fig.~\ref{fig:RHLadditional})  DN: it's true but not important to emphasize here 
are sound in any BiKAT.  
\end{theorem}

To cater for reasoning about related loops where data-dependent alignment is needed,
the following rule \rn{caWh} for \emph{conditionally aligned loops} 
has been shown sound \popladded{for specific models} in prior work~\cite{Beringer11,NagasamudramN21,BNNN19a}.
The rule features relations $Q$ (resp.~$R$) as conditions under which an iteration
on one side is aligned with doing nothing on the other side.
\[\inferrule*[left=caWh]{
  c\sep c' : \rspec{P\land \leftF{e}\land\rightF{e'}\land\neg Q\land\neg R}{P} \\
  c\sep\skipc : \rspec{P\land Q\land \leftF{e}}{P} \\
  \skipc\sep c' : \rspec{P\land R\land \rightF{e'}}{P} \\
  P \imp \eqbib{e}{e'} \lor (Q\land\leftF{e}) \lor (R\land\rightF{e'})  
}{ 
  \whilec{e}{c} \Sep \whilec{e'}{c'} : \rspec{P}{P\land \neg\leftF{e}\land\neg\rightF{e'}}
}\]
The rule can be used together with the one-side rules like \rn{SeqSk}.
To prove \rn{caWh} we use this expansion law:
\begin{equation}\label{eq:caWhexpand}
\begin{array}{lcl}
  \eml{e;c}^*;\eml{\neg e} ; \emr{e';c'}^*;\emr{\neg e'}   
  & = & ( Q;\eml{e;c} + R;\emr{e';c'} + \neg Q;\neg R;\emb{e;c}{e';c'} \\
  &&  \qquad + \neg Q;\emb{e;c}{\neg e'} + \neg R;\emb{\neg e}{e';c'} )^* ; \emb{\neg e}{\neg e'}
\end{array}
\end{equation}
\poplchanged{At the time of submission we have not found a derivation 
proving this law in all BiKATs.
We do, however, have a proof that (\ref{eq:caWhexpand}) holds in any 
*-continuous BiKAT, which includes relational nd trace BiKATs.
Rule \rn{caWh} is sound in any BiKAT that satisfies (\ref{eq:caWhexpand}).
(See appendix of extended version~\cite{BiKATarxiv}.)}{%
We do not know whether (\ref{eq:caWhexpand}) holds in all BiKATs,
but it holds in relational BiKATs and trace BiKATs.
\begin{theorem}\label{thm:caWh}
\begin{list}{}{}
\item[(a)]
Law (\ref{eq:caWhexpand}) holds in any *-continuous BiKAT.
\item[(b)] 
Rule \rn{caWh} is sound in any BiKAT that satisfies (\ref{eq:caWhexpand}).
\end{list}
\end{theorem}
}

\begin{example}\label{eg:caWh}
Consider the following, adapted from~\citet{NaumannISOLA20}.
\[\begin{array}{ll}
    P_1: & \kcode{y:=x; z:=24; w:=0; while y>4 do if w\%2=0 then z:=z*y; y:=y-1 fi; w:=w+1 od} \\
    P_2: & \kcode{y:=x; z:=16; w:=0; while y>4 do if w\%3=0 then z:=z*2; y:=y-1 fi; w:=w+1 od}
  \end{array} \]
For $\kcode{x} \geq 4$, $P_1$ computes $\kcode{x}!$ in $\kcode{z}$ and $P_2$,
$2^{\kcode{x}}$ in $\kcode{z}$.  We want to show that $P_1$ majorizes $P_2$,
i.e., $P_1\sep P_2: \rspec{\kcode{x} \eqbi \kcode{x}'}{\kcode{z} >
\kcode{z}'}$ (using primed variables to refer to those in $P_2$).  Notice that
both programs take gratuitous steps, making it difficult to reason by a simple
lockstep alignment of the two loops---the relational loop invariant would
become needlessly complicated.  Verification can be simplified by using the
following data-dependent alignment: if $\eml{\kcode{w\%2}\neq\kcode{0}}$,
perform a left-only iteration; if $\emr{\kcode{w\%3}\neq\kcode{0}}$, perform a
right-only iteration; otherwise, if
$\emb{\kcode{w\%2}=\kcode{0}}{\kcode{w\%3}=\kcode{0}}$, execute the loop
bodies jointly.  Then, $\kcode{y} \eqbi \kcode{y}' \land \kcode{z} >
\kcode{z}' > 0$ is invariant and sufficient to establish the postrelation. 
This reasoning is done using rule \rn{caWh} in~\citet{NagasamudramN21}.
An alternative is to reason in BiKAT with the following alignment:
\[ \begin{array}{ll}
  & \emb{\kcode{y:=x}}{\kcode{y:=x}};
    \emb{\kcode{z:=24}}{\kcode{z:=16}};
    \emb{\kcode{w:=0}}{\kcode{w:=0}}; \\
  & \big( \; \eml{[\kcode{w\%2}\neq\kcode{0}]; \kcode{w:=w+1}}
    
  + \emr{[\kcode{w\%3}\neq\kcode{0}]; \kcode{w:=w+1}} \\

  & + \, \emb{[\kcode{w\%2=0}]}{[\kcode{w\%3=0}]};
    \emb{\kcode{z:=z*y}; \kcode{y:=y-1}; \kcode{w:=w+1}}
    {\kcode{z:=z*2}; \kcode{y:=y-1}; \kcode{w:=w+1}}; \big)^\kstar;

    \emb{\neg [\kcode{y>4}]}{\neg [\kcode{y>4}]}.
  \end{array} \]
It can be derived by starting with
$[\kcode{x}\eqbi\kcode{x}'];\emb{P_1}{P_2}$, using (\ref{eq:caWhexpand})
with $Q:=\eml{\kcode{w\%2}\neq\kcode{0}}$ 
and $R:=\emr{\kcode{w\%2}\neq\kcode{0}}$, 
and then simplifying, relying on the loop invariant which cancels 
the $\neg Q$ and $\neg R$ cases in (\ref{eq:caWhexpand}).
\end{example}

\paragraph{Self-Composition Rule}

Researchers have repeatedly discovered that relational correctness can be encoded 
in (unary) Hoare logic, essentially because a pair of states can be represented by a single state,
e.g., using renamed variables~\cite{Francez83,BartheDArgenioRezk}.
A state relation $P$ can be expressed by a state predicate $\hat{P}$,
and a command $c$ can be renamed to $\hat{c}$ acting on the alternate variables.
Then the judgment $ c \sep d: \rspec{P}{Q}$ is represented by the Hoare triple
$c;\hat{d}:\spec{\hat{P}}{\hat{Q}}$.  (We write $c:\spec{p}{q}$ for $\{p\}c\{q\}$.)
A complete relational Hoare logic thereby comprises 
the single rule,
``from $c;\hat{d}:\spec{\hat{P}}{\hat{Q}}$ infer $c\sep d:\rspec{P}{Q}$'',
together with a complete (unary) Hoare logic.   
In terms of alignment, of course, this is the most degenerate form of reasoning.
The notion of alignment completeness explains the need for other rules,
in terms of alignment of automata; see~\citet{NagasamudramN21}.
\poplremoved{In some sense, BiKAT can be seen as self-composition together 
with the LRC condition which facilitates 
the use of better alignments.}

\section{Beyond 2-Safety: Properties and Logics}\label{sec:beyond}

Many relational requirements can be expressed as instances of the $\forall\forall$ (2-safety) form
depicted in (\ref{eq:allall}), 
or other conditions that must hold for all pairs of behaviors.
Some other frameworks, such as HyperLTL~\cite{ClarksonFKMRS14},
can express properties with other patterns of quantification.
In Sect.~\ref{sec:simu} we consider, two $\forall\exists$ patterns 
based on pre- and post-relations the way (\ref{eq:allall}) is.  
In Sect.~\ref{sec:exists} we consider $\exists\forall$ and $\exists\exists$ properties,
for the sake of systematic exploration.
The discussion focuses on relational models for concreteness,
though the goal is model-independent algebraic formulations.

\subsection{$\forall\exists$ Properties}\label{sec:simu}

For programs that may be nondeterministic, 
the $\forall\forall$ property depicted in (\ref{eq:allall})
is often too strong.
A range of interesting requirements 
such as possibilistic noninterference and data refinement are 
expressed in a $\forall\exists$ form.
We call it \dt{forward simulation} and 
write \graybox{$c\sep d:\aespec{R}{S}$} for the following.
%(where $c,d,R,S$ are relations on some set $\Sigma$).  obvious by now 
\begin{equation}\label{eq:fsim}
\begin{diagram}[h=4ex,w=3em]
        & \sigma  &\rMapsto^c & \tau &         &       &           &\tau\\
\Bforall& \dRel^R &           &      & \Bexists&       &           &\dRel_S\\
        & \sigma' &           &      &         &\sigma'&\rMapsto^d &\tau'\\
\end{diagram}
\end{equation}
Here $\sigma,\sigma',\tau,\tau'$ are states and $c,d,R,S$ are relations, so (\ref{eq:fsim}) says
\[ %begin{equation}\label{eq:fsimFmla}
\all{\sigma,\sigma',\tau}{ \sigma R \sigma' \land \sigma c \tau 
\imp \some{\tau'}{ \sigma' d \tau' \land \tau S \tau'}}
\]
An equivalent ``point-free'' formulation, using relation algebra, is $R^o;c \subseteq d;S^o$.
(Recall \popladded{from Sect.~\ref{sec:models}} that $R^o$ means the converse of $R$.)

For possibilistic noninterference, $c=d$ and the relations express low indistinguishability:
$R$ expresses agreement on low inputs and $S$ on low outputs.
In the case of data refinement, 
$R=S$ and $R$ captures some change of data representation.

For program refinement one sometimes needs the similar $\forall\exists$ property called \dt{backward simulation}, 
written \graybox{$c\sep d:\bespec{R}{S}$}.
\begin{equation}\label{eq:bsim}
\begin{diagram}[h=4ex,w=3em]
        & \sigma  &\rMapsto^c &\tau      &       & \sigma    &            & \\
\Bforall&         &           &\dRel_S &\Bexists& \dRel^R &           & \\
        &         &           &\tau'     &        & \sigma'   &\rMapsto^d & \tau' \\
\end{diagram}
\end{equation}
This is expressed in relation algebra as $c;S \subseteq R;d$,
and pointwise as 
\[ %begin{equation}\label{eq:bsimFmla}
\all{\sigma,\tau,\tau'}{ \sigma c \tau \land \tau S \tau' 
  \imp \some{\sigma'}{ \sigma R \sigma' \land \sigma' d \tau' }}
\] 

\popladded{In Sect.~\ref{sec:biKATsim} we give theorems that 
characterize the forward and backward simulation properties
in terms of existence of BiKAT witnesses.
Using those theorems, we derive (Sect.~\ref{sec:forwardBack}) rules for inferring 
judgments $c\sep d:\aespec{R}{S}$
and $c\sep d:\bespec{R}{S}$.
In Sect.~\ref{sec:trikat} we show how forward and backward simulation can be expressed 
in closed form, by generalizing BiKAT to a kind of 3-KAT.
}

\popladded{As an aside, we note that the simulation properties subsume unary underapproximation.  For unary tests $p,q$, the incorrectness logic~\cite{OHearn2019}
judgment 
``every state in $q$ can be reached by a terminating execution of $c$ from some state in $p$'' 
is equivalent to $\emr{c}:\bespec{\emr{p}}{\emr{q}}$.
The forward approximation condition 
``for every state in $p$ there is a terminating execution of $c$ that ends in $q$'' is equivalent to $\emr{c}:\aespec{\emr{p}}{\emr{q}}$.
}

\subsection{BiKAT Characterizations of Simulation}\label{sec:biKATsim}

Alignment is well known to play a role in verifying $\forall\exists$ properties.  
Example~\ref{eg:four} in Sect.~\ref{sec:overview}  
illustrates that one    
aligns the computations in a convenient way, in order to
winnow out execution pairs in which the second execution makes 
undesirable nondeterministic choices.
Bitests can serve as assume statements for this purpose.

\begin{theorem}[forward witness soundness]\label{thm:fsim}
In a relational BiKAT over a KAT with top, we have forward simulation $c|d:\aespec{R}{S}$ if there is some BiKAT term $W$, 
called \dt{alignment witness}, 
that is \dt{f-valid}, which means:
\[ \begin{array}{lll}
\mbox{(WC)} & \dot{R}; W \leq \dot{R}; W; \dot{S}  &\mbox{(witness $\forall\forall$ correct)} \\
\mbox{(WO)} & \dot{R}; \eml{c} \leq W ; \emr{\hav} &\mbox{(witness overapproximates $c$)} \\
\mbox{(WU)} & \dot{R}; W \leq \emb{\hav}{d}        &\mbox{(witness underapproximates $d$)}
\end{array}\]
\end{theorem}
\begin{proof}
To prove $c\sep d:\aespec{R}{S}$, suppose $\sigma R \sigma'$ and $\sigma c \tau$.
Thus $(\sigma,\sigma')\dot{R};\eml{c}(\tau,\sigma')$.
By (WO) there is $\tau'$ with $(\sigma,\sigma')W(\tau,\tau')\emr{\hav}(\tau,\sigma')$.
Then by (WU) we have $\sigma' d \tau'$ and by (WC) we have $\tau S \tau'$.
\end{proof}

\medskip
The theorem is not simply a reduction to  $\forall\forall$-logic; it relies essentially on the use of inequalities (WO) and (WU) that are not $\forall\forall$ conditions
of the form (\ref{eq:allall}).
It is no suprise that reasoning about the existential in a $\forall\exists$ property involves finding a witness---that is familiar 
in many settings.  Here a witness comprises an alignment of the programs, with embedded bitests that serve to select
the witnessing executions ---and all the correctness conditions are expressed equationally!

\begin{theorem}[backward witness soundness]\label{thm:bsim}
In a relational BiKAT over a KAT with top,  we have $c\sep d:\bespec{R}{S}$ 
if there is alignment witness  $W$ that is \dt{b-valid}, meaning:
\[\begin{array}{lll}
\mbox{(WCb)} & W;\dot{S} \leq \dot{R};W;\dot{S} &  \mbox{(witness reverse $\forall\forall$ correct)}\\
\mbox{(WOb)} & \eml{c};\dot{S} \leq \emr{\hav};W & \mbox{(witness overapproximates $c$)}\\
\mbox{(WUb)} & W;\dot{S} \leq \emb{\hav}{d} & \mbox{(witness underapproximates $d$)}
\end{array}\]
\end{theorem}
\begin{proof}
To prove $c\sep d:\bespec{R}{S}$, suppose $\sigma c \tau$ and $\tau S \tau'$.
By (WOb) there is $\sigma'$ with $(\sigma,\sigma')W(\tau,\tau')$.
By (WUb) we have $\sigma' d \tau'$.  
By (WCb) we have $\sigma R \sigma'$.
\end{proof}

\begin{theorem}[witness completeness]\label{thm:sim:complete}
%In a full relational BiKAT,
In a relational BiKAT over a KAT with top,
a forward (resp.\ backward) simulation judgment holds 
if and only if it has an f-valid (resp.\ b-valid) alignment  witness.\footnote{\label{fn:sim:complete}In the published 
version of this paper, the theorem requires the model to be full,
and the proof only shows that witnesses exist in the model,
without showing that they are expressible as BiKAT terms.
(Note that the theorem implicitly assumes the pre/post relations are expressible as bitests.)
%To prove completeness, one defines a witness relation;
%restricting to full models ensures that this relation is in the model (as is $\hav$).
For the record here is the original proof of completeness for forward simulation; 
it constructs witness elements that are in some sense minimal.

Suppose $c|d:\aespec{R}{S}$ holds.  
Using suggestive identifiers for bound variables, 
define the predicate
$\mathcal{P}(\sigma,\sigma',\tau) \eqdef \sigma R \sigma' \land \sigma c \tau$
and the set 
$\mathcal{X}(\sigma,\sigma',\tau) \eqdef \{ \tau' \mid \sigma' d \tau' \land \tau S \tau' \}$.
For any $(\sigma,\sigma',\tau)$ that satisfy $\mathcal{P}$ we have
$\mathcal{X}(\sigma,\sigma',\tau) \neq \emptyset$, owing to $c\sep d:\aespec{R}{S}$. 
So define 
$\mathcal{Y}(\sigma,\sigma',\tau)$ to be a chosen element of $\mathcal{X}(\sigma,\sigma',\tau)$ if $\mathcal{P}(\sigma,\sigma',\tau)$, and 
undefined otherwise.
Define 
\( W \eqdef \{((\sigma,\sigma'),(\tau,\tau')) \mid 
                   \mathcal{P}(\sigma,\sigma',\tau) \land 
                   \tau' = \mathcal{Y}(\sigma,\sigma',\tau) \} \).
By fullness, $W$ is in the BiKAT.
We have (WC) because $W = \dot{R} ;  W ; \dot{S}$.
We have (WU) using the definition of $W$. 
We have (WO) also by definitions:
If $(\sigma,\sigma')\dot{R};\eml{c}(\tau,\sigma')$
then $\mathcal{P}(\sigma,\sigma',\tau)$ so let $\tau'=\mathcal{Y}(\sigma,\sigma',\tau)$;
we get $(\sigma,\sigma')W(\tau,\tau')$ so
$(\sigma,\sigma')W;\emr{\hav}(\tau,\sigma')$.
} %footnote
\end{theorem}
\begin{proof}
To prove completeness for forward simulation, suppose $c|d:\aespec{R}{S}$ holds.  
Let witness $W$ be $\dot{R};\emb{c}{d};\dot{S}$.
We show $W$ is f-valid.
We have $\dot{R};W \leq \dot{R};W;\dot{S}$ because
$\dot{R};W 
= 
\dot{R};\dot{R};\emb{c}{d};\dot{S} 
= 
\dot{R};\dot{R};\emb{c}{d};\dot{S};\dot{S}$
by definition of $W$ and idempotence of tests.
To show $\dot{R};\eml{c} \leq W;\emr{\hav}$,
consider any $\sigma, \sigma', \tau, \tau'$.
By definitions, $(\sigma,\sigma') \dot{R};\eml{c}(\tau,\tau')$ is equivalent to 
$\sigma R \sigma' \land \sigma c \tau \land \tau'=\sigma'$.
So by $c|d:\aespec{R}{S}$ there is $\tau''$ with $\sigma' d \tau''$ and $\tau S \tau''$,
whence by definition of $W$ we have  $(\sigma,\sigma') W (\tau,\tau'')$.
Thus $(\sigma,\sigma') W ; \emr{\hav} (\tau,\tau')$ by definitions.
To show $\dot{R};W \leq \emb{\hav}{d}$, consider any $\sigma, \sigma', \tau, \tau'$.
By definitions, $(\sigma,\sigma') \dot{R}; W (\tau,\tau')$ iff
$\sigma R \sigma' \land \sigma c \tau \land \sigma' d \tau' \land \tau S \tau'$.
So $(\sigma,\sigma') \eml{\hav} (\tau,\sigma')$ and $(\tau,\sigma') \emr{d} (\tau,\tau')$,
hence $(\sigma,\sigma') \eml{\hav};\emr{d} (\tau,\tau')$,
i.e., $(\sigma,\sigma') \emb{\hav}{d} (\tau,\tau')$.

The proof for backward simulation is similar, again using 
$\dot{R};\emb{c}{d};\dot{S}$ as witness.
\end{proof}

For clarity we defined the judgment forms $c\sep d:\aespec{R}{S}$ and 
$c\sep d:\bespec{R}{S}$ for relational models,
but it is straightforward to interpret them in trace models.
Essentially the horizontal arrows in (\ref{eq:fsim}) and (\ref{eq:bsim})
are interpreted as sequences of zero or more steps.
Then the soundness Theorems~\ref{thm:fsim} and~\ref{thm:bsim} extend 
to trace models and we get a witness completeness theorem for full trace models,
by an argument like our proof of Theorem~\ref{thm:sim:complete}.

Theorem~\ref{thm:sim:complete} is about existence of a witness that satisfies the 
witness conditions in a model.  
Next we consider finding witnesses for which the conditions 
can be proved equationally.

\subsection{Examples Proving $\forall\exists$ with Witnesses}

Example~\ref{eg:four} in Sect.~\ref{sec:overview}     
considers the programs $E_1: \kcode{x:=any; y:=x}$ and $E_2: \kcode{t:=any; z:=t+1}$.
To establish 
$E_1\sep E_2 : \aespec{\kcode{true}}{\kcode{y} \eqbi \kcode{z}}$
we choose witness  % rephrased to save space 
$W\eqdef\emb{\kcode{x:=any}}{\kcode{t:=any}}; [\kcode{x-1} \eqbi \kcode{t}];\emb{\kcode{y:=x}}{\kcode{z:=t+1}}$.
The bitest $[\kcode{x-1} \eqbi \kcode{t}]$ winnows execution pairs 
so choices made by $E_2$  match favorably the nondeterministic assignment made by $E_1$.
Condition (WO) ensures that all executions of $E_1$ are still covered.
The three conditions are proved using 
axioms to express semantics of the primitives, 
e.g., the bitest $[\kcode{x-1} \eqbi \kcode{t}]$ commutes with $\eml{\kcode{y:=x}}$.
We also use a condition 
% $1 \leq \emr{\kcode{t:=any}};[x-1 \eqbi t];\emr{\kcode{t:=any}}$ 
which expresses the left-totality of the bitest $[\kcode{x-1} \eqbi \kcode{t}]$ as discussed later in connection with rule \rn{enAss}.

\begin{example}[$\forall\exists$ path alignment]\label{eg:paths}
When considering $\forall\exists$ properties, sometimes the choices made by the witness determine which paths are taken in the program, rather than merely values taken for variables. 
Consider this example adapted from~\citet{BeutnerF22}.
\begin{lstlisting}
a1 := any; a2 := any; if (h > l) o := l + a1;             // $k_1$
                      else { x := a2; if(x > l) o := x;   // $k_2$
                                      else      o := l; } // $k_3$ 
\end{lstlisting}
Here we are interested in the possibilistic non-interference property
$C\sep C : \aespec{\mkEqBi{l}}{\mkEqBi{o}}$.
%\added{ While such properties are beyond the scope of this paper, we do consider $\forall\exists$ so the following is a $\forall^1\exists^1$ interpretation of their Fig.~1 example.}
%For possibilistic non-interference example, the pre-relation  $R = \mkEqBi{l}$ and the post-relation $S = \mkEqBi{o}$.
This example has multiple cases to consider depending on (i) how the inputs to an execution impact the conditional and (ii) how those choices may differ from one execution to another.
Consequently, for any path taken in the left program, our choices for the \kcode{any}s in the second may involve taking different paths than were taken in the first program.
The following is the KAT representation of the three paths and the whole program
($k=k_1 + k_2 + k_3$):
\[
\begin{array}{lcl}
k_1 & \eqdef & \kcode{a1:=any; a2:=any; [h>l]; o:=l+a1}\\
k_2 & \eqdef & \kcode{a1:=any; a2:=any; [h<=l]; x:=a2; [x>l]; o:=x}\\
k_3 & \eqdef & \kcode{a1:=any; a2:=any; [h<=l]; x:=a2; [x<=l]; o:=l}\\
\end{array}
\]
We will refer to the right program as $k'$, having primed variables.
We use 
the following witness
\[
\begin{array}{lcll}
W &\eqdef&  \emb{ \kcode{a1:=any; a2:=any;} }{ \kcode{a1':=any; a2':=any;} };\\
	& & (\;\; [\kcode{h>l; h'>l'; a1'=a1}]
	&  \;+\; [\kcode{h>l; h'<=l'; a2'>l'; a2'=l+a1}]\\
	& & \;+\; [\kcode{h<=l; h'>l'; a2>l; a1'=a2-l'}]
	&  \;+\; [\kcode{h<=l; h'>l'; a2<=l; a1'=0}]\\
	& & \;+\; [\kcode{h<=l; h'<=l'; a2>l; a2'>l'; a2'=a2}]
	&  \;+\; [\kcode{h<=l; h'<=l'; a2<=l; a2'<=l'}]) \;;\; \emb{c}{c'}
%	& & \emb{c}{c'}
\end{array}
\]
where $c$ and $c'$ are the remainders of $k$ and $k'$, respectively,
after the nondeterministic choices.
(So $c \equiv \kcode{([h>l]; o:=l+a1 + [h<=l]; x:=a2; ([x>l]; o:=x + [x<=l]; o:=l;))}$.) 
The witness $W$ comprises six cases covering \emph{all} preconditions of the input \kcode{l, h} and \kcode{l', h'} of the two programs. The witness can be rewritten into $W = W_1 + W_2 + \ldots + W_6$ where each $W_i$ corresponds to a case in $W$. For example, the witness $W_2$ corresponding to the precondition $\kcode{h} > \kcode{l} \wedge \kcode{h'} \leq \kcode{l'}$ can be simplified into
\[
\begin{array}{lcl}
	W_2 & \eqdef & \emb{ \kcode{a1:=any; a2:=any;} }{ \kcode{a1':=any; a2':=any;} };\\
	& & [\kcode{h>l; h'<=l'; a2'>l'; a2'=l+a1}]; \emb{c}{c'}\\
	& = & \emb{ \kcode{a1:=any; a2:=any;} }{ \kcode{a1':=any; a2':=any;} };\\
	& & [\kcode{h>l; h'<=l'; a2'>l'; a2'=l+a1}];\\
	& & \emb{\kcode{[h>l]; o:=l+a1}}{\kcode{[h'<=l']; x':=a2'; ([x'>l']; o':=x' + [x'<=l']; o':=l';)}}\\
	& = & \emb{ \kcode{a1:=any; a2:=any;} }{ \kcode{a1':=any; a2':=any;} };\\
	& & [\kcode{h>l; h'<=l'; a2'>l'; a2'=l+a1}];
	     \emb{\kcode{o:=l+a1}}{\kcode{x':=a2'; o':=x'}} \quad\text{(distrib, cancel)}
\end{array}
\]
In the above $W_2$, infeasible paths in the left and right program under the precondition $\kcode{h} > \kcode{l} \wedge \kcode{h'} \leq \kcode{l'}$ and the chooser $\kcode{a2'} > \kcode{l'}$ have been pruned out. Under this precondition, path $k_1$ in the left program can be aligned with path $k'_2$ and path $k'_3$  in the right program. However, the alignment between the left $k_1$ and the right $k'_3$ is invalid w.r.t the $\forall\exists$ property because it requires that the left \kcode{a1} must be always 0, which is infeasible since under the $\forall$ quantifier, we have to consider every execution of the left program. The condition $\kcode{a2'} > \kcode{l'}$ in $W_2$ then chooses the right $k'_2$ to align with the left $k_1$ and the condition $\kcode{a2'} = \kcode{l + a1}$ shows that there exists an execution under that alignment to achieve agreement on \kcode{o}. With that intuition, the (WC) condition for the witness $W_2$, 
and all the conditions of Theorem~\ref{thm:fsim} for the witness terms, are straightforward to prove.
\end{example}

\begin{example}[Backward simulation]\label{eg:bsimEx}
Consider the following, where $\kcode{x}$, $\kcode{t}$, $\kcode{s}$, and
$\kcode{z}$ range over the natural numbers:
\[\begin{array}{ll}
    C_1: & \kcode{while x>n do x:=x-1 od; t:=any+x; z:=x+t} \\
    C_2: & \kcode{s:=any; while x>n do x:=x-1 od; z:=x+s}
  \end{array}
\]
We want to show their possibilistic equivalence, which could be expressed as 
$C_1\sep C_2:\aespec{R}{S}$, where
$R \eqdef \kcode{x}\eqbi\kcode{x} \land \kcode{n}\eqbi\kcode{n}$ and
$S \eqdef \kcode{z}\eqbi\kcode{z}$.
To prove this it would be convenient to align the two loops, to enable use of 
simple relational invariants $\kcode{x}\eqbi \kcode{x}$ etc.  But then the nondeterministic assignment to \kcode{t}
is aligned far after the assignment to \kcode{s} that needs to match it.  
There are two well known ways to deal with such situations:
introduce a prophecy variable~\cite{AbadiLamport88} or (equivalently an auxiliary variable~\cite{MorganAux}) use backward simulation.  
We can prove the following:\footnote{A stronger postcondition is needed for this backwards
property, for similar reasons to what happens in incorrectness logic~\cite{OHearn2019}.} 
$C_1\sep C_2:\bespec{R}{T}$, where
$T \eqdef \kcode{x}\eqbi\kcode{x} \land \kcode{n}\eqbi\kcode{n}
\land \kcode{z}\eqbi\kcode{z} \land \kcode{t}\eqbi\kcode{s}$.
To prove this using the witness technique of Theorem~\ref{thm:bsim}, we choose
the witness to be:
\[ W \eqdef
  \emr{\kcode{s:=any}}; \emb{X}{X}^\kstar;
  \eml{\neg [\kcode{x>n}]}; \eml{\kcode{t:=any+x}}; \emb{\kcode{z:=x+t}}{\kcode{z:=x+s}}; T
\]
where $X \eqdef [\kcode{x>n}]; \kcode{x:=x-1}$.  Notice that
$W$ ends with the postrelation $T$ and we do not need to introduce additional bitests.
The three conditions the witness must satisfy are easily proved.
As usual we rely on axioms for the semantics of primitives.
Interestingly, these include ones about backward preservation of bitests, e.g.,
$\emb{X}{X}; [\kcode{x}\eqbi\kcode{x}] =
[\kcode{x}\eqbi\kcode{x}]; \emb{X}{X}; [\kcode{x}\eqbi\kcode{x}]$.
\end{example}

% DN changed C_1' a.k.a. C_1^\prime to C_3 to avoid prime on the left 
\begin{example}[Forward simulation and prophecy]\label{eg:fsimProphecyEx}
As a variation on Example~\ref{eg:bsimEx} we can prove $C_3\sep C_2:\aespec{R}{S}$ 
for a modified version of $C_1$ that uses variable \kcode{p} to ``prophesize'' the value for \kcode{t}.
\[ \begin{array}{ll}
    C_3: & \kcode{p:=any; while x>n do x:=x-1 od; t:=p+x; z:=x+t} \\
    \end{array}
\]
To show $C_3\sep C_2:\aespec{R}{S}$ using the witness technique
in Theorem~\ref{thm:fsim}, choose witness $W$ to be:
\[
  W \eqdef
  \emb{\kcode{p:=any}}{\kcode{s:=any}}; B;
  \emb{X}{X}^\kstar;
  \eml{\neg [\kcode{x>n}]}; \emb{\kcode{t:=p+x; z:=x+t}}{\kcode{z:=x+s}}
\]
where $X \eqdef [\kcode{x>n}]; \kcode{x:=x-1}$ and 
$B \eqdef [\kcode{p+min(x,n)}\eqbi s]$.
As in the previous forward simulation examples,
the witness aligns the two nondeterministic assignments together,
the loops in lockstep, and introduces a bitest $B$ that filters executions of
$C_2$ to only those that ensure agreement on $\kcode{z}$ upon
termination.
The three witness conditions for $W$ are straightforward to check.  

In the next section we introduce deductive rules
for forward simulation.  With these, the judgment about $C_3$ and $C_2$
is proved in a way that implicitly follows the alignment $W$.
The chooser bitest $B$ is introduced as a postcondition, by a rule for aligned nondeterministic assignments, and it is manipulated in intermediate assertions rather than being inlined like it is in $W$.  
\ifarxiv
(See Appendix~\ref{sec:fsudr}.)
\else
(See the appendix of the extended version~\cite{BiKATarxiv}.)
\fi
\end{example}

\subsection{Logics of Forward and Backward Simulation}\label{sec:forwardBack}

\subsubsection{Forward Simulation Logic}

\newcommand{\anyfig}{\mbox{any}}

\begin{figure*}[t]
\begin{small}
\begin{mathpar}

\inferrule*[left=eWhL]{
  P \imp (\rightF{e'} \imp \leftF{e}) \\
  c\sep c' : \aespec{P \land \leftF{e} \land \rightF{e'}}{P} \\
  c\sep \skipc : \aespec{P \land \leftF{e}}{P}
}{\whilec{e}{c} \Sep \whilec{e'}{c'} : \aespec{P}{P \land \neg\leftF{e} \land
\neg\rightF{e'}}
}

\inferrule[eHav]{ 
        \bone \leq \emr{\hav};\dot{R};\emr{\hav} 
}
{ \emb{\hav}{\hav} : \aespec{true}{R} } 

% ALERT: not using \kcode{any} in the following, owing to format mess

\inferrule[enAss]{ 
        \bone \leq \emr{y:=\anyfig};\dot{R};\emr{y:=\anyfig}
}
{  \emb{x:=\anyfig}{y:=\anyfig} : \aespec{true}{R} }

\inferrule[eDisj]{
  c\sep d : \aespec{P}{Q} \\
  c\sep d : \aespec{R}{Q} 
}{
  c\sep d : \aespec{P\lor R}{Q} \\
}

\end{mathpar}
\vspace*{-3ex}  
\end{small}
\caption{Selected rules for $\forall\exists$ forward simulation correctness.
\textbf{Please note:} there are also rules \rn{eConseq}, \rn{eSeq}, \rn{eIf}, \rn{eWh},
which look the same as \rn{rConseq}, \rn{dSeq}, \rn{dIf}, \rn{dWh}
but using $\aespec{-}{-}$.}
\label{fig:aeRHL}
\end{figure*}

Theorem~\ref{thm:fsim} gives a way to prove forward simulation judgments,
by direct reasoning in a BiKAT.  
There are also inference rules for forward simulation.
In fact several of the inference rules for $\forall\forall$ 
judgments (Fig.~\ref{fig:RHLselected}) are also sound for forward simulation.
Fig.~\ref{fig:aeRHL} gives some rules for the $\aespec{-}{-}$ judgment.

Consider \rn{eWh}, which is simply \rn{dWh} but for the $\aespec{-}{-}$ judgment.
Informally, \rn{eWh} is sound because any terminating execution of the left 
program can be matched by one that terminates on the right, owing to the 
side condition that says the loop tests agree.
Rule \rn{eWhL} has a similarly simple side condition that suffices to ensure relative termination: in terms of alignment, if the right loop can continue to iterate then so can the left, and their joint iterations can be aligned in lockstep.
An additional premise handles the situation where the left loop has more iterations than the right.

One can consider two other situations.  One is where the iterations can be aligned in lockstep,
but the right loop may need more iterations.  Consider this rule:
\[\inferrule*[left=eWhR]{
  P \imp (\neg \leftF{e} \lor \rightF{e'}) \land \rightF{f \geq 0} \\
  \skipc\sep c' : \aespec{P \land \rightF{e' \land f = n}}{P \land \rightF{f < n}} \\
  c\sep c' : \aespec{P \land \leftF{e} \land \rightF{e'}}{P}
}{\whilec{e}{c} \Sep \whilec{e'}{c'} : \aespec{P}{P \land \neg\leftF{e} \land
\neg\rightF{e'}}
}\]
It uses a variant expression $f$ to establish termination on the right side
(like in total correctness Hoare logic).  
A related but different idea is the general rule to infer $c\sep c' : \aespec{P}{R}$
from $c\sep c' : \rspec{P}{R}$ together with termination of $c'$
from states in the codomain of $P$. 
But KAT does not support direct expression of termination, and 
we refrain from formulating the requisite notations for a BiKAT encoding 
of these rules.

The other situation for loops is where lockstep alignment is not sufficient.
We conjecture that a rule similar to rule \rn{caWh} can be devised,
but that is beyond the scope of this paper.

To relate two nondeterministic assignments, this axiom is sound in relational and trace models:
\( \emb{x:=\kcode{any}}{y:=\kcode{any}} : \aespec{(\all{\Left{x}}{\some{\Right{y}}{R}})}{R} \).
It uses suggestive informal notation for quantification over the left and right states.
The precondition ensures that for any value assigned to $x$ there is some value for 
$y$ such that $R$ holds.
But in this paper we refrain from formalizing formulas for relations.
Instead we consider a rule for judgments of the form $ \emb{x:=\kcode{any}}{y:=\kcode{any}} : \aespec{true}{R} $.
This holds provided that, in any pair of states, for every value of $x$ there is some value for $y$ making $R$ true.  As a step towards an algebraic formulation for that condition, first consider the fully nondeterministic action $\hav$.

For $\emb{\hav}{\hav} : \aespec{true}{\dot{R}}$ to hold, $R$ must be a domain-total relation.  In terms of relations this can be expressed by the equation $R;\hav=\hav$ 
but we prefer to express the condition in terms of the bitest $\dot{R}$ for $R$.
The condition in rule \rn{eHav} (Fig.~\ref{fig:aeRHL}), i.e., BiKAT equation 
$\bone \leq \emr{\hav};\dot{R};\emr{\hav}$,
holds (in a relational model) just if the relation $R$ is domain-total. 

For two nondeterministic assignments to satisfy the judgment
$\emb{x:=\kcode{any}}{y:=\kcode{any}} : \aespec{true}{R}$, the condition that we wrote as 
$\all{\Left{x}}{\some{\Right{y}}{R}}$
should be valid.  This is equivalent 
to the condition 
$\bone \leq \emr{y:=\kcode{any}};\dot{R};\emr{y:=\kcode{any}}$
of rule \rn{enAss}.  
An informal reading is that for any pair of states,
there is some value for $y$ on the right that makes $R$ hold.
(The trailing assignment to $y$ can restore the initial value of $y$.)

\begin{lemma}\label{lem:fsim:rules}
In any BiKAT, and for any rule in Fig.~\ref{fig:aeRHL},
given an f-valid witnesses for the premises, 
there is an f-valid witness for the conclusion.
\end{lemma}

\begin{theorem}\label{thm:fsim:rules}
The rules in Fig.~\ref{fig:aeRHL} are sound in any 
relational BiKAT over a KAT with top.
\end{theorem}
\begin{proof}
For each rule, if its premises are true then 
by Theorem~\ref{thm:sim:complete} there are f-valid witnesses for the premises.
So by Lemma~\ref{lem:fsim:rules} we obtain an f-valid witness for the conclusion,
so by Theorem~\ref{thm:fsim} the conclusion is true.
\end{proof}

We consider illustrative cases in the proof of Lemma~\ref{lem:fsim:rules}.

\paragraph{Proof of \rn{eHav}}

The rule has no premise judgment, only the antecedent 
$\bone \leq \emr{\hav};\dot{R};\emr{\hav}$.
To prove the conclusion we take witness $W$ to be $\emb{\hav}{\hav};\dot{R}$.
Using the BiKAT notation $\bone$ for the pre-relation $true$, 
the f-validity conditions are:
\\
\begin{tabular}{ll}
(WC) & $\bone;\emb{\hav}{\hav};\dot{R} \leq \bone; \emb{\hav}{\hav};\dot{R};\dot{R}$ \\
(WU) & $\bone;\emb{\hav}{\hav};\dot{R} \leq \emb{\hav}{\hav}$ \\
(WO) & $\bone;\eml{\hav} \leq \emb{\hav}{\hav};\dot{R};\emr{\hav}$
\end{tabular}
\\
We have (WC) by idempotence of the test $\dot{R}$,
and (WU) using $\dot{R}\leq\bone$.
The condition (WO), expressing existence, is proved 
using the antecedent condition for $\dot{R}$.
$\bone;\eml{\hav} 
= \bone;\eml{\hav};\bone 
\leq \bone;\eml{\hav};\emr{\hav};\dot{R};\emr{\hav}
= \emb{\hav}{\hav};\dot{R};\emr{\hav}
$.
For rule \rn{enAss} the proof is similar.
% OK but not important:
% The interesting case is (WO) which looks like
% $\bone; \eml{x:=\kcode{any}} = \bone; \eml{x:=\kcode{any}}; \bone
% \leq \bone; \eml{x:=\kcode{any}}; \emr{y:=\kcode{any}};\dot{R};\emr{y:=\kcode{any}}
% = \emb{x:=\kcode{any}}{y:=\kcode{any}};\dot{R};\emr{y:=\kcode{any}}
% \leq \emb{x:=\kcode{any}}{y:=\kcode{any}};\dot{R};\emr{\hav} $.

\paragraph{Proof of \rn{eSeq}}

Suppose for premise $c\sep c' : \aespec{P}{R}$ we have witness $Z$ and 
for $d\sep d' : \aespec{R}{Q}$ we have witness $W$, so the conditions are 

\begin{tabular}{llll}
(WCZ) & $\dot{P}; Z \leq \dot{P};Z; \dot{R}$ & 
(WCW) & $\dot{R}; W \leq \dot{R}; W; \dot{Q}$ \\

(WUZ) & $\dot{P}; Z \leq \emb{\hav}{c'}$ &
(WUW) & $\dot{R}; W \leq \emb{\hav}{d'}$ \\

(WOZ) & $\dot{P}; \eml{c} \leq Z ; \emr{\hav}$ &
(WOW) & $\dot{R}; \eml{d} \leq W ; \emr{\hav}$ 
\end{tabular}

\noindent
To prove $c;d \Sep c';d' : \aespec{P}{Q}$ we use $Z;W$ as witness.
\\
$\bullet$ (WC) To show $\dot{P};Z;W \leq \dot{P};Z;W;\dot{Q}$ we have
$ \dot{P};Z;W \leq \dot{P};Z;\dot{R};W \leq \dot{P};Z;\dot{R};W;\dot{Q} \leq \dot{P};Z;W;\dot{Q}$ using (WCZ), (WCW), and $\dot{R}\leq\bone$. %test below 1.
\\

\noindent $\bullet$ (WU) We have 
$\dot{P}; Z; W \leq 
\dot{P};Z;\dot{R};W \leq 
\emb{\hav}{c'};\emb{\hav}{d'} = 
\emb{\hav}{c';d'}$
using (WCZ), (WUZ), (WUW), embedding homomorphic, and idempotence of $\hav$.
\\
$\bullet$ (WO)
\(\begin{array}[t]{lll}
    & \dot{P};\eml{c;d} \\
=   & \dot{P};\eml{c};\eml{d} & \mbox{emb homo}\\
\leq& \dot{P};Z;\emr{\hav};\eml{d} & \mbox{(WOZ)} \\
=   & \dot{P};Z;\eml{d};\emr{\hav} & \mbox{(LRC)} \\
=   & \dot{P};Z;\dot{R};\eml{d};\emr{\hav} & \mbox{(WCZ)} \\
\leq& \dot{P};Z;W;\emr{\hav};\emr{\hav} & \mbox{(WOW)} \\
=   & \dot{P};Z;W;\emr{\hav} & \mbox{$\hav$ idem, emb homo}\\
\leq   & Z;W;\emr{\hav} & \mbox{$P\leq\bone$} %test below 1}
\end{array}\)

\subsubsection{Backward Simulation Logic}

The backward simulation judgment has a number of inference rules,
with some interesting differences from the $\forall\forall$ and forward simulation rules.
The rules are derivable in BiKAT.
\poplremoved{i.e., we have  results like Lemma~\ref{lem:fsim:rules} and
Theorem~\ref{thm:fsim:rules} but to save space we just sketch the highlights
and omit the formal statements.}

Hoare's assignment axiom based on weakest precondition, using substitution in the precondition,
is often called ``backwards''.  
The relational generalization works for $\forall\forall$ and forward simulation:
\begin{mathpar}
\inferrule[dAss]{}{ 
v:=e \sep v':=e' : \rspec{\subst{P}{v|v'}{e|e'}}{P}
}

\inferrule[eAss]{}{ 
v:=e \sep v':=e' : \aespec{\subst{P}{v|v'}{e|e'}}{P}
}
\end{mathpar}
Variables and substitution are not part of KAT/BiKAT but the rules are sound in models.
But a judgment of this form is unsound for backward simulation,
for the same reason as in incorrectness logic~\cite{OHearn2019}: 
with an arbitrary postcondition, there can be final states that are not in the image of the assignment.
The unary~\citet{Floyd67} axiom for assignment is 
$ v:=e : \spec{q}{\some{u}{\subst{q}{v}{u} \land v = \subst{e}{v}{u} }}$ 
which suggests the following:
\[ 
   v:=e \sep v':=e' \; : \; \bespec{P\; }{\; \some{u,u'}{\subst{P}{v|v'}{u|u'} 
                                   \land \leftF{v = \subst{e}{v}{u}} 
                                   \land \rightF{v' = \subst{e'}{v'}{u'}}}}
\]
Consider $\tau,\tau'$ that satisfy the postcondition
and let $\sigma$ be the left initial state.  Let $\hat{u},\hat{u}'$ be values
that witness the existential and observe that $\tau$ is $\sigma$ with
$v$ updated to $\hat{u}$.
Choose $\sigma'$ be a state such that $\tau'$ is $\sigma'$ with $v$ updated to $\hat{u'}$.
Conclude $\sigma,\sigma'$ satisfy the
precondition, and further, that executing $v':=e'$ in $\sigma'$ yields
$\tau'$; hence the judgment holds.

Unlike the simple law that can be used to express semantics of assignments
in a proof of \rn{eAss} for forward simulation, 
a BiKAT formulation of the Floyd rule would need a more complicated way to reason about assignment 
and postcondition formula.  
It can be done by choosing some type, say integers, for data, and then 
treating the existential as an integer-indexed sum of tests,
but we leave this to the reader.  

Another resemblance to incorrectness logic is that the consequence rule is reversed from the one for forward simulation 
and $\forall\forall$.
Suppose $W$ is a witness for
$c\sep d : \bespec{P}{Q}$,
and moreover $P\imp R$ and $S\imp Q$.
Then  $W$ is also a witness for 
$ c|d : \bespec{R}{S}$
To show (WCb) for the latter, we have
$W;\dot{S} \leq 
 W;\dot{Q} \leq
 \dot{P};W;\dot{Q} \leq
 \dot{R};W;\dot{Q}$
(using (WCd) for $W$, i.e., $W;Q \leq P;W;Q$).
And $W;\dot{S} \leq  \dot{R};W;\dot{Q}$
iff $W;\dot{S} \leq  \dot{R};W;\dot{S}$ using $S\imp Q$.

\begin{figure*}
\begin{small}
\begin{mathpar}

\popladded{
\inferrule*[left=bIf]{
  c\sep c' : \bespec{P\land \leftF{e}\land\rightF{e'}}{Q} \\
  d\sep d' : \bespec{P\land \neg\leftF{e}\land\neg\rightF{e'}}{Q} 
}{
  \ifc{e}{c}{d} \Sep \ifc{e'}{c'}{d'} : \bespec{P}{Q}
}
}

\inferrule*[left=bWh]{
  c\sep c' : \bespec{P\land \leftF{e}\land\rightF{e'}}{P} 
}{\whilec{e}{c} \Sep \whilec{e'}{c'} : \bespec{P}{P\land \neg\leftF{e}\land\neg\rightF{e'}}}

\popladded{
\inferrule*[left=bSeq]{
  c\sep c' : \bespec{P}{R} \\
  d\sep d' : \bespec{R}{Q}
}{
  c;d \Sep c';d' : \bespec{P}{Q}
}
}

\popladded{
\inferrule*[left=bnAss]{ 
   \bone \leq \emr{y:=\anyfig};\dot{R};\emr{y:=\anyfig}
}
{  \emb{x:=\anyfig}{y:=\anyfig} : \bespec{R}{true} }
}

\popladded{
\inferrule*[left=bConseq]{
  R\imp P \\
  c\sep d : \bespec{R}{S} \\
  Q \imp S 
}{
  c\sep d : \bespec{P}{Q} \\
}
}

\inferrule*[left=bDisj]{
  c\sep d : \bespec{P}{Q} \\
  c\sep d : \bespec{P}{R} 
}{
  c\sep d : \bespec{P}{Q\lor R} \\
}

\end{mathpar}
\end{small}
\vspace*{-3ex}
\caption{Selected rules for $\forall\exists$ backward simulation correctness.}
\label{fig:beRHL}
\end{figure*}

The rule for sequence looks just like \rn{dSeq} (and \rn{eSeq}),
and is proved by a calculation similar to the one proving \rn{eSeq}.
For conditional and loop, the pattern of rules \rn{dIf} and \rn{dWh} (also \rn{eIf} and \rn{eWh})
can be retained but the side conditions are not needed!
The conclusion of rule \rn{bWh} in Fig.~\ref{fig:beRHL}
says that given final states $(\tau,\tau')$ related by $P$ in which the loop
tests are false, and a terminating execution from some initial $\sigma$ with
$(\sigma,\sigma')$ related by $P$,
there is an execution from $\sigma'$ ending in $\tau'$.
Informally, the conclusion follows because we can repeatedly invoke the premise, starting from 
the last iteration, to obtain the requisite right execution.
If $\sigma_i$ is the state reached after the $i$th iteration on the left,
and $\tau_i$ the corresponding right state given by our induction hypothesis,
and there is a preceding iteration on the left from $\sigma_{i-1}$, 
the premise yields some $\tau_{i-1}$ with a matching right iteration,
and moreover $e'$ holds in $\tau_{i-1}$ so the loop does take this iteration.

\begin{theorem}\label{thm:bsim:rules}
The rules in Fig.~\ref{fig:beRHL} \poplremoved{(and appendix)} are sound in any 
relational BiKAT over a KAT with top.
\end{theorem}
The proof uses Theorem~\ref{thm:sim:complete}.  
Thorough investigation of loop rules for backward simulation is beyond the scope of this paper.

\subsection{On $\exists\forall$ and $\exists\exists$ Properties}\label{sec:exists}

We have looked at various forms of properties involving 2 executions, with pre/post-conditions playing a different roles. So far we have focused on some $\forall\forall$ properties (Sect.~\ref{sec:RHL}), and some $\forall\exists$ ones (Sects.~\ref{sec:simu}--\ref{sec:forwardBack}). Many of the relational properties can be expressed in these forms, but the classes that correspond to the duals of those properties are also interesting. Consider the property

\begin{equation}\label{eq:ea}
\begin{diagram}[h=4ex,w=2.4em]
        & \sigma  &\rMapsto^c & \tau &         &       &           &     &                                 & \tau \\
\BexistsX& \dRel^R &           &      & \BforallX&       &           &     &\mbox{\quad\Large$\imp$\quad}& \dRel_S \\
        & \sigma' &           &      &         &\sigma'&\rMapsto^d &\tau'&                                 & \tau' 
\end{diagram}
\end{equation}
that says $\exists \sigma, \sigma', \tau.\ \sigma R\sigma'\land \sigma c\tau\land \forall \tau' (\sigma' d \tau'\imp \tau S\tau')$.
%\timos{Give name to this property to better refer to it.}
This property is the dual of forward simulation with a negated postcondition, or in other words 
(\ref{eq:ea}) is equivalent to $\neg (c\sep d : \aespec{R}{\neg S})$. 
%% \[
%% \begin{array}{lll}
%% & \neg (\some{\sigma,\sigma',\tau}{ 
%%                \sigma R \sigma' \land \sigma x \tau \land 
%%                \all{\tau'}{ (\sigma' y \tau' \imp \tau S \tau') }} ) \\
%% \iff & \all{\sigma,\sigma',\tau}{ 
%%                \sigma R \sigma' \land \sigma x \tau \imp
%%                \neg \all{\tau'}{ (\sigma' y \tau' \imp \tau S \tau' ) }}  \\
%% \iff & \all{\sigma,\sigma',\tau}{ 
%%                \sigma R \sigma' \land \sigma x \tau \imp
%%                \some{\tau'}{ \sigma' y \tau' \land \neg( \tau S \tau' ) }}  \\
%% \iff & x|y : \aespec{R}{\neg S}
%% \end{array}
%% \]
Using Theorem~\ref{thm:sim:complete}, to prove this property, we can check whether the following holds: for any BiKAT term $Z$ 
at least one of the conditions (WO), (WU) or (WC) fails.

Although not many properties of the form $\exists\forall$ naturally occur in the literature on relational verification, they are nonetheless important, and not only as duals of the more frequently occurring $\forall\exists$ ones. The reason is that with the former, it suffices to find one single execution of the left program that captures the required behavior against as traces of the right one. Checking whether such an execution or trace exists is often hard, and a proof system of $\exists\forall$ properties will be of value. 
%% \timos{Remove the following?:}Consider the case where the cost associated with an execution of a program can be described by a tuple of values in possibly independent dimensions (such as memory resources and time ones). Checking whether there is an execution of a given program that uses minimal resources compared to any other execution is often helpful.

Arguing similarly we can use BiKAT reasoning to explore $\exists\exists$ properties on two executions. A prime example of such properties is ``definite  non-determinism'', or in other words the existence of an input, and two executions on that input that produce different output. Consider
$$\some{\sigma,\sigma',\tau,\tau'}{ \sigma R \sigma' \land \sigma c \tau \land \sigma' d \tau' \land \neg\tau S \tau'}$$
as a general form of $\exists\exists$ properties.
This formulation is the dual of the main 2-safety property we study in the previous sections (see Equation~(\ref{eq:allall})). Non-determinism can then be expressed by setting both programs $c$ and $d$ to be the same, and setting $R$ and $S$ to be relations expressing agreement on all variables.

\subsection{TriKAT}\label{sec:trikat}

One might hope that the existential quantifications of (\ref{eq:fsim}) and (\ref{eq:bsim}) 
could be expressed algebraically using BiKAT projection (Def.~\ref{def:bikatproj}).
Unfortunately, projection existentially quantifies both initial and final state
\popladded{of the second execution}
(see (\ref{eq:BiKATproj})).
So projection does not directly capture (\ref{eq:fsim}),
\popladded{where the second execution's initial state is universally quantified,
nor (\ref{eq:bsim}) where its final state is universally quantified.}
\popladded{We sketch a way to use projection that merits further investigation but is not used in the rest of the paper.}

Possibilistic noninterference can be expressed in HyperLTL, a temporal logic with explicit quantifiers that range over
the traces of a fixed program~\cite{ClarksonFKMRS14}.
The formula says that for all traces $\pi,\pi'$ with initial states related by $R$,
there is a trace $\pi''$ with the same initial state as $\pi'$, 
and $R$ relates the final states of $\pi$ and $\pi''$.
The reason three traces are needed is that two initial states are universally quantified in (\ref{eq:fsim}).  

To express $c\sep d:\aespec{R}{S}$ we consider three executions, 
using in place of (\ref{eq:fsim}) the following pattern. 
\[
\begin{diagram}[h=4ex,w=3em]
        & \sigma  &\rMapsto^c& \tau  &         &\sigma' &\rMapsto^\hav&\tau\\
\Bforall& \dRel^R &          &\dRel_{id}         & \Bexists&\dRel_{id}&            &\dRel_S\\
        & \sigma' &\rMapsto^\hav&\tau&         &\sigma' &\rMapsto^d   &\tau'\\
\end{diagram}
\]
We define a notion of TriKAT so these ingredients can be described in the form 
of the following diagram. 
It depicts a relation on state triples represented by the displayed TriKAT term 
using notation to be explained.
\begin{equation}\label{eq:fsimAlt}
\begin{diagram}[h=4ex,w=3em]
        & \sigma    &\rMapsto^c   & \tau     \\       
        & \dRel^R   &             &\dRel_{id} \\      
        & \sigma'   &\rMapsto^\hav&\tau       &
\quad \begin{array}{l}
      \quad\mbox{represented by TriKAT term } \\[.5ex] 
      \quad\triEmb{\dot{R}}{\dot{id}} ; \tricom{c}{\hav}{d} ; \triEmb{\dot{id}}{\dot{S}}
      \end{array}
\\
        & \dRel^{id} &            & \dRel_S   \\
        & \sigma'   &\rMapsto^d  & \tau'     \\
\end{diagram}
\end{equation}
The existential quantification will then be expressed by projection,
as we proceed to show.

One can easily generalize BiKAT with a three-argument embedding 
we will write as $\tricom{\_}{\_}{\_}$.  But we also need BiKAT elements to 
encode the relations $R,S,id$.  So we define a 
\emph{two}-argument embedding 
$\triEmb{\_}{\_}$ with the following meaning: if $A,B$ are BiKAT elements, 
thus denoting pairs of executions, then $\triEmb{A}{B}$ denotes triples of executions,
comprising pairs from $A$ and $B$ \poplchanged{where they}{that} agree on the middle one.
Specifically, for elements $A$ and $B$ of a relational BiKAT,
define $\triEmb{A}{B}$ to be this relation on $\Sigma\times\Sigma\times\Sigma$.
\[ (\sigma,\sigma',\sigma'') 
   \triEmb{A}{B}
   (\tau,\tau',\tau'') \eqdef (\sigma,\sigma') A (\tau,\tau') \land 
                              (\sigma',\sigma'') B (\tau',\tau'') \]
For the special case of BiKAT elements embedded from the underlying KAT,
define the abbreviation 
% $\tricom{a}{b}{c}$ by
\[ \tricom{a}{b}{c} \eqdef \triEmb{\emb{a}{b}}{\emb{b}{c}} \]
Observe that 
$(\sigma,\sigma',\sigma'') \tricom{a}{b}{c} (\tau,\tau',\tau'')$
iff $\sigma a \tau \land \sigma' b \tau' \land \sigma'' c \tau''$.
%(Another interesting form is $\triEmb{A}{\emb{\hav}{c}}$.)
Now one can check that the relation on state triples depicted by the diagram in (\ref{eq:fsimAlt})
is denoted by the term on the right in (\ref{eq:fsimAlt}).

Among the available projections we  need the one that projects the left two of three.
This can be considered as a projection from the ``TriKAT'' to the left underlying BiKAT.
For any relation $X$ on $\Sigma\times\Sigma\times\Sigma$, define the relation $\Lproj X$ on $\Sigma\times\Sigma$  
by 
\[ (\sigma,\sigma') \Lproj X (\tau,\tau') 
\quad\mbox{iff}\quad
\some{\sigma'',\tau''}{ (\sigma,\sigma',\sigma'')X(\tau,\tau',\tau'')}
\]

\begin{lemma}\label{lem:trikat}
In a relational model, $c\sep d:\aespec{R}{S}$, i.e., (\ref{eq:fsim}),
is equivalent to this BiKAT equation:
\begin{equation}\label{eq:bitri}
  \dot{R} ; \emb{c}{\hav} ; \dot{id} 
  \;\leq\;
  \Lproj ( \triEmb{\dot{R}}{\dot{id}} ; \tricom{c}{\hav}{d} ; \triEmb{\dot{id}}{\dot{S}}) 
\end{equation}
and $c\sep d:\bespec{R}{S}$, i.e., (\ref{eq:bsim}) is equivalent to 
\begin{equation}\label{eq:bitriBack}
\dot{id} ; \emb{c}{\hav} ; \dot{S} 
\;\leq\;
\Lproj(\triEmb{\dot{id}}{\dot{R}} ; \tricom{c}{\hav}{d} ; \triEmb{\dot{S}}{\dot{id}})
\end{equation}
\end{lemma}
The proofs are by unfolding definitions.
The whole story works also for trace models.
We hoped that (\ref{eq:bitri}) could be used directly to prove simulations,
but were unsuccessful.  The reader may check that in relational or trace models
of BiKAT, projection distributes only weakly over sequence, as
$\lproj(A;B)\leq \lproj A; \lproj B$, 
and 
$\Lproj(X;Y)\leq \Lproj X; \Lproj Y$, 
which is not helpful for proving  equations like (\ref{eq:bitri}). 
\poplchanged{We do use (\ref{eq:bitri}) and (\ref{eq:bitriBack}) to derive
a couple of the deductive rules for simulation
in Sect.~\ref{sec:forwardBack}.}{
Nonetheless, we can use 
(\ref{eq:bitri}) and (\ref{eq:bitriBack}) to derive some rules including \rn{eDisj}
in Fig.~\ref{fig:aeRHL}.
}

The operation $\tricom{-}{-}{-}$ used above does distribute 
homomorphically over sequence and the other operators.  
As a result, BiKAT generalizes straightforwardly to $n$-KAT, which 
encompasses $\forall^n$ properties of program $n$-tuples, as in Cartesian Hoare logic~\cite{SousaD2016}.  
However, it is an open question how to axiomatize projections and the bi-to-tri embedding $\triEmb{-}{-}$
in a way that is useful for $\forall\exists$, 
so we do not formally define TriKAT.

\section{Discussion}\label{sec:discuss}

We have described BiKAT, a theory for equational reasoning about alignment
for relational program verification.
In this section we consider implications for automated reasoning, connections with other work including \popladded{product programs and} deductive verification, and problems for future work.

\subsection{Automation of Alignment}

Recent years have seen advances in automation of relational verification, as summarized 
in Sect.~\ref{subsec:relwork}. 
To our knowledge, none of these techniques algebraically derive an alignment. As we developed the algebraic theory of BiKAT, along the way we found possible approaches for automation such as integrating with existing automated techniques for discovering alignments or relational invariants. 
We save automation for future work, but summarize some general approaches here.

\emph{Solving BiKAT equivalence queries.} A BiKAT is itself a KAT and we can, therefore exploit a wide range of KAT-based tools such as symbolic equality reasoning~\cite{Pous2015},
Coq tactics~\cite{BraibantP10},
abstract interpretation with KAT~\cite{AntonopoulosKL2019},
and methods for constructing and deciding equality in concrete KATs~\cite{BeckettCG22}. 
Existing KAT tools do not have a built-in way of representing our particular kinds of KATs. However, we can encode a BiKAT, minus LRC, as a KAT through suitable variable renaming. We did this manually for a simple, concrete BiKAT and confirmed that KMT~\cite{BeckettCG22} was able to verify some equivalence queries. Of course we interactively use LRC to obtain the desired alignment
and do not give KMT instances of the LRC axiom, owing to undecidability.

\emph{Constraint-based relational verification.} \citet{UnnoTK21} reduced $k$-safety and possibilisitic non-interference to a constraint-satisfaction problem. Although these reductions are relatively complete, the approach does not scale well, as it searches for possible alignments. One possible path forward is to use BiKAT reasoning to algebraically derive alignments at a coarse-grained level, and then employ a constraint-satisfaction problem to solve the fine-grained subproblems.

\emph{Semi-automation.} Automated solving (e.g. via KAT tools or constraint-satisfaction) could also be used as part of 
a larger semi-automated reasoning framework. Our Coq development already provides the 
basic BiKAT laws/lemmas and could be extended to include forward/backward simulation rules, which would then be used interactively on given problem to derive an alignment. Along the way, automated solvers could be integrated and used to discharge smaller semantic queries.

\popladded{
Numerous other works discuss automation of relational verification.
\citet{PickFG18} describe a technique that aligns conditional blocks (in addition to loops)
and exploits symmetries to reduce the verification burden.
\citet{FarzanV2019} describe an approach to hypersafety verification of unary 
programs by discovering representative executions of a product program, whose correctness proofs
are sufficient to prove the overall property.
\citet{UnnoTK21} work with transition systems and 
use a constraint-solving approach to automatically discover 
a ``scheduler''---a form of alignment described as a function that directs
which element of the $k$-tuple (product of transition systems) should take the next step.
\citet{Badihi20} describe ARDiff, using a combination of abstraction and refinement for automatically proving program equivalence.
\citet{Mordvinov2019} work in the context of CHCs, and infer relational invariants.
%They introduce rules for relating loops across $k$-tuples of programs, but cannot support more complicated alignments that relate different implementations across multiple programs. By contrast, the BiKAT equational theory is built atop a program equational theory (KAT)
%https://dl.acm.org/doi/10.1145/3498689     GaeherEtal 
}

\popladded{\subsection{Product Programs and Expressibility}}
\label{subsec:prod}

BiKAT serves as notation for alignment products which in turn represent alignments of executions.
We have already shown examples of some alignments in the literature, 
including the $2x^2$ example of \citet{ShemerGSV19} (our Example~\ref{eg:two})
and the array insertion example of Shemer \emph{et al.} (in Sect.~\ref{sec:examples}).
In this section we consider what alignments can be represented in BiKAT, compared with
product programs and other representations in the literature.
We also consider how adequacy is established in various works,
compared with our Theorem~\ref{thm:allall}.

In some works, products are literally programs (e.g.,~\cite{BartheDArgenioRezk,BartheCK-FM11,EilersMH18}).  In others, products are represented in some form of control-flow automata (e.g.,~\cite{ChurchillP0A19}) or transition system (e.g.,~\cite{ShemerGSV19}).  
In all cases, products represent alignments between corresponding points in execution pairs (or $k$-tuples), for reasoning based on assertions at the aligned points.  

Representation of products as programs has the advantages (and disadvantages) of syntactic representation.  
For products as programs, one approach to ensuring adequacy is developed by 
\citet{BartheCK-FM11,BartheCK16} 
who develop a ternary judgment that connects two programs to a third that represents 
an aligned product of them.  
Assert commands are used to ensure adequacy. For example, in the case of if-else,
the guard-agreement side condition of rule \rn{dIf} (in Fig.~\ref{fig:RHLselected})
can be added to the product program as an initial assertion.
If the assertion holds, the product is adequate.  
For $\ifc{e}{c}{d}$ and $\ifc{e'}{c'}{d'}$,
one might try to write their product as this BiKAT term:
$\eqbib{e}{e'};(\eml{e};\emb{c}{c'} + \neg\eml{e};\emb{d}{d'})$.
However, this treats the initial agreement as an assumption, whereas for adequacy it must be a consequence of the precondition.  
KAT can be extended with failures (FailKAT) in order to express assertions~\cite{Mamouras17}.

The approach of Barthe \emph{et al.} has the advantage of making a close connection with Hoare logic,
but the disadvantage of lacking means to leverage left-right-commutativity as such.
\citet{BanerjeeNN16,BNNN19a} handle products 
using custom syntax for what they call \emph{biprograms}.
The semantics of their bi-if is essentially like having the guard-agreement assertion.
Their bi-command form $(c|c')$ serves as a product with no designated intermediate alignment.
Relational judgments apply to biprograms and a verification problem is posed 
in the form $(c|c')$.  
There is an auxiliary relation on biprograms, called weaving,
that effectively performs left-right-commutings in a way that preserves adequacy.
Whereas Barthe \emph{et al.} reduce relational verification to unary Hoare logic,
Banerjee \emph{et al.} use a custom proof system for biprograms.
In both of these lines of work, adequacy is proved as a general result about the system.
By inspection of our examples and Theorems~\ref{thm:RHL} and~\ref{thm:caWh},
one can see that BiKAT can represent the alignments 
achievable in these systems, when they are adequate,
but (as noted above) cannot directly encode adequacy checks expressed as assertions.
We conjecture that the systems can be encoded in an extension of BiKAT
based on FailKAT~\cite{Mamouras17}.
Such an extension would also serve another purpose, 
namely to encode the fault-sensitive variation of $\forall\forall$ used in practical verification systems and logics including~\citet{Yang:tcs04,BanerjeeNN16,BNNN19a}.

As mentioned earlier, the tiling example is handled in~\citet{BanerjeeNN16} 
using a custom rewriting relation with KAT-like rules.
The probabilistic relational Hoare logic of \citet{BartheGHS17} 
has a similar rule, called structural equivalence (as well as a rule for conditionally aligned loops).  The logic's relational correctness judgment connects the related programs to a product program that witnesses a probabilistic coupling.

As we discuss in Sect.~\ref{sec:examples},
procedure calls can be treated as primitives in BiKAT,
which can thereby express alignment of calls for use with relational specs as hypotheses.
However, BiKAT has no mean to express patterns of alignments involving 
nested procedure calls as in the work of \citet{GodlinS08}.
Nor does BiKAT provide for directly expressing alignment defined by code overlay as in the ghost monitors of \citet{ClochardMP20}.
However, history-sensitive alignment can be expressed in BiKAT and other systems 
using ghost code.  

Many prior works use products based on the representation of programs, and products, as transition systems.  We sketch how BiKAT can express such products quite generally. 
A transition system can be presented as rules of the form $g\rightarrow a$ where the guard $g$ is a state condition, and $a$ is an assignment command (or basic block).  So the program can be written in Dijkstra's guarded command notation~\cite{AptOld3} as 
 $\textsf{do}\ g_0\rightarrow a_0\ \talloblong \ g_1\rightarrow a_1\ \talloblong\ \ldots\ \textsf{od}$.
Regardless of the form of the commands $a_i$, this has a simple representation in KAT, as  
$ (g_0;a_0 + g_1;a_1 ...)^*;\neg(\Sigma_i \, g_i) $.
For clarity in the following discussion we ignore the negated condition and simply write
$ (g_0;a_0 + g_1;a_1 ...)^* $.
So the following BiKAT term represents a product of two transition systems.
\[ \emb{ (g_0;a_0 + g_1;a_1 ...)^* }{ (g_0';a_0' + g_1';a_1' ...)^* } \]
Alignment products in the literature constrain executions of the underlying programs by some conditions $L,R,J$ on state-pairs, that designate whether to take a left-only step, right-only step, or joint step.  (More generally, which of $k$ copies, as in \citet{ShemerGSV19},
\citet{EilersMH18}.)
In a BiKAT, assuming the conditions $L,R,J$ are expressible as bitests, the alignment is expressible using finite sums as 
\begin{equation}\label{eq:alignA}
((\Sigma_{g,a,g',a'}\ J;\emb{g;a}{g';a'}) + (\Sigma_{g,a}\ L;\emb{g;a}{1}) + (\Sigma_{g',a'}\ R;\emb{1}{g';a'}))^* 
\end{equation}
In these sums, $g,a$ range over the guarded actions $g\rightarrow a$ of the left program and $g',a'$ range over the right. 

So much for expressing alignments.  What about proving adequacy?  One can
formulate general conditions under which a product is adequate.  Roughly, the
idea is that $L\lor R\lor J\lor TRM$ must be invariant, where $TRM$ stands for
``both sides terminated''.  \citet{ShemerGSV19} give an adequacy
result of this form. (Their term is ``fairness''.  There is no standard term; we
take ``adequacy'' from~\citet{NagasamudramN21}.)  
\citet{ChurchillP0A19} formulate adequacy as verification conditions involving
their alignment invariants.  In our setting, adequacy is proved equationally
(Theorem~\ref{thm:allall}), raising the question whether the term
(\ref{eq:alignA}) can be derived from $\emb{ (g_0;a_0 + g_1;a_1 ...)^* }{
(g_0';a_0' + g_1';a_1' ...)^* }$.  Invariance of $L\lor R\lor J\lor TRM$ is
analogous to the side condition of the proof rule \rn{caWh}, and the equality of
terms $\emb{ (g_0;a_0 + g_1;a_1 ...)^* }{ (g_0';a_0' + g_1';a_1' ...)^* }$ and
(\ref{eq:alignA}) is analogous to how the expansion law (\ref{eq:caWhexpand}) is
used to prove \rn{caWh}.  We conjecture that the general equality can be proved,
in *-continuous BiKATs, just as we have done for (\ref{eq:caWhexpand}).  

The $L,R,J$ form discussed above is very general. 
In implementations, the conditions $L,R,J$ tend to be restricted to constraints supported by an efficient solver, and the same restrictions would be applicable in uses of BiKAT. 
Ignoring such restrictions, that general form seems as expressive as BiKAT.  
We are not aware of patterns that can be expressed in BiKAT but not by products represented as transition systems.

\subsection{Other Related Work}
\label{subsec:relwork}

We have covered many of the most related works; we now mention a few others.

\popladded{
In recent work \citet{DosualdoFD2022} describe a logic for hyper-triple composition (LHC) based on weakest preconditions that can decompose a hypersafety proof along the boundary of hyper tuples,
offering ways of combining multiple $k$-safety proofs with differing $k$s.
In contrast with BiKAT, LHC is a calculus based on weakest pre-condition rather than an equational system,
and supports only $\forall\forall$ pre/post $k$-safety properties.
There may be a connection between LHC-style decomposition and our proposed work on TriKAT discussed in Sect.~\ref{sec:trikat}. Both permit ways to combine two relational proofs that share some common program terms into an overall proof by correlating the common terms. However, we leave this investigation to future work. 
}

\popladded{
\citet{BartheEGGKM2019} discuss relational verification in a first order predicate logic in which program variables are represented as functions $v(i,tr)$ over a time step $i$ and trace identified by $tr$. The authors' encoding can express highly non-local relationships between traces such as equating the value of $v$ at the beginning of one trace with the value of $v$ at the end of another trace. 
This approach does not involve deriving an alignment, but rather enables correlating of arbitrary computation steps.
The encoding in FO exploits quantifiers available in first-order provers.
Although, in principle, the traces could be quantified existentially, the authors only discuss 
$\forall\forall$ properties non-interference and sensitivity without mention of quantifier alternation over traces.
} 

\citet{MurrayUnderApproxRHL} introduces a relational incorrectness logic for imperative programs,
inspired by incorrectness logic of \citet{OHearn2019}.  Using relational semantics, the judgment relates $c$ to $d$ for spec $R,S$ iff
$\all{\tau,\tau'}{ \tau R \tau' \imp \some{\sigma,\sigma'}{\sigma R \sigma' \land \sigma c \tau \land \sigma' d \tau' }}$.  So the postcondition is an underapproximation of the reachable pairs.
In a BiKAT over a KAT with havoc, this can be expressed 
as $\emb{\hav}{\hav};\dot{S} \leq \emb{\hav}{\hav};\dot{R};\emb{c}{d}$,
generalizing O'Hearn's KAT formulation of incorrectness~\cite[Sect.~5.3]{OHearn2019}.
\citet{ZhangAG22} investigated KATs with top for (unary) incorrectness logic.

Although KAT equations under commutativity hypotheses are undecidable, there are recent positive results for other classes of hypotheses~\cite{DoumaneKPP19,PousRW21}.
Synchronous KAT~\cite{WagemakerBKR019} has models based on strings-of-sets 
which can be interpreted as multiple simultaneous actions;
this could perhaps be used to model the step-by-step alignments of~\citet{KovacsSF13}
and~\citet{BNNN19a}.

Apropos $\forall\forall$ properties, 
\citet{TerauchiA2005} 
introduce the term 2-safety and 
describe a type system-based alignment used for secure information flow. Their rules (e.g.\ their Fig.~8) can be formulated so that the alignments are expressed in BiKAT, leading to a more equational algebraic derivation strategy of the non-interference property.
\citet{SousaD2016} describe Cartesian Hoare Logic for reasoning about $k$-safety of individual programs by alignment without explicit representation of a product program.
\iffalse
\citet{EilersMH18} described modular product programs (MPPs), 
in which the procedure call boundary is the unit of modularity for relational verification. Their work creates a $k$-way lock-step product program and interprets relational specifications in the unary product program.
BiKAT subsumes some aspects of MPP, including native support 
for relational specifications as bitests 
(and unary specifications as left/right embedded tests)
that can be algebraically manipulated and a wider class of alignments beyond MPPs' lock-step alignment,
and can reason about pairs of programs rather than two copies of one program. 
On the other hand, MPPs are automated
and have full support for procedures with parameters, whereas BiKAT currently only supports procedure calls as fixed atomic actions and hypotheses. 
\fi
\citet{EilersMH18} describe a $k$-way lock-step product 
encoding for single programs, that facilitates the use of $k$-safety procedure specifications.

Apropos $\forall\exists$ properties, 
\citet{LamportS21} use TLA+ as a logic for deductive reasoning about
such properties in the setting of temporal logic.
\citet{ClochardMP20} manually encode alignment products as programs 
in the Why3 deductive verification tool, including resolution of nondeterminacy to prove $\forall\exists$ properties.
\citet{BartheCK13} define left products for proving forward simulation,
in a formulation based on control flow graphs.
\citet{hawblitzelklr13} use the forward simulation property
(calling it \emph{relative termination}), in translation validation using relational summaries
and verification conditions.
Example~\ref{eg:paths} was drawn from \citet{BeutnerF22} who discuss
a more generalized possibilistic non-interference $\forall^k\exists^l$.

% FUTURE
%% In reasoning about program equivalence (or refinement) it is often desirable 
%% to \emph{transitively compose} judgments, e.g., between intermediate representations in a compiler pipeline.
%% \eric{a proposed major trimming:}
%% \changed{
%% The basic $\forall\forall$ property does not transitively compose, owing to termination-insensitivity.
%% Forward simulation does compose transitively (sometimes called vertical composition):
%% from $b\sep c:\aespec{P}{Q}$ and $c\sep d:\aespec{R}{S}$ 
%% it follows that $b\sep d:\aespec{(P;R)}{(Q;S)}$.
%% Equivalence and indistinguishability are examples of relations that are reflexive and transitive;
%% for such an $R$, from $b\sep c:\aespec{R}{R}$ and $c\sep d:\aespec{R}{R}$ 
%% it follows that $b\sep d:\aespec{R}{R}$.
%% Hawblitzl \emph{et al.}~\cite{hawblitzelklr13} use the forward simulation property, dubbed \emph{relative termination}, 
%% for transitive composition of relational procedure summaries in translation validation
%% based on verification conditions. 
%% }{
%% While the basic $\forall\forall$ property does not transitively compose, forward simulation does compose. Hawblitzl \emph{et al.}~\cite{hawblitzelklr13} use transitivity and forward simulation for translation validation.
%% }

None of the preceding works provide inference rules for $\forall\exists$ judgments.
The Iris-based relational logic~\cite{FruminKB18,GaeherEtal} does so, 
in particular a judgment for contextual refinement, in forward simulation form.
The logic has complex features catering for higher order concurrent programs.  
Prophecy variables have been added~\cite{FruminKB20} so backward simulation 
reasoning is available in some form.
It is unclear whether the simple forward and backward rules applicable to sequential programs
can be extracted from this framework.
\citet{MaillardHRM20} develop a 
%general framework,
relational dependent type theory
that accounts for core RHL rules encompassing a range of computational effects.
It is not clear that the framework facilitates manipulation of intricate data-dependent 
alignments together with simple first-order relational assertions as used 
in automated relational verification.

\poplremoved{
Numerous other works discuss automation of relational verification.
Pick \emph{et al.}~\cite{PickFG18} describe a technique that aligns conditional blocks (in addition to loops)
and exploits symmetries to reduce the verification burden.
Farzan and Vandikas~\cite{FarzanV2019} describe an approach to hypersafety verification of unary 
programs by discovering representative executions of a product program, whose correctness proofs
are sufficient to prove the overall property.
Unno \emph{et al.}~\cite{UnnoTK21} work with transition systems and 
use a constraint-solving approach to automatically discover 
a ``scheduler''---a form of alignment described as a function that directs
which element of the $k$-tuple (product of transition systems) should take the next step.
Badihi \emph{et al.}~\citep{Badihi20} describe ARDiff, using a combination of abstraction and refinement for automatically proving program equivalence.
Mordvinov and Fedyukovich~\citep{Mordvinov2019} work in the context of CHCs, and infer relational invariants.
}

\subsection{Future Work}

BiKAT provides a foundation for several interesting directions for future work
in addition to questions raised in Sects.~\ref{sec:RHL} and~\ref{sec:beyond}.
As already detailed above, we are optimistic about the outlook for automation.

There are several open questions about completeness.  We have proved soundness for
a number of forward and backward simulation rules, but have not investigated their completeness.
Completeness in the sense of ``true implies provable'' is problematic 
owing to the possibility to reduce the problem to unary Hoare logic via the self-composition rule.
A more relevant notion is alignment completeness, which so far was formulated only for 2-safety~\cite{NagasamudramN21}.
One may think that BiKAT provides a notion of ``regular alignment'', which could 
provide a yardstick to evalute alignment completeness of other systems.
However, in an applied BiKAT the interpretation of tests, in combination with conditional
alignments like Eqn.~(\ref{eq:caWhexpand}), expresses more than regular patterns.

\poplchanged{As for completeness of BiKAT axioms with respect to models, 
we note that straightforward extensions 
of the guarded string model (used by Kozen and Smith for completeness of KAT~\cite{Kozen1996}) do not validate LRC.}{%
What about completeness of the BiKAT axioms with respect to models? 
For KAT, \citet{Kozen1996} obtain completeness for relational models from a completeness result for a language model called guarded strings.
A \emph{guarded string} model is given by sets of primitive actions 
and primitive tests. 
% (For some purposes one restricts to finite sets.)
An \emph{atom} is a boolean valuation of primitive tests.
The model is the trace model where $\Sigma$ is the set of atoms and all sequences are admissible. 
The guarded string model has a canonical interpretation of its primitive tests and actions.
Being a trace model, it gives rise to a trace BiKAT.
One can then interpret primitive bitests as relations on atoms.
However, the trace BiKAT for guarded strings does not determine a canonical interpretation of bitests; nor does it validate LRC.
Several problems remain open:
Is BiKAT, or BiKAT with projections, complete for relational models? 
Is there a class of BiKATs where the underlying KAT is a guarded string model,
for which BiKAT is complete? 
}

\poplchanged{
General study of algebras of programs has a long history~\cite{Hoare:laws}
and remains active~\cite{HofnerPS19}.  Because KAT is about control structure it seems
particularly suited to describing alignment.  For reasoning, however, other 
algebraic structures are relevant, e.g., we have highlighted the relevance of relation 
algebra, as also done in the Coq library of Damien Pous for KAT~\cite{PousKATlib}.
Such settings may be helpful for further exploration of $\forall\exists$ properties
along the lines of TriKAT,
and for relational properties of nonterminating programs.
One specific question is whether some formulation of TriKAT can be used to obtain 
an algebraic proof to generalize Theorem~\ref{thm:sim:complete}.
}{
General study of algebras of programs has a long history~\cite{Hoare:laws}
and remains active~\cite{HofnerPS19}.  Because KAT is about control structure it seems
particularly suited to describing alignment.  For reasoning, however, other 
algebraic structures are relevant, e.g., Concurrent Kleene Algebra~\cite{HoareSMSZ16}
has an operator that 
can model spatial separation and it would be interesting to explore some notion of 
spatial locality in connection with alignment and BiKAT.
We have highlighted the relevance of relation 
algebra, as also done in the Coq library of Damien Pous for KAT~\cite{PousKATlib}.
Such settings may be helpful for further exploration of $\forall\exists$ properties
along the lines of TriKAT,
and for relational properties of nonterminating programs
(e.g., using $\omega$-algebra~\cite{Cohen00}).     
One specific question is whether some formulation of TriKAT can be used to obtain 
an algebraic proof to generalize Theorem~\ref{thm:sim:complete}.
}

% Future - Brotherston et al, Commutativity rule. 

\begin{acks}
	\popladded{
The authors would like to thank Anindya Banerjee, Lennart Beringer and Michael Greenberg for helpful discussions and the anonymous reviewers for their valuable feedback.
}

\popladded{
Authors Antonopoulos, Koskinen, Le, and Naumann were supported in part by 
the \grantsponsor{GSONR}{Office of Naval Research}{https://www.nre.navy.mil/} under Grant No.~\grantnum{GSONR}{N00014-17-1-2787}. 
Antonopoulos, Koskinen and Le were supported in part by \grantsponsor{GSNSF}{NSF}{https://www.nsf.gov} award \grantnum{GSNSF}{CCF-2106845}. 
Antonopoulos was supported in part by \grantsponsor{GSNSF}{NSF}{https://www.nsf.gov} award \grantnum{GSNSF}{CCF-2131476}.
Naumann, Nagasamudram, and Ngo were supported in part by \grantsponsor{GSNSF}{NSF}{https://www.nsf.gov} award \grantnum{GSNSF}{CNS-1718713}.
}
\end{acks}

% \vfill
% \pagebreak

\bibliographystyle{ACM-Reference-Format}
\bibliography{bib}


\begin{thebibliography}{73}

%%% the Web address.

\ifx \showCODEN    \undefined \def \showCODEN     #1{\unskip}     \fi
\ifx \showDOI      \undefined \def \showDOI       #1{#1}\fi
\ifx \showISBNx    \undefined \def \showISBNx     #1{\unskip}     \fi
\ifx \showISBNxiii \undefined \def \showISBNxiii  #1{\unskip}     \fi
\ifx \showISSN     \undefined \def \showISSN      #1{\unskip}     \fi
\ifx \showLCCN     \undefined \def \showLCCN      #1{\unskip}     \fi
\ifx \shownote     \undefined \def \shownote      #1{#1}          \fi
\ifx \showarticletitle \undefined \def \showarticletitle #1{#1}   \fi
\ifx \showURL      \undefined \def \showURL       {\relax}        \fi
% invisible to TeX
\providecommand\bibfield[2]{#2}
\providecommand\bibinfo[2]{#2}
\providecommand\natexlab[1]{#1}
\providecommand\showeprint[2][]{arXiv:#2}

\bibitem[\protect\citeauthoryear{Abadi and Lamport}{Abadi and Lamport}{1988}]%
        {AbadiLamport88}
\bibfield{author}{\bibinfo{person}{Mart\'{i}n Abadi} {and}
  \bibinfo{person}{Leslie Lamport}.} \bibinfo{year}{1988}\natexlab{}.
\newblock \showarticletitle{The Existence of Refinement Mappings}. In
  \bibinfo{booktitle}{\emph{Proceedings of LICS}}.
\newblock


\bibitem[\protect\citeauthoryear{Antonopoulos, Koskinen, and Le}{Antonopoulos
  et~al\mbox{.}}{2019}]%
        {AntonopoulosKL2019}
\bibfield{author}{\bibinfo{person}{Timos Antonopoulos}, \bibinfo{person}{Eric
  Koskinen}, {and} \bibinfo{person}{Ton~Chanh Le}.}
  \bibinfo{year}{2019}\natexlab{}.
\newblock \showarticletitle{Specification and inference of trace refinement
  relations}.
\newblock \bibinfo{journal}{\emph{Proceedings of the ACM on Programming
  Languages}} \bibinfo{volume}{3}, \bibinfo{number}{OOPSLA}
  (\bibinfo{year}{2019}), \bibinfo{pages}{1--30}.
\newblock


\bibitem[\protect\citeauthoryear{Antonopoulos, Koskinen, Le, Nagasamudram,
  Naumann, and Ngo}{Antonopoulos et~al\mbox{.}}{2022a}]%
        {BiKATarxiv}
\bibfield{author}{\bibinfo{person}{Timos Antonopoulos}, \bibinfo{person}{Eric
  Koskinen}, \bibinfo{person}{Ton~Chanh Le}, \bibinfo{person}{Ramana
  Nagasamudram}, \bibinfo{person}{David~A. Naumann}, {and}
  \bibinfo{person}{Minh Ngo}.} \bibinfo{year}{2022}\natexlab{a}.
\newblock \showarticletitle{An algebra of alignment for relational verification
  (extended version)}.
\newblock \bibinfo{journal}{\emph{CoRR}}  \bibinfo{volume}{abs/2202.04278}
  (\bibinfo{year}{2022}).
\newblock
\showeprint[arXiv]{2202.04278}
\urldef\tempurl%
\url{https://arxiv.org/abs/2202.04278}
\showURL{%
\tempurl}


\bibitem[\protect\citeauthoryear{Antonopoulos, Koskinen, Le, Nagasamudram,
  Naumann, and Ngo}{Antonopoulos et~al\mbox{.}}{2022b}]%
        {BiKATv}
\bibfield{author}{\bibinfo{person}{Timos Antonopoulos}, \bibinfo{person}{Eric
  Koskinen}, \bibinfo{person}{Ton~Chanh Le}, \bibinfo{person}{Ramana
  Nagasamudram}, \bibinfo{person}{David~A. Naumann}, {and}
  \bibinfo{person}{Minh Ngo}.} \bibinfo{year}{2022}\natexlab{b}.
\newblock \bibinfo{booktitle}{\emph{{An algebra of alignment for relational
  verification (artifact)}}}.
\newblock
\urldef\tempurl%
\url{https://doi.org/10.5281/zenodo.7144067}
\showDOI{\tempurl}


\bibitem[\protect\citeauthoryear{Apt, de~Boer, and Olderog}{Apt
  et~al\mbox{.}}{2009}]%
        {AptOld3}
\bibfield{author}{\bibinfo{person}{Krzysztof~R. Apt}, \bibinfo{person}{Frank~S.
  de Boer}, {and} \bibinfo{person}{Ernst-R\"{u}diger Olderog}.}
  \bibinfo{year}{2009}\natexlab{}.
\newblock \bibinfo{booktitle}{\emph{Verification of Sequential and Concurrent
  Programs} (\bibinfo{edition}{3} ed.)}.
\newblock \bibinfo{publisher}{Springer}.
\newblock
\urldef\tempurl%
\url{https://doi.org/10.1007/978-1-84882-745-5}
\showDOI{\tempurl}


\bibitem[\protect\citeauthoryear{Badihi, Akinotcho, Li, and Rubin}{Badihi
  et~al\mbox{.}}{2020}]%
        {Badihi20}
\bibfield{author}{\bibinfo{person}{Sahar Badihi}, \bibinfo{person}{Faridah
  Akinotcho}, \bibinfo{person}{Yi Li}, {and} \bibinfo{person}{Julia Rubin}.}
  \bibinfo{year}{2020}\natexlab{}.
\newblock \showarticletitle{ARDiff: Scaling Program Equivalence Checking via
  Iterative Abstraction and Refinement of Common Code}. In
  \bibinfo{booktitle}{\emph{Joint Meeting on European Software Engineering
  Conference and Symposium on the Foundations of Software Engineering}}.
  \bibinfo{pages}{13–24}.
\newblock
\urldef\tempurl%
\url{https://doi.org/10.1145/3368089.3409757}
\showDOI{\tempurl}


\bibitem[\protect\citeauthoryear{Banerjee, Nagasamudram, Nikouei, and
  Naumann}{Banerjee et~al\mbox{.}}{2022}]%
        {BNNN19a}
\bibfield{author}{\bibinfo{person}{Anindya Banerjee}, \bibinfo{person}{Ramana
  Nagasamudram}, \bibinfo{person}{Mohammad Nikouei}, {and}
  \bibinfo{person}{David~A. Naumann}.} \bibinfo{year}{2022}\natexlab{}.
\newblock \showarticletitle{A Relational Program Logic with Data Abstraction
  and Dynamic Framing}.
\newblock \bibinfo{journal}{\emph{ACM Transactions on Programming Languages and
  Systems}} (\bibinfo{year}{2022}).
\newblock
\newblock
\shownote{Accepted for publication. Available as
  \url{http://arxiv.org/abs/1910.14560}}.


\bibitem[\protect\citeauthoryear{Banerjee, Naumann, and Nikouei}{Banerjee
  et~al\mbox{.}}{2016}]%
        {BanerjeeNN16}
\bibfield{author}{\bibinfo{person}{Anindya Banerjee}, \bibinfo{person}{David~A.
  Naumann}, {and} \bibinfo{person}{Mohammad Nikouei}.}
  \bibinfo{year}{2016}\natexlab{}.
\newblock \showarticletitle{Relational Logic with Framing and Hypotheses}. In
  \bibinfo{booktitle}{\emph{36th {IARCS} Annual Conference on Foundations of
  Software Technology and Theoretical Computer Science}}.
\newblock
\newblock
\shownote{Long version at \url{http://arxiv.org/abs/1611.08992}}.


\bibitem[\protect\citeauthoryear{Barthe, Crespo, and Kunz}{Barthe
  et~al\mbox{.}}{2011a}]%
        {BartheCK-FM11}
\bibfield{author}{\bibinfo{person}{Gilles Barthe}, \bibinfo{person}{Juan~Manuel
  Crespo}, {and} \bibinfo{person}{C{\'e}sar Kunz}.}
  \bibinfo{year}{2011}\natexlab{a}.
\newblock \showarticletitle{Relational Verification Using Product Programs}. In
  \bibinfo{booktitle}{\emph{Formal Methods}}.
\newblock


\bibitem[\protect\citeauthoryear{Barthe, Crespo, and Kunz}{Barthe
  et~al\mbox{.}}{2013}]%
        {BartheCK13}
\bibfield{author}{\bibinfo{person}{Gilles Barthe}, \bibinfo{person}{Juan~Manuel
  Crespo}, {and} \bibinfo{person}{C{\'{e}}sar Kunz}.}
  \bibinfo{year}{2013}\natexlab{}.
\newblock \showarticletitle{Beyond 2-Safety: Asymmetric Product Programs for
  Relational Program Verification}. In \bibinfo{booktitle}{\emph{Logical
  Foundations of Computer Science {(LFCS)}}} \emph{(\bibinfo{series}{LNCS},
  Vol.~\bibinfo{volume}{7734})}. \bibinfo{pages}{29--43}.
\newblock


\bibitem[\protect\citeauthoryear{Barthe, Crespo, and Kunz}{Barthe
  et~al\mbox{.}}{2016}]%
        {BartheCK16}
\bibfield{author}{\bibinfo{person}{Gilles Barthe}, \bibinfo{person}{Juan~Manuel
  Crespo}, {and} \bibinfo{person}{C{\'e}sar Kunz}.}
  \bibinfo{year}{2016}\natexlab{}.
\newblock \showarticletitle{Product Programs and Relational Program Logics}.
\newblock \bibinfo{journal}{\emph{J. Logical and Algebraic Methods in
  Programming}} \bibinfo{volume}{85}, \bibinfo{number}{5}
  (\bibinfo{year}{2016}), \bibinfo{pages}{847--859}.
\newblock


\bibitem[\protect\citeauthoryear{Barthe, D'Argenio, and Rezk}{Barthe
  et~al\mbox{.}}{2004}]%
        {BartheDArgenioRezk}
\bibfield{author}{\bibinfo{person}{Gilles Barthe}, \bibinfo{person}{Pedro~R.
  D'Argenio}, {and} \bibinfo{person}{Tamara Rezk}.}
  \bibinfo{year}{2004}\natexlab{}.
\newblock \showarticletitle{Secure Information Flow by Self-Composition}. In
  \bibinfo{booktitle}{\emph{{IEEE CSFW}}}.
\newblock
\newblock
\shownote{See extended version~\cite{BartheDR11}}.


\bibitem[\protect\citeauthoryear{Barthe, D'Argenio, and Rezk}{Barthe
  et~al\mbox{.}}{2011b}]%
        {BartheDR11}
\bibfield{author}{\bibinfo{person}{Gilles Barthe}, \bibinfo{person}{Pedro~R.
  D'Argenio}, {and} \bibinfo{person}{Tamara Rezk}.}
  \bibinfo{year}{2011}\natexlab{b}.
\newblock \showarticletitle{Secure information flow by self-composition}.
\newblock \bibinfo{journal}{\emph{Math. Struct. Comput. Sci.}}
  \bibinfo{volume}{21}, \bibinfo{number}{6} (\bibinfo{year}{2011}).
\newblock


\bibitem[\protect\citeauthoryear{Barthe, Eilers, Georgiou, Gleiss, Kov{\'a}cs,
  and Maffei}{Barthe et~al\mbox{.}}{2019}]%
        {BartheEGGKM2019}
\bibfield{author}{\bibinfo{person}{Gilles Barthe}, \bibinfo{person}{Renate
  Eilers}, \bibinfo{person}{Pamina Georgiou}, \bibinfo{person}{Bernhard
  Gleiss}, \bibinfo{person}{Laura Kov{\'a}cs}, {and} \bibinfo{person}{Matteo
  Maffei}.} \bibinfo{year}{2019}\natexlab{}.
\newblock \showarticletitle{Verifying relational properties using trace logic}.
  In \bibinfo{booktitle}{\emph{2019 Formal Methods in Computer Aided Design
  (FMCAD)}}. \bibinfo{pages}{170--178}.
\newblock


\bibitem[\protect\citeauthoryear{Barthe, Gr{\'{e}}goire, Hsu, and Strub}{Barthe
  et~al\mbox{.}}{2017}]%
        {BartheGHS17}
\bibfield{author}{\bibinfo{person}{Gilles Barthe}, \bibinfo{person}{Benjamin
  Gr{\'{e}}goire}, \bibinfo{person}{Justin Hsu}, {and}
  \bibinfo{person}{Pierre{-}Yves Strub}.} \bibinfo{year}{2017}\natexlab{}.
\newblock \showarticletitle{Coupling proofs are probabilistic product
  programs}. In \bibinfo{booktitle}{\emph{ACM Symposium on Principles of
  Programming Languages}}. \bibinfo{pages}{161--174}.
\newblock
\urldef\tempurl%
\url{https://doi.org/10.1145/3009837.3009896}
\showDOI{\tempurl}


\bibitem[\protect\citeauthoryear{Beckert and Ulbrich}{Beckert and
  Ulbrich}{2018}]%
        {BeckertU18short}
\bibfield{author}{\bibinfo{person}{Bernhard Beckert} {and}
  \bibinfo{person}{Mattias Ulbrich}.} \bibinfo{year}{2018}\natexlab{}.
\newblock \showarticletitle{Trends in relational program verification}.
\newblock In \bibinfo{booktitle}{\emph{Principled Software Development}}.
  \bibinfo{publisher}{Springer}, \bibinfo{pages}{41--58}.
\newblock


\bibitem[\protect\citeauthoryear{Benton}{Benton}{2004}]%
        {Benton:popl04}
\bibfield{author}{\bibinfo{person}{N. Benton}.}
  \bibinfo{year}{2004}\natexlab{}.
\newblock \showarticletitle{Simple Relational Correctness Proofs for Static
  Analyses and Program Transformations}. In \bibinfo{booktitle}{\emph{{POPL}}}.
  \bibinfo{pages}{14--25}.
\newblock


\bibitem[\protect\citeauthoryear{Beringer}{Beringer}{2011}]%
        {Beringer11}
\bibfield{author}{\bibinfo{person}{Lennart Beringer}.}
  \bibinfo{year}{2011}\natexlab{}.
\newblock \showarticletitle{Relational Decomposition}. In
  \bibinfo{booktitle}{\emph{Interactive Theorem Proving {(ITP)}}}
  \emph{(\bibinfo{series}{LNCS}, Vol.~\bibinfo{volume}{6898})}.
\newblock


\bibitem[\protect\citeauthoryear{Beutner and Finkbeiner}{Beutner and
  Finkbeiner}{2022}]%
        {BeutnerF22}
\bibfield{author}{\bibinfo{person}{Raven Beutner} {and} \bibinfo{person}{Bernd
  Finkbeiner}.} \bibinfo{year}{2022}\natexlab{}.
\newblock \showarticletitle{Software Verification of Hyperproperties Beyond
  k-Safety}. In \bibinfo{booktitle}{\emph{Computer Aided Verification}}.
  \bibinfo{pages}{341--362}.
\newblock
\urldef\tempurl%
\url{https://doi.org/10.1007/978-3-031-13185-1\_17}
\showDOI{\tempurl}


\bibitem[\protect\citeauthoryear{Blatter, Kosmatov, and Prevosto}{Blatter
  et~al\mbox{.}}{2022}]%
        {Blatter22}
\bibfield{author}{\bibinfo{person}{Lionel Blatter}, \bibinfo{person}{Nikolai
  Kosmatov}, {and} \bibinfo{person}{Pascale Prevosto, Virgileand Le~Gall}.}
  \bibinfo{year}{2022}\natexlab{}.
\newblock \showarticletitle{Certified Verification of Relational Properties}.
  In \bibinfo{booktitle}{\emph{Integrated Formal Methods}}.
  \bibinfo{pages}{86--105}.
\newblock
\showISBNx{978-3-031-07727-2}


\bibitem[\protect\citeauthoryear{Braibant and Pous}{Braibant and Pous}{2010}]%
        {BraibantP10}
\bibfield{author}{\bibinfo{person}{Thomas Braibant} {and}
  \bibinfo{person}{Damien Pous}.} \bibinfo{year}{2010}\natexlab{}.
\newblock \showarticletitle{An efficient {Coq} tactic for deciding {Kleene}
  algebras}. In \bibinfo{booktitle}{\emph{International Conference on
  Interactive Theorem Proving}}. \bibinfo{pages}{163--178}.
\newblock


\bibitem[\protect\citeauthoryear{Churchill, Padon, Sharma, and Aiken}{Churchill
  et~al\mbox{.}}{2019}]%
        {ChurchillP0A19}
\bibfield{author}{\bibinfo{person}{Berkeley~R. Churchill},
  \bibinfo{person}{Oded Padon}, \bibinfo{person}{Rahul Sharma}, {and}
  \bibinfo{person}{Alex Aiken}.} \bibinfo{year}{2019}\natexlab{}.
\newblock \showarticletitle{Semantic program alignment for equivalence
  checking}. In \bibinfo{booktitle}{\emph{{PLDI}}}.
\newblock


\bibitem[\protect\citeauthoryear{Clarkson, Finkbeiner, Koleini, Micinski, Rabe,
  and S{\'{a}}nchez}{Clarkson et~al\mbox{.}}{2014}]%
        {ClarksonFKMRS14}
\bibfield{author}{\bibinfo{person}{Michael~R. Clarkson}, \bibinfo{person}{Bernd
  Finkbeiner}, \bibinfo{person}{Masoud Koleini}, \bibinfo{person}{Kristopher~K.
  Micinski}, \bibinfo{person}{Markus~N. Rabe}, {and}
  \bibinfo{person}{C{\'{e}}sar S{\'{a}}nchez}.}
  \bibinfo{year}{2014}\natexlab{}.
\newblock \showarticletitle{Temporal Logics for Hyperproperties}. In
  \bibinfo{booktitle}{\emph{Principles of Security and Trust ({POST})}}
  \emph{(\bibinfo{series}{LNCS}, Vol.~\bibinfo{volume}{8414})}.
  \bibinfo{pages}{265--284}.
\newblock


\bibitem[\protect\citeauthoryear{Clarkson and Schneider}{Clarkson and
  Schneider}{2010}]%
        {ClarksonSchneiderHyper10}
\bibfield{author}{\bibinfo{person}{Michael~R. Clarkson} {and}
  \bibinfo{person}{Fred~B. Schneider}.} \bibinfo{year}{2010}\natexlab{}.
\newblock \showarticletitle{Hyperproperties}.
\newblock \bibinfo{journal}{\emph{Journal of Computer Security}}
  \bibinfo{volume}{18}, \bibinfo{number}{6} (\bibinfo{year}{2010}),
  \bibinfo{pages}{1157--1210}.
\newblock


\bibitem[\protect\citeauthoryear{Clochard, March{\'{e}}, and
  Paskevich}{Clochard et~al\mbox{.}}{2020}]%
        {ClochardMP20}
\bibfield{author}{\bibinfo{person}{Martin Clochard}, \bibinfo{person}{Claude
  March{\'{e}}}, {and} \bibinfo{person}{Andrei Paskevich}.}
  \bibinfo{year}{2020}\natexlab{}.
\newblock \showarticletitle{Deductive Verification with Ghost Monitors}.
\newblock \bibinfo{journal}{\emph{Proc. {ACM} Program. Lang.}}
  \bibinfo{volume}{4}, \bibinfo{number}{{POPL}} (\bibinfo{year}{2020}).
\newblock


\bibitem[\protect\citeauthoryear{Cohen}{Cohen}{2000}]%
        {Cohen00}
\bibfield{author}{\bibinfo{person}{Ernie Cohen}.}
  \bibinfo{year}{2000}\natexlab{}.
\newblock \showarticletitle{Separation and Reduction}. In
  \bibinfo{booktitle}{\emph{Mathematics of Program Construction}}
  \emph{(\bibinfo{series}{LNCS}, Vol.~\bibinfo{volume}{1837})}.
  \bibinfo{pages}{45--59}.
\newblock
\urldef\tempurl%
\url{https://doi.org/10.1007/10722010\_4}
\showDOI{\tempurl}


\bibitem[\protect\citeauthoryear{de~Roever, de~Boer, Hannemann, Hooman,
  Lakhnech, Poel, and Zwiers}{de~Roever et~al\mbox{.}}{2001}]%
        {Concur:deRoever}
\bibfield{author}{\bibinfo{person}{Willem-Paul de Roever},
  \bibinfo{person}{Frank de Boer}, \bibinfo{person}{Ulrich Hannemann},
  \bibinfo{person}{Jozef Hooman}, \bibinfo{person}{Yassine Lakhnech},
  \bibinfo{person}{Mannes Poel}, {and} \bibinfo{person}{Job Zwiers}.}
  \bibinfo{year}{2001}\natexlab{}.
\newblock \bibinfo{booktitle}{\emph{Concurrency Verification: Introduction to
  Compositional and Noncompositional Methods}}.
\newblock \bibinfo{publisher}{Cambridge University}.
\newblock


\bibitem[\protect\citeauthoryear{de~Roever and Engelhardt}{de~Roever and
  Engelhardt}{1998}]%
        {deRdataref}
\bibfield{author}{\bibinfo{person}{Willem-Paul de Roever} {and}
  \bibinfo{person}{Kai Engelhardt}.} \bibinfo{year}{1998}\natexlab{}.
\newblock \bibinfo{booktitle}{\emph{Data Refinement: Model-Oriented Proof
  Methods and their Comparison}}.
\newblock \bibinfo{publisher}{Cambridge University Press}.
\newblock


\bibitem[\protect\citeauthoryear{Doumane, Kuperberg, Pous, and Pradic}{Doumane
  et~al\mbox{.}}{2019}]%
        {DoumaneKPP19}
\bibfield{author}{\bibinfo{person}{Amina Doumane}, \bibinfo{person}{Denis
  Kuperberg}, \bibinfo{person}{Damien Pous}, {and} \bibinfo{person}{Pierre
  Pradic}.} \bibinfo{year}{2019}\natexlab{}.
\newblock \showarticletitle{Kleene Algebra with Hypotheses}. In
  \bibinfo{booktitle}{\emph{Foundations of Software Science and Computation
  Structures ({FOSSACS})}} \emph{(\bibinfo{series}{LNCS},
  Vol.~\bibinfo{volume}{11425})}. \bibinfo{pages}{207--223}.
\newblock


\bibitem[\protect\citeauthoryear{D’Osualdo, Farzan, and Dreyer}{D’Osualdo
  et~al\mbox{.}}{2022}]%
        {DosualdoFD2022}
\bibfield{author}{\bibinfo{person}{Emanuele D’Osualdo},
  \bibinfo{person}{Azadeh Farzan}, {and} \bibinfo{person}{Derek Dreyer}.}
  \bibinfo{year}{2022}\natexlab{}.
\newblock \showarticletitle{Proving Hypersafety Compositionally}.
\newblock \bibinfo{journal}{\emph{Proc. ACM Program. Lang.}}
  \bibinfo{volume}{6}, \bibinfo{number}{OOPSLA2}, Article
  \bibinfo{articleno}{135} (\bibinfo{year}{2022}),
  \bibinfo{numpages}{26}~pages.
\newblock
\urldef\tempurl%
\url{https://doi.org/10.1145/3563298}
\showDOI{\tempurl}


\bibitem[\protect\citeauthoryear{Eilers, M{\"{u}}ller, and Hitz}{Eilers
  et~al\mbox{.}}{2018}]%
        {EilersMH18}
\bibfield{author}{\bibinfo{person}{Marco Eilers}, \bibinfo{person}{Peter
  M{\"{u}}ller}, {and} \bibinfo{person}{Samuel Hitz}.}
  \bibinfo{year}{2018}\natexlab{}.
\newblock \showarticletitle{Modular Product Programs}. In
  \bibinfo{booktitle}{\emph{European Symposium on Programming}}.
\newblock


\bibitem[\protect\citeauthoryear{Farzan and Vandikas}{Farzan and
  Vandikas}{2019}]%
        {FarzanV2019}
\bibfield{author}{\bibinfo{person}{Azadeh Farzan} {and}
  \bibinfo{person}{Anthony Vandikas}.} \bibinfo{year}{2019}\natexlab{}.
\newblock \showarticletitle{Automated hypersafety verification}. In
  \bibinfo{booktitle}{\emph{Computer Aided Verification}}.
  \bibinfo{pages}{200--218}.
\newblock


\bibitem[\protect\citeauthoryear{Floyd}{Floyd}{1967}]%
        {Floyd67}
\bibfield{author}{\bibinfo{person}{Robert Floyd}.}
  \bibinfo{year}{1967}\natexlab{}.
\newblock \showarticletitle{Assigning Meaning to Programs}. In
  \bibinfo{booktitle}{\emph{Symp. on Applied Math. 19, Math. Aspects of Comp.
  Sci.}} \bibinfo{publisher}{Amer. Math. Soc.}, \bibinfo{pages}{19--32}.
\newblock


\bibitem[\protect\citeauthoryear{Francez}{Francez}{1983}]%
        {Francez83}
\bibfield{author}{\bibinfo{person}{Nissim Francez}.}
  \bibinfo{year}{1983}\natexlab{}.
\newblock \showarticletitle{Product Properties and Their Direct Verification}.
\newblock \bibinfo{journal}{\emph{Acta Informatica}}  \bibinfo{volume}{20}
  (\bibinfo{year}{1983}), \bibinfo{pages}{329--344}.
\newblock


\bibitem[\protect\citeauthoryear{Freyd and Scedrov}{Freyd and Scedrov}{1990}]%
        {Freyd:Scedrov}
\bibfield{author}{\bibinfo{person}{Peter~J. Freyd} {and} \bibinfo{person}{Andre
  Scedrov}.} \bibinfo{year}{1990}\natexlab{}.
\newblock \bibinfo{booktitle}{\emph{Categories, Allegories}}.
\newblock \bibinfo{publisher}{North-Holland}.
\newblock


\bibitem[\protect\citeauthoryear{Frumin, Krebbers, and Birkedal}{Frumin
  et~al\mbox{.}}{2018}]%
        {FruminKB18}
\bibfield{author}{\bibinfo{person}{Dan Frumin}, \bibinfo{person}{Robbert
  Krebbers}, {and} \bibinfo{person}{Lars Birkedal}.}
  \bibinfo{year}{2018}\natexlab{}.
\newblock \showarticletitle{{ReLoC}: {A} Mechanised Relational Logic for
  Fine-Grained Concurrency}. In \bibinfo{booktitle}{\emph{IEEE Symp. on Logic
  in Computer Science}}. \bibinfo{pages}{442--451}.
\newblock


\bibitem[\protect\citeauthoryear{Frumin, Krebbers, and Birkedal}{Frumin
  et~al\mbox{.}}{2020}]%
        {FruminKB20}
\bibfield{author}{\bibinfo{person}{Dan Frumin}, \bibinfo{person}{Robbert
  Krebbers}, {and} \bibinfo{person}{Lars Birkedal}.}
  \bibinfo{year}{2020}\natexlab{}.
\newblock \showarticletitle{{ReLoC} Reloaded: {A} Mechanized Relational Logic
  for Fine-Grained Concurrency and Logical Atomicity}.
\newblock \bibinfo{journal}{\emph{CoRR}}  \bibinfo{volume}{abs/2006.13635}
  (\bibinfo{year}{2020}).
\newblock
\showeprint[arXiv]{2006.13635}
\urldef\tempurl%
\url{https://arxiv.org/abs/2006.13635}
\showURL{%
\tempurl}


\bibitem[\protect\citeauthoryear{Godlin and Strichman}{Godlin and
  Strichman}{2008}]%
        {GodlinS08}
\bibfield{author}{\bibinfo{person}{Benny Godlin} {and} \bibinfo{person}{Ofer
  Strichman}.} \bibinfo{year}{2008}\natexlab{}.
\newblock \showarticletitle{Inference rules for proving the equivalence of
  recursive procedures}.
\newblock \bibinfo{journal}{\emph{Acta Inf.}} \bibinfo{volume}{45},
  \bibinfo{number}{6} (\bibinfo{year}{2008}), \bibinfo{pages}{403--439}.
\newblock


\bibitem[\protect\citeauthoryear{Goyal, Azeem, Madhukar, and Venkatesh}{Goyal
  et~al\mbox{.}}{2021}]%
        {Goyal2021}
\bibfield{author}{\bibinfo{person}{Manish Goyal}, \bibinfo{person}{Muqsit
  Azeem}, \bibinfo{person}{Kumar Madhukar}, {and} \bibinfo{person}{R.
  Venkatesh}.} \bibinfo{year}{2021}\natexlab{}.
\newblock \bibinfo{title}{Direct Construction of Program Alignment Automata for
  Equivalence Checking}.
\newblock
\newblock
\urldef\tempurl%
\url{https://doi.org/10.48550/ARXIV.2109.01864}
\showDOI{\tempurl}


\bibitem[\protect\citeauthoryear{Greenberg, Beckett, and Campbell}{Greenberg
  et~al\mbox{.}}{2022}]%
        {BeckettCG22}
\bibfield{author}{\bibinfo{person}{Michael Greenberg}, \bibinfo{person}{Ryan
  Beckett}, {and} \bibinfo{person}{Eric~Hayden Campbell}.}
  \bibinfo{year}{2022}\natexlab{}.
\newblock \showarticletitle{Kleene algebra modulo theories: a framework for
  concrete KATs}. In \bibinfo{booktitle}{\emph{{PLDI}}}.
  \bibinfo{pages}{594--608}.
\newblock
\urldef\tempurl%
\url{https://doi.org/10.1145/3519939.3523722}
\showDOI{\tempurl}


\bibitem[\protect\citeauthoryear{Gäher, Sammler, Spies, Jung, Dang, Krebbers,
  Kang, and Dreyer}{Gäher et~al\mbox{.}}{2022}]%
        {GaeherEtal}
\bibfield{author}{\bibinfo{person}{Lennard Gäher}, \bibinfo{person}{Michael
  Sammler}, \bibinfo{person}{Simon Spies}, \bibinfo{person}{Ralf Jung},
  \bibinfo{person}{Hoang-Hai Dang}, \bibinfo{person}{Robbert Krebbers},
  \bibinfo{person}{Jeehoon Kang}, {and} \bibinfo{person}{Derek Dreyer}.}
  \bibinfo{year}{2022}\natexlab{}.
\newblock \showarticletitle{Simuliris: a separation logic framework for
  verifying concurrent program optimizations}.
\newblock \bibinfo{journal}{\emph{Proc. {ACM} Program. Lang.}}
  \bibinfo{volume}{6}, \bibinfo{number}{{POPL}} (\bibinfo{year}{2022}).
\newblock


\bibitem[\protect\citeauthoryear{Hawblitzel, Kawaguchi, Lahiri, and
  Reb{\^{e}}lo}{Hawblitzel et~al\mbox{.}}{2013}]%
        {hawblitzelklr13}
\bibfield{author}{\bibinfo{person}{Chris Hawblitzel}, \bibinfo{person}{Ming
  Kawaguchi}, \bibinfo{person}{Shuvendu~K. Lahiri}, {and}
  \bibinfo{person}{Henrique Reb{\^{e}}lo}.} \bibinfo{year}{2013}\natexlab{}.
\newblock \showarticletitle{Towards Modularly Comparing Programs Using
  Automated Theorem Provers}. In \bibinfo{booktitle}{\emph{{CADE}}}.
  \bibinfo{pages}{282--299}.
\newblock


\bibitem[\protect\citeauthoryear{Hoare, Hayes, Jifeng, Morgan, Roscoe, Sanders,
  Sorensen, Spivey, and Sufrin}{Hoare et~al\mbox{.}}{1987}]%
        {Hoare:laws}
\bibfield{author}{\bibinfo{person}{C.~A.~R. Hoare}, \bibinfo{person}{I.J.
  Hayes}, \bibinfo{person}{He Jifeng}, \bibinfo{person}{C.C. Morgan},
  \bibinfo{person}{A.W. Roscoe}, \bibinfo{person}{J.W. Sanders},
  \bibinfo{person}{I.H. Sorensen}, \bibinfo{person}{J.M. Spivey}, {and}
  \bibinfo{person}{B.A. Sufrin}.} \bibinfo{year}{1987}\natexlab{}.
\newblock \showarticletitle{Laws of Programming}.
\newblock \bibinfo{journal}{\emph{Commun. ACM}}  \bibinfo{volume}{30}
  (\bibinfo{year}{1987}), \bibinfo{pages}{672--686,770}.
\newblock


\bibitem[\protect\citeauthoryear{Hoare, van Staden, M{\"{o}}ller, Struth, and
  Zhu}{Hoare et~al\mbox{.}}{2016}]%
        {HoareSMSZ16}
\bibfield{author}{\bibinfo{person}{Tony Hoare}, \bibinfo{person}{Stephan van
  Staden}, \bibinfo{person}{Bernhard M{\"{o}}ller}, \bibinfo{person}{Georg
  Struth}, {and} \bibinfo{person}{Huibiao Zhu}.}
  \bibinfo{year}{2016}\natexlab{}.
\newblock \showarticletitle{Developments in concurrent {Kleene} algebra}.
\newblock \bibinfo{journal}{\emph{J. Log. Algebraic Methods Program.}}
  \bibinfo{volume}{85}, \bibinfo{number}{4} (\bibinfo{year}{2016}),
  \bibinfo{pages}{617--636}.
\newblock
\urldef\tempurl%
\url{https://doi.org/10.1016/j.jlamp.2015.09.012}
\showDOI{\tempurl}


\bibitem[\protect\citeauthoryear{H{\"{o}}fner, Pous, and Struth}{H{\"{o}}fner
  et~al\mbox{.}}{2019}]%
        {HofnerPS19}
\bibfield{author}{\bibinfo{person}{Peter H{\"{o}}fner}, \bibinfo{person}{Damien
  Pous}, {and} \bibinfo{person}{Georg Struth}.}
  \bibinfo{year}{2019}\natexlab{}.
\newblock \showarticletitle{Relational and algebraic methods in computer
  science}.
\newblock \bibinfo{journal}{\emph{J. Log. Algebraic Methods Program.}}
  \bibinfo{volume}{106} (\bibinfo{year}{2019}), \bibinfo{pages}{198--199}.
\newblock
\urldef\tempurl%
\url{https://doi.org/10.1016/j.jlamp.2019.05.005}
\showDOI{\tempurl}


\bibitem[\protect\citeauthoryear{Kov{\'a}cs, Seidl, and Finkbeiner}{Kov{\'a}cs
  et~al\mbox{.}}{2013}]%
        {KovacsSF13}
\bibfield{author}{\bibinfo{person}{M{\'a}t{\'e} Kov{\'a}cs},
  \bibinfo{person}{Helmut Seidl}, {and} \bibinfo{person}{Bernd Finkbeiner}.}
  \bibinfo{year}{2013}\natexlab{}.
\newblock \showarticletitle{Relational abstract interpretation for the
  verification of 2-hypersafety properties}. In \bibinfo{booktitle}{\emph{ACM
  Computer and Communications Security}}.
\newblock


\bibitem[\protect\citeauthoryear{Kozen}{Kozen}{1997}]%
        {Kozen1997}
\bibfield{author}{\bibinfo{person}{Dexter Kozen}.}
  \bibinfo{year}{1997}\natexlab{}.
\newblock \showarticletitle{Kleene algebra with tests}.
\newblock \bibinfo{journal}{\emph{ACM Transactions on Programming Languages and
  Systems}} \bibinfo{volume}{19}, \bibinfo{number}{3} (\bibinfo{year}{1997}),
  \bibinfo{pages}{427--443}.
\newblock


\bibitem[\protect\citeauthoryear{Kozen}{Kozen}{2000}]%
        {Kozen00}
\bibfield{author}{\bibinfo{person}{Dexter Kozen}.}
  \bibinfo{year}{2000}\natexlab{}.
\newblock \showarticletitle{On {H}oare logic and {K}leene algebra with tests}.
\newblock \bibinfo{journal}{\emph{{ACM} Trans. Comput. Log.}}
  \bibinfo{volume}{1}, \bibinfo{number}{1} (\bibinfo{year}{2000}),
  \bibinfo{pages}{60--76}.
\newblock


\bibitem[\protect\citeauthoryear{Kozen}{Kozen}{2003}]%
        {Kozen2003}
\bibfield{author}{\bibinfo{person}{Dexter Kozen}.}
  \bibinfo{year}{2003}\natexlab{}.
\newblock \bibinfo{booktitle}{\emph{Kleene algebra with tests and the static
  analysis of programs}}.
\newblock \bibinfo{type}{{T}echnical {R}eport}. \bibinfo{institution}{Cornell
  University}.
\newblock


\bibitem[\protect\citeauthoryear{Kozen}{Kozen}{2004}]%
        {Kozen04}
\bibfield{author}{\bibinfo{person}{Dexter Kozen}.}
  \bibinfo{year}{2004}\natexlab{}.
\newblock \showarticletitle{Some results in dynamic model theory}.
\newblock \bibinfo{journal}{\emph{Sci. Comput. Program.}} \bibinfo{volume}{51},
  \bibinfo{number}{1-2} (\bibinfo{year}{2004}), \bibinfo{pages}{3--22}.
\newblock


\bibitem[\protect\citeauthoryear{Kozen and Smith}{Kozen and Smith}{1996}]%
        {Kozen1996}
\bibfield{author}{\bibinfo{person}{Dexter Kozen} {and}
  \bibinfo{person}{Frederick Smith}.} \bibinfo{year}{1996}\natexlab{}.
\newblock \showarticletitle{Kleene algebra with tests: Completeness and
  decidability}. In \bibinfo{booktitle}{\emph{International Workshop on
  Computer Science Logic}}. Springer, \bibinfo{pages}{244--259}.
\newblock


\bibitem[\protect\citeauthoryear{Lamport and Schneider}{Lamport and
  Schneider}{2021}]%
        {LamportS21}
\bibfield{author}{\bibinfo{person}{Leslie Lamport} {and}
  \bibinfo{person}{Fred~B. Schneider}.} \bibinfo{year}{2021}\natexlab{}.
\newblock \showarticletitle{Verifying Hyperproperties With {TLA}}. In
  \bibinfo{booktitle}{\emph{IEEE Computer Security Foundations Symposium
  {(CSF)}}}. \bibinfo{pages}{1--16}.
\newblock


\bibitem[\protect\citeauthoryear{Maillard, Hrit\c{c}u, Rivas, and
  Muylder}{Maillard et~al\mbox{.}}{2020}]%
        {MaillardHRM20}
\bibfield{author}{\bibinfo{person}{Kenji Maillard},
  \bibinfo{person}{C\u{a}t\u{a}lin Hrit\c{c}u}, \bibinfo{person}{Exequiel
  Rivas}, {and} \bibinfo{person}{Antoine~Van Muylder}.}
  \bibinfo{year}{2020}\natexlab{}.
\newblock \showarticletitle{The Next 700 Relational Program Logics}.
\newblock \bibinfo{journal}{\emph{Proc. {ACM} Program. Lang.}}
  \bibinfo{volume}{4}, \bibinfo{number}{{POPL}} (\bibinfo{year}{2020}),
  \bibinfo{pages}{4:1--4:33}.
\newblock


\bibitem[\protect\citeauthoryear{Mamouras}{Mamouras}{2017}]%
        {Mamouras17}
\bibfield{author}{\bibinfo{person}{Konstantinos Mamouras}.}
  \bibinfo{year}{2017}\natexlab{}.
\newblock \showarticletitle{Equational Theories of Abnormal Termination Based
  on {Kleene} Algebra}. In \bibinfo{booktitle}{\emph{{FoSSaCS}}}.
  \bibinfo{pages}{88--105}.
\newblock


\bibitem[\protect\citeauthoryear{Mordvinov and Fedyukovich}{Mordvinov and
  Fedyukovich}{2019}]%
        {Mordvinov2019}
\bibfield{author}{\bibinfo{person}{Dmitry Mordvinov} {and}
  \bibinfo{person}{Grigory Fedyukovich}.} \bibinfo{year}{2019}\natexlab{}.
\newblock \showarticletitle{Property Directed Inference of Relational
  Invariants}. In \bibinfo{booktitle}{\emph{Formal Methods in Computer Aided
  Design}}. \bibinfo{pages}{152--160}.
\newblock
\urldef\tempurl%
\url{https://doi.org/10.23919/FMCAD.2019.8894274}
\showDOI{\tempurl}


\bibitem[\protect\citeauthoryear{Morgan}{Morgan}{1988}]%
        {MorganAux}
\bibfield{author}{\bibinfo{person}{Carroll Morgan}.}
  \bibinfo{year}{1988}\natexlab{}.
\newblock \showarticletitle{Auxiliary Variables in Data Refinement}.
\newblock \bibinfo{journal}{\emph{Inform. Process. Lett.}}
  \bibinfo{volume}{29}, \bibinfo{number}{6} (\bibinfo{year}{1988}),
  \bibinfo{pages}{293--296}.
\newblock


\bibitem[\protect\citeauthoryear{Murray}{Murray}{2020}]%
        {MurrayUnderApproxRHL}
\bibfield{author}{\bibinfo{person}{Toby Murray}.}
  \bibinfo{year}{2020}\natexlab{}.
\newblock \showarticletitle{An Under-Approximate Relational Logic: Heralding
  Logics of Insecurity, Incorrect Implementation {\&} More}.
\newblock \bibinfo{journal}{\emph{CoRR}}  \bibinfo{volume}{abs/2003.04791}
  (\bibinfo{year}{2020}).
\newblock
\showeprint[arXiv]{2003.04791}
\urldef\tempurl%
\url{https://arxiv.org/abs/2003.04791}
\showURL{%
\tempurl}


\bibitem[\protect\citeauthoryear{Nagasamudram and Naumann}{Nagasamudram and
  Naumann}{2021}]%
        {NagasamudramN21}
\bibfield{author}{\bibinfo{person}{Ramana Nagasamudram} {and}
  \bibinfo{person}{David~A. Naumann}.} \bibinfo{year}{2021}\natexlab{}.
\newblock \showarticletitle{Alignment Completeness for Relational {Hoare}
  Logics}. In \bibinfo{booktitle}{\emph{IEEE Symp. on Logic in Computer
  Science}}. \bibinfo{pages}{1--13}.
\newblock
\newblock
\shownote{Extended version at \url{https://arxiv.org/abs/2101.11730}.}.


\bibitem[\protect\citeauthoryear{Naumann}{Naumann}{2020}]%
        {NaumannISOLA20}
\bibfield{author}{\bibinfo{person}{David~A. Naumann}.}
  \bibinfo{year}{2020}\natexlab{}.
\newblock \showarticletitle{Thirty-Seven Years of Relational Hoare Logic:
  Remarks on Its Principles and History}. In \bibinfo{booktitle}{\emph{9th
  International Symposium On Leveraging Applications of Formal Methods,
  Verification and Validation}}. \bibinfo{pages}{93--116}.
\newblock
\urldef\tempurl%
\url{https://doi.org/10.1007/978-3-030-61470-6\_7}
\showDOI{\tempurl}
\newblock
\shownote{Extended version at \url{https://arxiv.org/abs/2007.06421}.}.


\bibitem[\protect\citeauthoryear{O'Hearn}{O'Hearn}{2019}]%
        {OHearn2019}
\bibfield{author}{\bibinfo{person}{Peter~W O'Hearn}.}
  \bibinfo{year}{2019}\natexlab{}.
\newblock \showarticletitle{Incorrectness logic}.
\newblock \bibinfo{journal}{\emph{Proceedings of the ACM on Programming
  Languages}} \bibinfo{volume}{4}, \bibinfo{number}{POPL}
  (\bibinfo{year}{2019}), \bibinfo{pages}{1--32}.
\newblock


\bibitem[\protect\citeauthoryear{Pick, Fedyukovich, and Gupta}{Pick
  et~al\mbox{.}}{2018}]%
        {PickFG18}
\bibfield{author}{\bibinfo{person}{Lauren Pick}, \bibinfo{person}{Grigory
  Fedyukovich}, {and} \bibinfo{person}{Aarti Gupta}.}
  \bibinfo{year}{2018}\natexlab{}.
\newblock \showarticletitle{Exploiting Synchrony and Symmetry in Relational
  Verification}. In \bibinfo{booktitle}{\emph{Computer Aided Verification}}.
  \bibinfo{pages}{164--182}.
\newblock


\bibitem[\protect\citeauthoryear{Pous}{Pous}{[n.\,d.]}]%
        {PousKATlib}
\bibfield{author}{\bibinfo{person}{Damien Pous}.}
  \bibinfo{year}{[n.\,d.]}\natexlab{}.
\newblock \bibinfo{title}{Relation Algebra and {KAT} in {Coq}}.
\newblock
\newblock
\urldef\tempurl%
\url{http://perso.ens-lyon.fr/damien.pous/ra/}
\showURL{%
\tempurl}
\newblock
\shownote{Coq library, accessed July 2022}.


\bibitem[\protect\citeauthoryear{Pous}{Pous}{2015}]%
        {Pous2015}
\bibfield{author}{\bibinfo{person}{Damien Pous}.}
  \bibinfo{year}{2015}\natexlab{}.
\newblock \showarticletitle{Symbolic algorithms for language equivalence and
  Kleene algebra with tests}. In \bibinfo{booktitle}{\emph{ACM Symposium on
  Principles of Programming Languages}}. \bibinfo{pages}{357--368}.
\newblock


\bibitem[\protect\citeauthoryear{Pous, Rot, and Wagemaker}{Pous
  et~al\mbox{.}}{2021}]%
        {PousRW21}
\bibfield{author}{\bibinfo{person}{Damien Pous}, \bibinfo{person}{Jurriaan
  Rot}, {and} \bibinfo{person}{Jana Wagemaker}.}
  \bibinfo{year}{2021}\natexlab{}.
\newblock \showarticletitle{On Tools for Completeness of Kleene Algebra with
  Hypotheses}. In \bibinfo{booktitle}{\emph{Relational and Algebraic Methods in
  Computer Science ({RAMiCS})}} \emph{(\bibinfo{series}{LNCS},
  Vol.~\bibinfo{volume}{13027})}. \bibinfo{pages}{378--395}.
\newblock


\bibitem[\protect\citeauthoryear{Sabelfeld and Myers}{Sabelfeld and
  Myers}{2003}]%
        {Sabelfeld:Myers:JSAC}
\bibfield{author}{\bibinfo{person}{Andrei Sabelfeld} {and}
  \bibinfo{person}{Andrew~C. Myers}.} \bibinfo{year}{2003}\natexlab{}.
\newblock \showarticletitle{Language-Based Information-Flow Security}.
\newblock \bibinfo{journal}{\emph{IEEE J. Selected Areas in Communications}}
  \bibinfo{volume}{21}, \bibinfo{number}{1} (\bibinfo{date}{Jan.}
  \bibinfo{year}{2003}), \bibinfo{pages}{5--19}.
\newblock


\bibitem[\protect\citeauthoryear{Sharma, Schkufza, Churchill, and Aiken}{Sharma
  et~al\mbox{.}}{2013}]%
        {Sharma2013}
\bibfield{author}{\bibinfo{person}{Rahul Sharma}, \bibinfo{person}{Eric
  Schkufza}, \bibinfo{person}{Berkeley Churchill}, {and} \bibinfo{person}{Alex
  Aiken}.} \bibinfo{year}{2013}\natexlab{}.
\newblock \showarticletitle{Data-Driven Equivalence Checking}. In
  \bibinfo{booktitle}{\emph{ACM Conf. on Object-Oriented Programming Languages,
  Systems, and Applications}}. \bibinfo{pages}{391–406}.
\newblock
\urldef\tempurl%
\url{https://doi.org/10.1145/2509136.2509509}
\showDOI{\tempurl}


\bibitem[\protect\citeauthoryear{Shemer, Gurfinkel, Shoham, and Vizel}{Shemer
  et~al\mbox{.}}{2019}]%
        {ShemerGSV19}
\bibfield{author}{\bibinfo{person}{Ron Shemer}, \bibinfo{person}{Arie
  Gurfinkel}, \bibinfo{person}{Sharon Shoham}, {and} \bibinfo{person}{Yakir
  Vizel}.} \bibinfo{year}{2019}\natexlab{}.
\newblock \showarticletitle{Property directed self composition}. In
  \bibinfo{booktitle}{\emph{International Conference on Computer Aided
  Verification}}. Springer, \bibinfo{pages}{161--179}.
\newblock


\bibitem[\protect\citeauthoryear{Sousa and Dillig}{Sousa and Dillig}{2016}]%
        {SousaD2016}
\bibfield{author}{\bibinfo{person}{Marcelo Sousa} {and} \bibinfo{person}{Isil
  Dillig}.} \bibinfo{year}{2016}\natexlab{}.
\newblock \showarticletitle{Cartesian {Hoare} logic for verifying k-safety
  properties}. In \bibinfo{booktitle}{\emph{Proceedings of the 37th ACM SIGPLAN
  Conference on Programming Language Design and Implementation}}.
  \bibinfo{pages}{57--69}.
\newblock


\bibitem[\protect\citeauthoryear{Terauchi and Aiken}{Terauchi and
  Aiken}{2005}]%
        {TerauchiA2005}
\bibfield{author}{\bibinfo{person}{Tachio Terauchi} {and} \bibinfo{person}{Alex
  Aiken}.} \bibinfo{year}{2005}\natexlab{}.
\newblock \showarticletitle{Secure information flow as a safety problem}. In
  \bibinfo{booktitle}{\emph{International Static Analysis Symposium}}.
  Springer, \bibinfo{pages}{352--367}.
\newblock


\bibitem[\protect\citeauthoryear{Unno, Terauchi, and Koskinen}{Unno
  et~al\mbox{.}}{2021}]%
        {UnnoTK21}
\bibfield{author}{\bibinfo{person}{Hiroshi Unno}, \bibinfo{person}{Tachio
  Terauchi}, {and} \bibinfo{person}{Eric Koskinen}.}
  \bibinfo{year}{2021}\natexlab{}.
\newblock \showarticletitle{Constraint-Based Relational Verification}. In
  \bibinfo{booktitle}{\emph{Computer Aided Verification}}.
  \bibinfo{pages}{742--766}.
\newblock


\bibitem[\protect\citeauthoryear{Wagemaker, Bonsangue, Kapp{\'{e}}, Rot, and
  Silva}{Wagemaker et~al\mbox{.}}{2019}]%
        {WagemakerBKR019}
\bibfield{author}{\bibinfo{person}{Jana Wagemaker},
  \bibinfo{person}{Marcello~M. Bonsangue}, \bibinfo{person}{Tobias
  Kapp{\'{e}}}, \bibinfo{person}{Jurriaan Rot}, {and}
  \bibinfo{person}{Alexandra Silva}.} \bibinfo{year}{2019}\natexlab{}.
\newblock \showarticletitle{Completeness and Incompleteness of Synchronous
  {Kleene} Algebra}. In \bibinfo{booktitle}{\emph{Mathematics of Program
  Construction}} \emph{(\bibinfo{series}{LNCS}, Vol.~\bibinfo{volume}{11825})}.
  \bibinfo{pages}{385--413}.
\newblock


\bibitem[\protect\citeauthoryear{Yang}{Yang}{2007}]%
        {Yang:tcs04}
\bibfield{author}{\bibinfo{person}{Hongseok Yang}.}
  \bibinfo{year}{2007}\natexlab{}.
\newblock \showarticletitle{Relational Separation Logic}.
\newblock \bibinfo{journal}{\emph{Theo. Comp. Sci.}}  \bibinfo{volume}{375}
  (\bibinfo{year}{2007}).
\newblock


\bibitem[\protect\citeauthoryear{Zhang, de~Amorim, and Gaboardi}{Zhang
  et~al\mbox{.}}{2022}]%
        {ZhangAG22}
\bibfield{author}{\bibinfo{person}{Cheng Zhang},
  \bibinfo{person}{Arthur~Azevedo de Amorim}, {and} \bibinfo{person}{Marco
  Gaboardi}.} \bibinfo{year}{2022}\natexlab{}.
\newblock \showarticletitle{On incorrectness logic and {Kleene} algebra with
  top and tests}.
\newblock \bibinfo{journal}{\emph{Proc. {ACM} Program. Lang.}}
  \bibinfo{volume}{6}, \bibinfo{number}{{POPL}} (\bibinfo{year}{2022}),
  \bibinfo{pages}{1--30}.
\newblock
\urldef\tempurl%
\url{https://doi.org/10.1145/3498690}
\showDOI{\tempurl}


\end{thebibliography}

\ifarxiv
\else
\end{document} % END HERE TO OMIT APPENDIX
\fi

\vfill
\pagebreak
%\onecolumn % this will be needed in the two-column submission
\section{Appendix}\label{sec:appendix}

\subsection{Some KAT Lemmas}

This section proves some results that hold in any KAT, 
including useful consequences of commutativity assumptions.
We begin with a few well known and easily proved facts for which proofs are omitted.

\begin{lemma}[invariance]\label{lem:invar}
Let $p$ be a test, then 
$p;x \leq p;x;p \imp p; x^* \leq p;x^*;p$.
\end{lemma}
%% Proof:
%% By the induction axiom $b+xa\leq x \imp b a^* \leq x$,
%% the consequent $p; x^* \leq p;x^*;p$ follows from 
%% $p + p;x^*;p;x \leq p; x^*; p$.  This in turn follows from the antecedent 
%% because 
%% \[\begin{array}{lll}
%%     & p + p;x^*;p;x \\
%% \leq& p + p;x^*;p;x;p & \mbox{assumption} \\
%%  =  & p;1;p + p;x^*;p;x;p & \mbox{$p$ test} \\
%%  =  & p;(1 + x^*;p;x);p & \mbox{distrib} \\
%% \leq& p;(1 + x^*;x);p & \mbox{test $p\leq 1$, mono} \\
%% =   & p;x^*;p & \mbox{expansion law} 
                   %%      \end{array}\]

\begin{lemma}[backward invariance]\label{lem:binvarEq}
Let $p$ be a test, then
$x;p = p;x;p \imp x^*;p = p;x^*;p$.
\end{lemma}

\begin{lemma}\label{lem:binvar}
Let $p$ be a test, then 
$x;p \leq p;x;p \imp x^*;p \leq p;(x;p)^*;p$.
\end{lemma}
%% Proof:
%% By the induction axiom $b+ax\leq x \imp a^* b \leq x$,
%% and join property
%% the consequent $x^* p \leq p (x p)^* p$  follows from 
%% $p \leq p (x p)^* p$ and 
%% $x p (x p)^* p \leq p (x p)^* p$.
%% We have $p \leq p (x p)^* p$ by star unfolding and tests idempotent.
%% For the other, we have 
%% $x p (x p)^* p 
%% \leq p x p (x p)^* p 
%% \leq p (x p)^* p$
%% using assumption 
%% $x;p \leq p;x;p$
%% and $(x p)(x p)^* \leq (x p)^*$.

\begin{lemma}\label{lem:L2}
If $x;y \leq x;y;x$ then $x;y^\kstar \leq x;(y;x)^\kstar$.
\end{lemma}
%% NOTE: We had used this in the form 
%% $a;x = a;x;a \imp a;x^{\kstar} =  a;(x;a)^{\kstar}$
%% but I don't think the reverse inequality holds even with stronger assumption $a;x=a;x;a$.
%% And we don't need the stronger one.
%% Proof:
%% \[\begin{array}{lll}
%%       & x y^\kstar \leq x (y x)^\kstar \\
%% \impby& x + x (y x)^\kstar y \leq x (y x)^\kstar & \mbox{*-induction} \\
%% \impby& x (y x)^\kstar y \leq x (y x)^\kstar & \mbox{join, and $x \leq x (y x)^\kstar$}\\
%% \iff  & x y + x (y x)^\kstar y x y \leq x (y x)^\kstar & \mbox{*-unfold} \\
%% \impby& x y x + x (y x)^\kstar y x y x \leq x (y x)^\kstar & \mbox{use $x y \leq x y x$ twice}\\
%% \end{array}\]
%% The last line follows by *-unfold.

Kozen~\cite{Kozen1997} uses the following in the case that $w$ is a test.
\begin{lemma}\label{lem:L1}
If $x;w=w;x$ and $w;w=w$ then
$x^*;w = (x;w)^*;w$. 
\end{lemma}
%% Proof: We have 
%% $(x;w)^*;w \leq x^*;w$ by tests below 1. 
%% By induction rule $b+ay\leq y \imp a^*b\leq y$,
%% the reverse inequality 
%% $x^*;w \leq (x ;w)^*;w $
%% follows from 
%% $w + x;(x;w)^*;w \leq (x;w)^*;w$
%% and we have
%% \[\begin{array}{lll}
%%    & w + x (x w)^* w \\
%% = & w + x (wx )^* w &\mbox{commute assumption $xw=wx$}\\
%% = & w + (x w)^*x w &\mbox{sliding}\\
%% = & w + (x w)^*x w w &\mbox{idempotence assumption $w w = w$}\\
%% = & (1 + (x w)^* x w) w &\mbox{distrib}\\
%% = & (x w)^*w &\mbox{star unfold}
%% \end{array}\]

\begin{lemma}\label{lem:havocish}
If $1\leq z$ and $z;z \leq z$ then $x\leq y;z \iff x;z \leq y;z$
\end{lemma}
%% Proof by mutual implication.
%% \[\begin{array}{lll}
%%      & x \leq y;z \\
%% \imp & x;z \leq y;z;z &\mbox{monot}\\
%% \imp & x;z \leq y;z &\mbox{assumption $z;z\leq z$}\\
%% \imp & x \leq y;z &\mbox{assumption $1\leq z$, transitivity}
%% \end{array}
%% \]

\paragraph{KAT and Commutativity}
There are several useful laws for iteration involving two terms that commute.
These hold in any KAT.   

\begin{lemma}\label{lem:item-over-star}
$w;z = z;w \imp w^*;z = z;w^*$
\end{lemma}
Proof by mutual inclusion.  
For $w^*;z \leq z;w^*$ we have
\[\begin{array}{lll}
       &  w^*;z \leq z;w^* \\
\impby & z + w;z;w^* \leq z;w^* & \mbox{ind $b+ax\leq x \imp a^*b\leq x$}\\
\iff   & z \leq z;w^* \land w;z;w^* \leq z;w^* & \mbox{join} \\
\iff   & z \leq z;w^* \land z;w;w^* \leq z;w^* & \mbox{$w;z=z;w$} 
\end{array}\]
The two conjuncts hold by the star laws 
$1\leq x^*$ and $x;x^*\leq x^*$ 
and monotonicity of $;$.
The reverse inclusion is proved similarly, using the other induction rule
$b+xa\leq x \imp b a^* \leq x$.  

\begin{lemma}\label{lem:star-over-plus}
$w;z = z;w \imp (w+z)^* = w^*;z^*$.
\end{lemma}
Proof by mutual inclusion.
One inclusion does not require the commutativity assumption:
\[\begin{array}{lll}
     &  w^*;z^* \\
\leq & (w+z)^*;(w+z)^* & \mbox{$x\leq x+y$, star mono} \\
=    & (w+z)^* & \mbox{star law $x^*;x^* = x^*$}
\end{array}\]
For the reverse we have:
\[\begin{array}{lll}
     &  (w+z)^* \leq w^*;z^* \\
\iff &  1 + (w+z);(w+z)^* \leq w^*;z^* &\mbox{star unfold}\\
\iff &  1 \leq w^*;z^* \land (w+z);(w+z)^* \leq w^*;z^* &\mbox{$+$ join} \\
\iff &  (w+z);(w+z)^* \leq w^*;z^* &\mbox{using $1\leq x^*$} \\
\impby & w + z + w^*;z^*;(w+z) \leq w^*;z^* & \mbox{star induct}\\ %. $b+xa\leq x \imp b a^* \leq x$}\\  
\iff & w + z + w^*;z^*;w + w^*;z^*;z \leq w^*;z^* & \mbox{distrib}\\  
\iff & w  \leq w^*;z^* \land 
       z \leq w^*;z^* \land 
       w^*;z^*;w \leq w^*;z^* \land 
       w^*;z^*;z \leq w^*;z^* &\mbox{$+$ join}\\ 
\iff & w  \leq w^*;z^* \land 
       z \leq w^*;z^* \land 
       w^*;w;z^* \leq w^*;z^* \land 
       w^*;z^*;z \leq w^*;z^* & \mbox{Lemma~\ref{lem:item-over-star}}
\end{array}\]
The left two conjuncts follow using $x\leq x^*$ and $1\leq x^*$.
The right two conjuncts follow using $x^*;x\leq x^*$.

\begin{lemma}\label{lem:star-comm}
$w;z = z;w \imp w^*;z^* = z^*;w^*$
\end{lemma}
Proof: $w^*;z^* = (w+z)^* = (z+w)^* = z^*;w^*$
using Lemma~\ref{lem:star-over-plus} twice.

\begin{lemma}\label{lem:star-seq-comm}
$ w;z = z;w \imp (w;z)^* \leq w^*;z^* $
\end{lemma}
Proof:
\[\begin{array}{lll}
      & (w;z)^* \leq w^*;z^*  \\
\iff  & 1 \leq w^*;z^* \land w;z;(w;z)^* \leq w^*;z^* & \mbox{star unfold, $+$ join} \\
\iff  & w;z;(w;z)^* \leq w^*;z^* & \mbox{using $1\leq x^*$} \\
\impby& w;z + w^*;z^*;w;z \leq w^*;z^* &\mbox{ind $b + xa \leq x \imp ba^*\leq x$}\\
\iff  & w^*;z^*;w;z \leq w^*;z^* &\mbox{$+$ join, $x\leq x^*$}\\
\iff  & w^*;z^*;z;w \leq w^*;z^* &\mbox{$w;z=z;w$}\\
\impby& w^*;z^*;w \leq w^*;z^* &\mbox{$x^*x\leq x^*$}\\
\impby& w^*;w;z^* \leq w^*;z^* & \mbox{Lemma~\ref{lem:item-over-star}} \\
\impby& w^*;z^* \leq w^*;z^*     & \mbox{$x x^*\leq x^*$}
\end{array}\]
As usual the hints do not mention uses of transitivity and monotonicity.

%The next lemma is essentially the Francez product (also used in Cartesian Hoare Logic)

\begin{lemma}\label{lem:seq-star-comm}
$ w;z = z;w \imp w^*;z^* = (w;z)^* ; (w^* + z^*) $ 
\end{lemma}
Proof by mutual inclusion.
First,
\[\begin{array}{lll}
     & (w;z)^* ; (w^* + z^*) \\
\leq & w^*; z^*; (w^* + z^*) & \mbox{Lemma~\ref{lem:star-seq-comm}} \\
\leq & w^*; z^*; w^*; z^* & \mbox{$x^*+y^*\leq x^*;y^*$ by KAT} \\
=    & w^*; w^*; z^*; z^* & \mbox{Lemma~\ref{lem:star-comm}}\\ 
=    & w^*; z^* & \mbox{$x^*;x^*=x^*$ by KAT} 
\end{array}\]
For the reverse inclusion, 
\[\begin{array}{lll}
      & w^*;z^* \leq (w;z)^* ; (w^* + z^*) \\ 
\iff  & (w+z)^* \leq (w;z)^* ; (w^* + z^*) & \mbox{Lemma~\ref{lem:star-over-plus}}\\ 

\iff  & 1 \leq (w;z)^* ; (w^* + z^*) \land (w+z)^*;(w+z) \leq (w;z)^* ; (w^* + z^*) 
      & \mbox{star expansion, $+$ join}\\
\iff  & (w+z)^*;(w+z) \leq (w;z)^* ; (w^* + z^*) 
      & \mbox{using $1\leq x^*$}\\
\impby& w + z + (w+z);(w;z)^*;(w^* + z^*) \leq (w;z)^* ; (w^* + z^*) & \mbox{ind $b+ax\leq x \imp a^*b\leq x$}\\
\end{array}\]
By join property of $+$, the last line is equivalent to the conjunction of
$w + z \leq (w;z)^* ; (w^* + z^*)$, which follows easily using $x\leq x^*$,
and this: 
\[ (w+z);(w;z)^*;(w^* + z^*) \leq (w;z)^* ; (w^* + z^*) \]
Using distributivity on the left side, and the join property, 
the displayed inequality is equivalent to the conjunction of
four conditions: \\
$w;(w;z)^*;w^* \leq (w;z)^* ; (w^* + z^*)$,\\
$z;(w;z)^*;w^* \leq (w;z)^* ; (w^* + z^*)$,\\
$w;(w;z)^*;z^* \leq (w;z)^* ; (w^* + z^*)$, and \\
$z;(w;z)^*;z^* \leq (w;z)^* ; (w^* + z^*)$. \\
For the first condition:
\[\begin{array}{lll}
  & w;(w;z)^*;w^*  \\ 
= & w;(z;w)^*;w^* & \mbox{assumption $w;z=z;w$}\\ 
= & (w;z)^*;w;w^* & \mbox{sliding $x(yx)^*=(xy)^*x$}\\ 
\leq & (w;z)^*;w^* & \mbox{$x;x^*\leq x^*$}\\ 
\leq & (w;z)^*;(w^*+z^*) & \mbox{$x\leq x+y$, mono}
\end{array}\]
For the second condition,
\[\begin{array}{lll}
  & z;(w;z)^*;w^*  \\ 
= & (z;w)^*;z;w^* & \mbox{sliding $x(yx)^*=(xy)^*x$}\\ 
= & (z;w)^*;z;(1+w;w^*) & \mbox{star unfold}\\ 
= & (z;w)^*;z+(z;w)^*;z;w;w^* & \mbox{distrib}\\ 
\leq & (z;w)^*;z^*+(z;w)^*;z;w;w^* & \mbox{$x\leq x^*$}\\ 
\leq & (z;w)^*;z^*+(z;w)^*;w^* & \mbox{$x^*;x\leq x^*$}\\ 
= & (z;w)^*;(z^*+w^*) & \mbox{distrib}\\ 
= & (w;z)^*;(w^*+z^*) & \mbox{assumption $w;z = z;w$}\\ 
\end{array}\]
The other two conditions are symmetric with these two.

Here's another expansion lemma that can be used to relate two loops.

\begin{lemma}\label{lem:seq-star-expansion}
$ w;z = z;w \imp w^*;z^* = (w;z + w + z)^*$
\end{lemma}
Proof by mutual inclusion.
We have $w^*;z^* \leq (w;z + w + z)^*$ 
because $w^*;z^* = (w+z)^* \leq  (w;z + w + z)^* $
by Lemma~\ref{lem:star-over-plus} and then $(y+)$ increasing, star mono.

For the reverse, we have
\[\begin{array}{lll}
    & (w;z + w + z)^* \leq w^*;z^* \\
\iff& 1 + (w;z + w + z)^*;(w;z+w+z) \leq w^*;z^* &\mbox{star expansion}\\
\iff& (w;z + w + z)^*;(w;z+w+z) \leq w^*;z^* &\mbox{using $+$ join and $1\leq x^*$}\\
\impby& w;z + w + z + (w;z + w + z);w^*;z^* \leq w^*;z^* &\mbox{ind $b+ax\leq x\imp a^*b\leq x$}
\end{array}\]
The last line can be decomposed, using join, to easy conditions
$w;z \leq w^*;z^*$, $w\leq w^*;z^*$, and $z \leq w^*;z^*$,
together with this:
\[ (w;z + w + z);w^*;z^* \leq w^*;z^* \]
which by distribution and join property is equivalent to the
three conjuncts 
\[ 
w;z;w^*;z^* \leq w^*;z^* \land 
w;w^*;z^* \leq w^*;z^* \land
z;w^*;z^* \leq w^*;z^* 
\]
The first holds because
\[\begin{array}{lll}
     & w;z;w^*;z^* \\
=    & z;w;w^*;z^* & \mbox{assumption $w;z=z;w$} \\
\leq & z;w^*;z^* & \mbox{$xx^*\leq x^*$} \\
=    & w^*;z;z^* & \mbox{Lemma~\ref{lem:item-over-star}} \\
\leq & w^*;z^* & \mbox{$xx^*\leq x^*$} 
\end{array}\]
The second and third are similar.

\subsection{BiKAT Theorems}

This section proves results that hold in any BiKAT.

\begin{lemma}%\label{lem:WOalt}
In a BiKAT we have $X;\eml{y} \leq Z;\emr{\hav}$
iff $X;\emb{y}{\hav} \leq Z;\emr{\hav}$.
\end{lemma}
Proof: By definition $X;\emb{y}{\hav}= X;\eml{y};\emr{\hav}$,
and now we can apply Lemma~\ref{lem:havocish} 
because $\emr{\hav}$ is idempotent and above 1.

\begin{lemma}%\label{lem:eq:expand}
Equation (\ref{eq:expand}) holds, i.e.,
for any $c,c'$ and tests $e,e'$ in the underlying KAT:
\[
\emb{(e;c)^*}{(e';c')^*};\emb{\neg e}{\neg e'} 
 = \emb{e;c}{e';c'}^*;(\emb{e;c}{\neg e'}^* + \emb{\neg e}{e';c'}^*) ;\emb{\neg e}{\neg e'}
\] 
\end{lemma}
Proof by mutual inclusion.
For readability we omit the embedding notations, i.e., $c$ stands for $\eml{c}$, $c'$ for $\emr{c'}$, etc.  The sequence operation is written variously by juxtaposition, semicolon, and bar $|$.
(Bar because $\emb{x}{y}$ is shorthand for the sequence $\eml{x};\emr{y}$.)

For $RHS \leq LHS$ we have
\[\begin{array}{lll}
    & (ec|e'c')^* ; ( (ec|\neg e')^* + (\neg e|e'c')^* ) (\neg e|\neg e') \\ 
\leq& (ec)^* ; (e'c')^* ; ( (ec|\neg e')^* + (\neg e|e'c')^* ) (\neg e|\neg e')  &\mbox{Lemma~\ref{lem:seq-star-comm}}\\
\leq& (ec|)^*; (|e'c')^*; (\neg e|\neg e') &\mbox{1 below star} 
\end{array}\]

For $LHS \leq RHS$ we have
\begin{small}
\[\begin{array}{lll}
    & (ec|)^*; (|e'c')^* ; (\neg e|\neg e') \leq (ec|e'c')^* ; ( (ec|\neg e')^* + (\neg e|e'c')^* ) ; (\neg e|\neg e') \\
\iff& \quad \mbox{using Lemma~\ref{lem:star-over-plus}} \\
    &( (ec|) + (|e'c') )^*; (\neg e|\neg e') \leq (ec|e'c')^* ; ( (ec|\neg e')^* + (\neg e|e'c')^* ) ; (\neg e|\neg e') \\
\impby& \quad \mbox{by ind $b + ax \leq x \imp a^*b\leq x$} \\
    & (\neg e|\neg e') + ( (ec|) + (|e'c') ); (ec|e'c')^* ; ( (ec|\neg e')^* + (\neg e|e'c')^* ) ; (\neg e|\neg e') \\
    & \leq (ec|e'c')^* ; ( (ec|\neg e')^* + (\neg e|e'c')^* ) ; (\neg e|\neg e') \\
\iff& \quad \mbox{using distrib leftmost; $+$ join} \\
    & (\neg e|\neg e') \leq (ec|e'c')^* ; ( (ec|\neg e')^* + (\neg e|e'c')^* ) ; (\neg e|\neg e') \\
    & \land \;
     ( (ec|) + (|e'c') ); (ec|e'c')^* ; ( (ec|\neg e')^* + (\neg e|e'c')^* ) ; (\neg e|\neg e') \\
& \quad \leq (ec|e'c')^* ; ( (ec|\neg e')^* + (\neg e|e'c')^* ) ; (\neg e|\neg e')
\\
\iff& \quad \mbox{first conjunct holds using $1\leq x^*$} \\
    & ( (ec|) + (|e'c') ); (ec|e'c')^* ; ( (ec|\neg e')^* + (\neg e|e'c')^* ) ; (\neg e|\neg e') \\
& \quad  \leq (ec|e'c')^* ; ( (ec|\neg e')^* + (\neg e|e'c')^* ) ; (\neg e|\neg e')
\end{array}\]
\end{small}
By distributivity the last line is equivalent to these four conditions:
\[\begin{array}{ll}
(A) & (ec|) ; (ec|e'c')^* ;  (ec|\neg e')^* ; (\neg e|\neg e') \leq  (ec|e'c')^* ; ( (ec|\neg e')^* + (\neg e|e'c')^* ) ; (\neg e|\neg e') \\
(B) & (|e'c') ; (ec|e'c')^* ; (ec|\neg e')^* ; (\neg e|\neg e') \leq  (ec|e'c')^* ; ( (ec|\neg e')^* + (\neg e|e'c')^* ) ; (\neg e|\neg e') \\
(C) & (ec|) ; (ec|e'c')^* ; (\neg e|e'c')^* ; (\neg e|\neg e') \leq  (ec|e'c')^* ; ( (ec|\neg e')^* + (\neg e|e'c')^* ) ; (\neg e|\neg e') \\ 
(D) & (|e'c') ; (ec|e'c')^* ; (\neg e|e'c')^* ; (\neg e|\neg e') \leq  (ec|e'c')^* ; ( (ec|\neg e')^* + (\neg e|e'c')^* ) ; (\neg e|\neg e')
\end{array}\]
We prove $(A)$ and $(B)$; then $(C)$ and $(D)$ follow for reasons of symmetry.

For $(A)$:
\[\begin{array}{lll}
    & (ec|) ; (ec;e'c')^* ;  (ec|\neg e')^* ; (\neg e|\neg e') \\
=   & (ec|) ; (e'c';ec)^* ;  (ec|\neg e')^* ; (\neg e|\neg e')  & \mbox{LRC (and def emb)} \\
=   & (ec;e'c')^*; ec ;  (ec|\neg e')^* ; (\neg e|\neg e') &\mbox{star sliding} \\
=   & (ec;e'c')^*; ec ; (\neg e';ec)^* ; (\neg e|\neg e') &\mbox{LRC} \\
=   & (ec;e'c')^*; \neg e';ec ; (\neg e';ec)^*  ; (\neg e|\neg e') &\mbox{test $e'$ idem, commutes with $ec$,  Lemma~\ref{lem:item-over-star}}\\
\leq & (ec;e'c')^*; (\neg e';ec)^*  ; (\neg e|\neg e') & \mbox{using $x;x^* \leq x^*$}\\
\leq & (ec;e'c')^*; ((\neg e';ec)^* + (\neg e|e'c')^*) ; (\neg e|\neg e') &\mbox{$+$ mono} 
\end{array}\]

For (B):
\begin{small}
\[\begin{array}{lll}
    & (|e'c') ; (ec|e'c')^* ; (ec|\neg e')^* ; (\neg e|\neg e') \\
=   & (e'c';ec)^*;(e'c') ; (ec|\neg e')^* ; (\neg e|\neg e') & \mbox{star sliding} \\
=   & (e'c'ec)^*;(e'c') ; (1+(ec|\neg e')(ec|\neg e')^*) ; (\neg e|\neg e') & \mbox{star expand} \\
=   & (e'c'ec)^*;(e'c') ; (\neg e|\neg e') + (e'c'ec)^*;(e'c') ; (ec|\neg e');(ec|\neg e')^* ; (\neg e|\neg e') &\mbox{distrib} \\
\leq& (e'c'ec)^*;(e'c')^* ; (\neg e|\neg e') + (e'c'ec)^*;(e'c') ; (ec|\neg e');(ec|\neg e')^* ; (\neg e|\neg e') &\mbox{using $x\leq x^*$} \\
=   & (e'c'ec)^*;(e'c')^* ; (\neg e|\neg e') + (e'c'ec)^*;(e'c'ec);\neg e' ; (ec|\neg e')^* ; (\neg e|\neg e') &\mbox{emb def, assoc} \\
\leq& (e'c'ec)^*;(e'c')^* ; (\neg e|\neg e') + (e'c'ec)^*; \neg e'; (ec|\neg e')^* ; (\neg e|\neg e') &\mbox{$x^*x\leq x$ for $x:=e'c'ec$}\\
=   & (e'c'ec)^*;(e'c')^* ; (\neg e|\neg e') + (e'c'ec)^*; (ec|\neg e')^* ; \neg e'; (\neg e|\neg e') &\mbox{sliding, $\neg e'$ comm.\ $ec$}\\
=   & (e'c'ec)^*;(e'c')^* ; (\neg e|\neg e') + (e'c'ec)^*; (ec|\neg e')^* ; (\neg e|\neg e') &\mbox{tests idempotent} \\
=   & (e'c'ec)^*; ((e'c')^* ; (\neg e|\neg e') + (ec|\neg e')^* ; (\neg e|\neg e')) &\mbox{distrib} \\
=   & (ece'c')^*; ((e'c')^* ; (\neg e|\neg e') + (ec|\neg e')^* ; (\neg e|\neg e')) &\mbox{LRC: $ece'c' = e'c'ec$}\\
=   & (ece'c')^*; ((ec|\neg e')^* ; (\neg e|\neg e') + (e'c')^* ; (\neg e|\neg e')) &\mbox{$+$ commut}\\
=   & (ece'c')^*; ((ec|\neg e')^* ; (\neg e|\neg e') + (e'c')^*; \neg e ; (\neg e|\neg e')) &\mbox{tests idempotent} \\
\leq& (ece'c')^*; ((ec|\neg e')^* ; (\neg e|\neg e') + (\neg e|e'c')^*; (\neg e|\neg e')) &\mbox{Lemma~\ref{lem:L1}} \\
=   & (ece'c')^*; ((ec|\neg e')^* + (\neg e|e'c')^*); (\neg e|\neg e')) &\mbox{distrib} \\
=   & (ec|e'c')^*; ((ec|\neg e')^* + (\neg e|e'c')^*); (\neg e|\neg e')) &\mbox{def} 
\end{array}\]
\end{small}

\begin{theoremApp}{thm:undecidability}
It is undecidable whether a given identity holds in all $*$-continuous BiKATs.
\end{theoremApp}

We present the undecidability argument from~\cite{Kozen1996} adapted for BiKATs.

\begin{proof}
  Let $I$ be an instance of the Post Correspondence Problem (PCP) over some
  alphabet $\{p,q\}$. In other words let $I$ be $k$ pairs of strings
  $x_i,y_i\in\{p,q\}^+$.  We define two BiKAT expressions $S$ and $T$ over unary
  primitives $\{p,q\}$.  (So they will appear in the forms $\eml{p}$, $\eml{q}$,
  $\emr{p}$ and $\emr{q}$, noting that for any string $s=s_1\cdots s_n$ over
  $\{p,q\}$, the BiKAT axioms imply $\eml{s} = \eml{s_1};\ldots;\eml{s_n}$ and
  $\emr{s} = \emr{s_1};\ldots;\emr{s_n}$.)

Let $S=(\emb{x_1}{y_1}+\ldots+\emb{x_k}{y_k})^*$.
Let $T$ be
$$(\emb{p}{p}+\emb{q}{q})^*;\big((\eml{p}+ \eml{q})^++(\emr{p}+\emr{q})^++(\emb{p}{q}+\emb{q}{p});(\eml{p}+\eml{q}+\emr{p}+\emr{q})^*\big).$$
We will show that the equation $S\leq T$ holds in all $*$-continuous BiKATs if and only if $I$ has no solution. 

Suppose first that $I$ has no solution, and let $\alpha=\alpha_1\cdots\alpha_n$ be any sequence from $\{1,\ldots,k\}^n$, which by assumption it cannot be a solution. Then $\emb{x_{\alpha_1}\cdots x_{\alpha_n}}{y_{\alpha_1}\cdots y_{\alpha_n}}$ is equivalent to an element of the form $\emb{z}{z}\emb{p}{q}w$ or $\emb{z}{z}\emb{q}{p}w$ for $w$ an arbitrary element, or of the form $\emb{z}{z}w$ for $w$ a non-null element of only left embedded or only right embedded symbols. There cannot be any other elements since $\alpha$ is not a solution. By the axioms for BiKATs, all such elements can be shown to be less than or equal to $T$. For any $*$-continuous BiKAT, $S$ is the supremum of all such elements of the form $\emb{x_{\alpha_1}\cdots x_{\alpha_n}}{y_{\alpha_1}\cdots y_{\alpha_n}}$ and therefore $S\leq T$.

% \begin{verbatim}
% 	Underlying KAT:
%
% 	A: wp({p,q}*)
% 	B: { {ε}, \emptyset }
%
% 	0 = \emptyset
% 	1 = {ε}
% 	dot : concatenation
% 	+ : union
% 	star : iterated dot
% 	p,q as expressions are interpreted as singleton sets
%
%
% 	D = {p_L,q_L}* \cup {p_R,q_R}*
%
% 	BiKAT: wp(D \times D)
%
% 	embLeft(S) : union of f_L(x) for x in S
% 				 f_L: {p,q}* -> wp(D\times D)
% 				 inductively defined:
%
% 	f_L(ε) = {(ε, ε)}
%
% 	f_L(p) = right hand side of <p|
% 	f_L(q) = right hand side of <q|
%
% 	f_L(x y) = f_L(x);f_L(y)
%
% 	Similar for embRight
%
% 	BiKAT:
% 		M(<u|)
%
% \end{verbatim}

For the other direction, suppose then that $\alpha=\alpha_1\cdots\alpha_n$ is a solution of $I$, meaning that $x_{\alpha_1}\cdots x_{\alpha_n}=y_{\alpha_1}\cdots y_{\alpha_n}$. Denote these strings by $z$. It suffices to show there is a $*$-continuous BiKAT for which the equation $S\leq T$ does not hold.

First consider the KAT $K \eqdef (\wp(\{p,q\}^*),\{\emptyset, \{\epsilon\}\},\cup,\kdot,\kstar,\kneg,\{\epsilon\},\emptyset)$, where $\kdot$ is the operation of element-wise concatenation on two sets of strings, and $\kstar$ is the iterated application of element-wise concatenation. We can define an interpretation $I$ of KAT expressions over $\{p,q\}$ by letting $I(p)=\{p\}$ (respectively $I(q)=\{q\}$). We will define a BiKAT with actions being binary relations on the set of strings $\{p_L,q_L\}^*\cup\{p_R,q_R\}^*$ where $0$, $1$, dot, sum and star are interpreted in the usual way as in relation models.

Given any element $a$ of $K$, let $\eml{a}$ be defined as the union of all $f_L(x)$ for $x\in a$, where $f_L$ is defined inductively as follows (where $\epsilon$ is the empty string):

\begin{displaymath}
	\begin{array}{rlr}
		f_L(\epsilon) = & \{(x, x)\mid x\in\{p_L,q_L\}^*\cup\{p_R,q_R\}^*\}\\
		f_L(s\cdot u) = & f_L(s);\{(x_L,x_L u_L)\mid x_L\in \{p_L,q_L\}^*\}\cup\{(u_R x_R,x_R)\mid x_R\in \{p_R,q_R\}^*\} & \text{for }u\in\{p,q\}.\\
	\end{array}
\end{displaymath}
The embedding $\emr{a}$ is defined analogously, but using the helper function $f_R$:
\begin{displaymath}
	\begin{array}{rlr}
		f_R(\epsilon) = & \{(x, x)\mid x\in\{p_L,q_L\}^*\cup\{p_R,q_R\}^*\}\\
		f_R(s\cdot u) = & f_L(s);\{(x_R,x_R u_R)\mid x_R\in \{p_R,q_R\}^*\}\cup\{(u_L x_L,x_L)\mid x_L\in \{p_L,q_L\}^*\} & \text{for }u\in\{p,q\}.\\
	\end{array}
\end{displaymath}
As a reminder, $;$ above is the usual relational composition. It is straightforward to show that the left and right embeddings as defined are KAT homomorphisms.

\bigskip % avoid hard coded lengths \vskip10pt
\noindent
{\it Claim 1:} For any KAT expression $e$ over $\{p,q\}$ we have:
\begin{displaymath}
	\begin{array}{rl}
		I(\eml{e}) =& \{(x_L,x_L y_L)\mid x_L\in \{p_L,q_L\}^*\}\cup\{(y_R x_R,x_R)\mid x_R\in \{p_R,q_R\}^*\}\\
		I(\emr{e}) =& \{(x_R,x_R y_R)\mid x_R\in \{p_R,q_R\}^*\}\cup\{(y_L x_L,x_L)\mid x_L\in \{p_L,q_L\}^*\}.\\
	\end{array}
\end{displaymath}
where $y_L$ is the ``left'' version of some $y$ in $I(e)$ and $y_R$ is the ``right'' version of some $y$ in $I(e)$.

\bigskip
\noindent
{\it Claim 2:} For any KAT expression $e$ over $\{p,q\}$ we have:
\[ I(\eml{e})\cup I(\emr{e}) \subseteq
   (\{p_L,q_L\}^*\times\{p_L,q_L\}^*)\cup(\{p_R,q_R\}^*\times\{p_R,q_R\}^*)
\]

\medskip
To show that the BiKAT is well-defined it remains to show LRC, and in particular that $\emb{w}{v} = \emb{v}{w}$, where $w$ and $v$ are arbitrary elements from the KAT $K$
(and we omit writing $I$).  
Let $(x_1,x_2)$ be in $\emb{w}{v}$ and recall $\emb{w}{v}= \eml{w};\emr{v}$ by definition.
Then there is $x$ such that $(x_1,x)\in\eml{w}$ and $(x,x_2)\in\emr{v}$. Then, by Claim 1 and properties of relational composition, either $x_1,x_2$ and $x$ are all ``left'' strings in $\{p_L,q_L\}^*$, or they are all ``right'' strings in $\{p_R,q_R\}^*$. Without loss of generality, suppose they are both ``left'' strings. But then, by Claim 1, the strings have to be of the form $x_1 = a_L\cdot s_L$, $x = a_L\cdot s_L\cdot b_L$ and $x_2 = s_L\cdot b_L$.
It remains to show that $(x_1,x_2)\in\emr{v};\eml{w}$. It suffices to show $(x_1,s_L)\in \emr{v}$ and $(s_L,x_2)\in\emr{w}$, which follows from Claim 1.

% But then, $(b_Ry_R,y_R)\in \eml{w}$ and $(y_R,y_Ra_R)\in \emr{v}$.

Let the BiKAT also contain the bitest $E$, interpreted as the identity relation on the empty string $\epsilon$. 
It can be shown that $E;\emb{p}{p}= E;\emb{q}{q} = E$, which implies that $E;(\emb{p}{p}+\emb{q}{q})^*=E$, and $E;\emb{z}{z};E = E$. Since $\emb{z}{z} = \emb{x_{\alpha_1}\cdots x_{\alpha_n}}{y_{\alpha_1}\cdots y_{\alpha_n}}\leq S$ it follows that $E\leq E;S;E$ and therefore $E;S;E\neq 0$.

It also follows that $E;\emb{p}{q}= E;\emb{q}{p}=0$ and $E;(\eml{p}+\eml{q})^+;E= E;(\emr{p}+\emr{q})^+;E=0$. Therefore
\begin{small}
\begin{displaymath}
  \begin{array}{rl}
    & E;T;E \\
          = & E;(\emb{p}{p}+\emb{q}{q})^*;((\eml{p}+\eml{q})^+ + (\emr{p}+\emr{q})^+ + (\emb{p}{q}+\emb{q}{p});
              (\eml{p}+\eml{q}+\emr{p}+\emr{q})^*);E \\
          = & E;((\eml{p}+\eml{q})^+ + (\emr{p}+\emr{q})^+ + (\emb{p}{q}+\emb{q}{p});(\eml{p}+\eml{q}+\emr{p}+\emb{q})^*)\\
		=& E;(\eml{p}+\eml{q})^+;E+E;(\emr{p}+\emr{q})^+;E + E;(\emb{p}{q}+\emb{q}{p});(\eml{p}+\eml{q}+\emr{p}+\emr{q})^*;E\\
		= & 0
	\end{array}
\end{displaymath}
\end{small}
Since $E;S;E\neq 0$ and $E;T;E=0$, it cannot be the case that $S\leq T$.
\end{proof}

\subsection{Details for Sect.~\ref{sec:examples}}

Following on from the discussion in Sect.~\ref{sec:examples}
of Example~\ref{eg:three} from Sect.~\ref{sec:overview}:
In Fig.~\ref{fig:inv_loop_tiling}, we show the proof of the invariant $\bdots{\mathcal{I}}$ over the BiKAT alignment of the two inner loops. After using the BiKAT expansion law (\ref{eq:expand}) to lockstep align the loop bodies and deriving sub-proofs over these alignments, we can easily prove the validity of the sub-proofs. Similarly, we align the two outer loops and prove their invariant $\bdots{\mathcal{J}}$.
\begin{figure*}[t]
  \begin{footnotesize}
    \[ 
    \begin{array}{ll}
      & \bdots{\mathcal{I}} \ksemi \emb{(\kcode{[x<N*M]} \ksemi \kcode{[x\%M!=0]} \ksemi b_1)^\kstar \ksemi \kNeg{\kcode{[x<N*M]} \ksemi \kcode{[x\%M!=0]}}}
      {(\kcode{[j<M]} \ksemi b_2)^\kstar \ksemi \kNeg{\kcode{[j<M]}}} \ksemi \kneg \bdots{\mathcal{I}} = \kzero \\
      
      \impby & 
      \bdots{\mathcal{I}} \ksemi 
      \emb{\kcode{[x<N*M]} \ksemi \kcode{[x\%M!=0]} \ksemi b_1}
      {\kcode{[j<M]} \ksemi b_2}^\kstar \ksemi
      (\emb{\kcode{[x<N*M]} \ksemi \kcode{[x\%M!=0]} \ksemi b_1}{\kNeg{\kcode{[j<M]}}}^\kstar \kplus 
      \emb{\kNeg{\kcode{[x<N*M]} \ksemi \kcode{[x\%M!=0]}}}{\kcode{[j<M]}; b_2}^\kstar) \ksemi \\  
      & \qquad 
      \emb{\kNeg{\kcode{[x<N*M]} \ksemi \kcode{[x\%M!=0]}}}{\kNeg{\kcode{[j<M]}}} 
      \ksemi \kneg \bdots{\mathcal{I}} = \kzero \\
      
      \impby &
      \left\{
      \begin{array}{l}
        \bdots{\mathcal{I}} \ksemi \emb{\kcode{[x<N*M]} \ksemi \kcode{[x\%M!=0]} \ksemi b_1}
        {\kcode{[j<M]} \ksemi b_2} \ksemi \kneg \bdots{\mathcal{I}} = \kzero
        \wedge
        \bdots{\mathcal{I}} \ksemi \emb{\kcode{[x<N*M]} \ksemi \kcode{[x\%M!=0]} \ksemi b_1}{\kNeg{\kcode{[j<M]}}} \ksemi \kneg \bdots{\mathcal{I}} = \kzero \,\wedge\\ 
        \bdots{\mathcal{I}} \ksemi \emb{\kNeg{\kcode{[x<N*M]} \ksemi \kcode{[x\%M!=0]}}}{\kcode{[j<M]} \ksemi b_2} \ksemi \kneg \bdots{\mathcal{I}} = \kzero
        \wedge
        \bdots{\mathcal{I}} \ksemi \emb{\kNeg{\kcode{[x<N*M]} \ksemi \kcode{[x\%M!=0]}}}{\kNeg{\kcode{[j<M]}}} \ksemi \kneg \bdots{\mathcal{I}} = \kzero 
        \end{array}\right.\\
    \end{array}
    \]
  \end{footnotesize}
  \vspace*{-3ex}
  \caption{BiKAT reasoning for the loop tiling example}
  \label{fig:inv_loop_tiling}
\end{figure*}

\subsection{Details for Sect.~\ref{sec:RHL}}

\subsubsection{Proof of Theorem~\ref{thm:RHL}}

\begin{figure*}[t]

\begin{mathpar}

\inferrule*[left=dIf4]{
  c\sep c' : \rspec{P\land \leftF{e}\land\rightF{e'}}{Q} \\
  d\sep d' : \rspec{P\land \neg\leftF{e}\land\neg\rightF{e'}}{Q} \\
  c\sep d' : \rspec{P\land \leftF{e}\land\neg\rightF{e'}}{Q} \\
  d\sep c' : \rspec{P\land \neg\leftF{e}\land\rightF{e'}}{Q} 
}{
  \ifc{e}{c}{d} \Sep \ifc{e'}{c'}{d'} : \rspec{P}{Q}
}

\inferrule[dSkip]{}{
\skipc \sep \skipc : \rspec{P}{P}
}

\inferrule[AssSkip]{}{
v:=e \sep \skipc  : \rspec{\subst{P}{v|}{e|}}{P}
}

\inferrule*[left=IfSkip]{
  c\sep \skipc : \rspec{P \land \leftF{e}}{Q} \\
  d\sep \skipc : \rspec{P \land \leftF{\neg e}}{Q} 
}{
  \ifc{e}{c}{d} \Sep \skipc : \rspec{P}{Q}
}

\inferrule*[left=WhSk]{
  c\sep \skipc : \rspec{P\land \leftF{e}}{P} 
}{ 
  \whilec{e}{c} \Sep \skipc : \rspec{P}{P\land \leftF{\neg e}}
}

\inferrule*[left=rConj]{
  c\sep d : \rspec{P}{Q} \\
  c\sep d : \rspec{P}{R} 
}{
  c\sep d : \rspec{P}{Q\land R} \\
}

\inferrule[falsePre]{}{    
c\sep c' : \rspec{False}{P}
}

\end{mathpar}

\caption{Some additional $\forall\forall$ rules}
\label{fig:RHLadditional}
\end{figure*}

In this section we prove some of the $\forall\forall$ RHL rules (see Fig.~\ref{fig:RHLselected}).  
Fig.~\ref{fig:RHLadditional} gives some additional RHL rules that are also provable in any BiKAT.

\bigskip

For \rn{dSkip}: the judgment $\skipc \sep \skipc : \rspec{P}{P}$ is encoded as
$P; \emb{1}{1} = P; \emb{1}{1}; P$ which holds because $\emb{1}{1} = \bone$ by 
the homomorphism property of the embeddings.
Strictly: $\emb{1}{1} = \eml{1};\emr{1} = \bone;\bone = \bone$.

\bigskip

For \rn{rConseq}:
$c\sep d : \rspec{P}{Q}$ is encoded as
$P;\emb{c}{d}\leq P;\emb{c}{d};Q$
and the implications as $R\leq P$ and $Q\leq S$.
These immediately yield 
$R;\emb{c}{d}\leq R;\emb{c}{d};S$
by transitivity.

\bigskip

For \rn{rDisj}:
The premises 
$  c\sep d : \rspec{P}{Q}$ and 
$  c\sep d : \rspec{R}{Q}$
are encoded as 
$P;\emb{c}{d}\leq P;\emb{c}{d};Q$
and 
$R;\emb{c}{d}\leq R;\emb{c}{d};Q$.
The conclusion $ c\sep d : \rspec{P\lor R}{Q}$ is encoded as 
$(P+R);\emb{c}{d}\leq (P+R);\emb{c}{d};Q$, and it is proved by
\[\begin{array}{lll}
     & (P+R);\emb{c}{d} \\
=    & P;\emb{c}{d} + R;\emb{c}{d} & \mbox{distrib} \\
\leq & P;\emb{c}{d};Q + R;\emb{c}{d};Q & \mbox{premises} \\
=    & (P+R);\emb{c}{d};Q & \mbox{distrib} 
\end{array}\]
Note that the proof is independent of the form of $\emb{c}{d}$.
The reasoning embodied by this rule is sound for any BiKAT term.  

\bigskip

For \rn{rConj}:
The premises 
$  c\sep d : \rspec{P}{Q}$ and 
$  c\sep d : \rspec{P}{R}$
are encoded as 
$P;\emb{c}{d}\leq P;\emb{c}{d};Q$
and 
$P;\emb{c}{d}\leq P;\emb{c}{d};R$.
The conclusion $ c\sep d : \rspec{P}{Q\land R}$ is encoded as 
$P;\emb{c}{d}\leq P;\emb{c}{d};Q;R$, and it is proved by
$ P;\emb{c}{d} 
\leq  P;\emb{c}{d};R 
\leq P;\emb{c}{d};Q;R $
using the second premise and then the first.

As in the case of \rn{rDisj}, the proof is independent of the form of $\emb{c}{d}$.
The reasoning embodied by this rule is sound for any BiKAT term.  

\bigskip

For \rn{dSeq}:  
\[\begin{array}{lll}
  & P; \emb{c;d}{c';d'} \\
=  & P; \emb{c}{c'}; \emb{d}{d'} & \mbox{embed homo, i.e., LRC} \\
\leq  & P; \emb{c}{c'}; R; \emb{d}{d'} & \mbox{premise} \\ 
\leq  & P; \emb{c}{c'}; R; \emb{d}{d'}; Q & \mbox{premise} \\ 
\leq  & P; \emb{c}{c'}; \emb{d}{d'}; Q &\mbox{test $R\leq 1$} \\
=     & P; \emb{c;d}{c';d'};Q &\mbox{emb homo}
\end{array}\]
The premises are used in the form
$P;\emb{c}{c'}\leq P;\emb{c}{c'};R$ and 
$R;\emb{d}{d'}\leq R;\emb{d}{d'};Q$
and those steps use monotonicity of sequence.

\bigskip

For \rn{SeqSk}, 
suppose $P;\emb{c}{1}\leq P;\emb{c}{1};R$
and $R;\emb{d}{1}\leq R;\emb{d}{1};Q$.
Then
$P;\emb{c;d}{1}
=
P;\emb{c}{1};\emb{d}{1}
\leq
P;\emb{c}{1};R;\emb{d}{1}
\leq
P;\emb{c}{1};R;\emb{d}{1};Q
\leq 
P;\emb{c}{1};\emb{d}{1};Q$
using homomorphism, the premises, and test below 1.

\bigskip

For \rn{WhSkip}, the premise has the form
$P;\eml{e};\emb{c}{1}\leq P;\eml{e};\emb{c}{1};P$
which can be rewritten to 
$P;\eml{e;c}\leq P;\eml{e;c};P$.
The conclusion follows by
\[\begin{array}{lll}
   & P;\emb{(e;c)^*;\neg e}{1} \\
=  & P;\eml{e;c}^*;\neg\eml{e} &\mbox{emb def and homo}\\
\leq& P;\eml{e;c}^*;P;\neg\eml{e} &\mbox{premise, invariance lem}\\
=  & P;\emb{(e;c)^*;\neg e}{1};P;\neg\eml{e}&\mbox{emb def and homo}
\end{array}\]

\subsubsection{Proof of Theorem~\ref{thm:caWh}}

\paragraph{Proof of Expansion Law (\ref{eq:caWhexpand})}

We prove that in any *-continuous model, the following holds,
where $e,e'$ are tests in the underlying model and $Q,R$ are bitests.
\begin{equation}\label{eq:caWhexpandCopy}
\begin{array}{l}
\eml{e;c}^*;\eml{\neg e} ; \emr{e';c'}^*;\emr{\neg e'} = d^* ; \emb{\neg e}{\neg e'} \\
\mbox{ where } d = Q;\eml{e;c} + R;\emr{e';c'} + \neg Q \neg R;\emb{e;c}{e';c'} 
+ \neg Q;\emb{e;c}{\neg e'} + \neg R;\emb{\neg e}{e';c'} 
\end{array}
\end{equation}
Note that (\ref{eq:caWhexpandCopy}) is the same as 
(\ref{eq:caWhexpand}); it is re-stated here for convenience and to define abbreviation $d$ for the right side loop body.
The proof is by mutual inclusion.  

One direction holds in any biKAT, namely 
$d^*   \emb{\neg e}{\neg e'} \leq \eml{e c}^* \eml{\neg e}   \emr{e' c'}^* \emr{\neg e'}$.
(We omit semicolon to save space.)
In its proof we use that tests are below 1, and in particular the consequence that
$\emb{e c}{\neg e'}\leq  \eml{e c}$ and $\emb{\neg e}{e' c'} \leq \emr{e' c'}$.
\[
\begin{array}{lll}
     &( Q \eml{e c} + R \emr{e' c'} + \neg Q \neg R \emb{e c}{e' c'} 
		+ \neg Q \emb{e c}{\neg e'} + \neg R \emb{\neg e}{e' c'} )^*   \emb{\neg e}{\neg e'}\\
\leq & ( \eml{e c} + \emr{e' c'} + \emb{e c}{e' c'} 
		+ \emb{e c}{\neg e'} + \emb{\neg e}{e' c'} )^*   \emb{\neg e}{\neg e'}
     %&\mbox{tests below 1}
\\
=  & ( \eml{e c} + \emr{e' c'} + \emb{e c}{e' c'}  )^*   \emb{\neg e}{\neg e'}
     %&\mbox{$\emb{e c}{\neg e'}\leq  \eml{e c}$, $\emb{\neg e}{e' c'} \leq \emr{e' c'}$} 
\\
  = & \eml{e c}^*  \emr{e' c'}^*   \emb{\neg e}{\neg e'}
    & \mbox{Lemma~\ref{lem:seq-star-expansion} }
	\end{array}
\]
To prove the reverse inclusion, 
$\eml{e c}^* \eml{\neg e}   \emr{e' c'}^* \emr{\neg e'} \leq d^*   \emb{\neg e}{\neg e'}$,
we first commute $\eml{\neg e}$ with $\emr{e' c'}^*$ 
(by LRC and Lemma~\ref{lem:item-over-star}),
to get the equivalent inequality
\[ \eml{e c}^* \emr{e' c'}^* \emb{\neg e}{\neg e'} \leq d^*   \emb{\neg e}{\neg e'} \]
Now rewrite the left side using *-continuity:
\[ 
(sup : n,m\in\mathbb{N} : \eml{e c}^n \emr{e' c'}^m \emb{\neg e}{\neg e'}) \leq d^*   \emb{\neg e}{\neg e'}
\] 
By the sup property this is equivalent to 
\[ 
\all{n, m\in\mathbb{N}}{
   \eml{e c}^n  \emr{e' c'}^m \emb{\neg e}{\neg e'} \leq d^*   \emb{\neg e}{\neg e'} }
\] 
We prove this by induction on $n+m$.  

For the base case,  $n=0=m$, we have
\( \eml{e c}^0 \emr{e' c'}^0 \emb{\neg e}{\neg e'} 
= 
\emb{\neg e}{\neg e'}
\leq d^*   \emb{\neg e}{\neg e'}
\) 
using the definition of iterate 0 as $\kone$, and $1\leq d^*$.

For the induction step we consider separately the cases where $n$ or $m$ is zero.
\begin{list}{}{}
\item[\underline{Case} $n>0$ and $m=0$:]
\[
\begin{array}{lll}
    & \eml{e c}^n \emr{e' c'}^0 \emb{\neg e}{\neg e'} \\
=   & \eml{e c} \eml{e c}^{n-1} \emb{\neg e}{\neg e'} & \mbox{def iterate}\\
=   &   Q \eml{e c} \eml{e c}^{n-1} \emb{\neg e}{\neg e'} 
      + \kNeg{Q} \eml{e c} \eml{e c}^{n-1} \emb{\neg e}{\neg e'} &\mbox{KAT} \\
=   &   Q \eml{e c} \eml{e c}^{n-1} \emb{\neg e}{\neg e'} 
      + \kNeg{Q} \emb{e c}{\neg e'} \eml{e c}^{n-1} \emb{\neg e}{\neg e'} &\mbox{$\emr{\neg e'}$ idempotent, LRC} \\
\leq&   Q \eml{e c} d^* \emb{\neg e}{\neg e'} 
      + \kNeg{Q} \emb{e c}{\neg e'} d^* \emb{\neg e}{\neg e'} &\mbox{induction, twice} \\
\leq&   d d^* \emb{\neg e}{\neg e'} + d d^* \emb{\neg e}{\neg e'}
    &\mbox{ $Q \eml{e c} \leq d$ and $\kNeg{Q} \emb{e c}{\neg e'} \leq d$} \\
\leq&   d^* \emb{\neg e}{\neg e'} + d^* \emb{\neg e}{\neg e'}  &\mbox{$d d^*\leq d^*$} 
      \end{array}
\]

\item[\underline{Case} $n=0$ and $m>0$:]
Symmetric with the preceding case $n>0$ and $m=0$,
using $R\emr{e' c'} \leq d$ and $\kNeg{R} \emb{\neg e}{e' c'} \leq d$.

\item[\underline{Case} $n>0$ and $m>0$:]
By distributivity and boolean algebra, we rewrite
$\eml{e c}^n \emr{e' c'}^m \emb{\neg e}{\neg e'} $
to the sum
\[  \begin{array}{ll}
  & \kNeg{Q} \, \kNeg{R}\eml{e c}^n \emr{e' c'}^m \emb{\neg e}{\neg e'} \\
+ & Q\kNeg{R}\eml{e c}^n  \emr{e' c'}^m \emb{\neg e}{\neg e'}  \\
+ & \kNeg{Q} R \eml{e c}^n \emr{e' c'}^m \emb{\neg e}{\neg e'} \\
+ & Q R \eml{e c}^n \emr{e' c'}^m \emb{\neg e}{\neg e'} 
    \end{array}
\]
We show each term of the sum is below $d^* \emb{\neg e}{\neg e'}$.

\begin{itemize}
\item
\[  \begin{array}{lll}
    & \kNeg{Q} \, \kNeg{R}\eml{e c}^n \emr{e' c'}^m \emb{\neg e}{\neg e'} \\
=   & \kNeg{Q} \, \kNeg{R}\emb{e c}{e' c'} \eml{e c}^{n-1} \emr{e' c'}^{m-1} \emb{\neg e}{\neg e'} 
    & \mbox{def iterate, LRC} \\
\leq& \kNeg{Q} \, \kNeg{R}\emb{e c}{e' c'} d^* \emb{\neg e}{\neg e'} 
    & \mbox{induction} \\
\leq& d d^* \emb{\neg e}{\neg e'}  & \mbox{$\kNeg{Q} \, \kNeg{R}\emb{e c}{e' c'} \leq d$} \\
\leq& d^* \emb{\neg e}{\neg e'}  & \mbox{$dd^* \leq d^*$}
    \end{array}
\]

\item
\[  \begin{array}{lll}
     & Q \kNeg{R}\eml{e c}^n \emr{e' c'}^m \emb{\neg e}{\neg e'} \\
=    & Q \kNeg{R} \eml{e c} \eml{e c}^{n-1} \emr{e' c'}^m \emb{\neg e}{\neg e'} &\mbox{def iterate, LRC}\\
\leq & Q \kNeg{R} \eml{e c} d^*  \emb{\neg e}{\neg e'} &\mbox{induction} \\
\leq & Q \eml{e c} d^*  \emb{\neg e}{\neg e'} &\mbox{test $\kNeg{R} \leq 1$} \\
\leq & d d^*  \emb{\neg e}{\neg e'} &\mbox{$Q \eml{e c} \leq d$} \\
\leq & d^*  \emb{\neg e}{\neg e'} &\mbox{$dd^* \leq d^*$}
\end{array}
\]

\item The argument for 
$ \kNeg{Q} R \eml{e c}^n \emr{e' c'}^m \emb{\neg e}{\neg e'}$ is symmetric with the preceding bullet,
using that $R\emr{e' c'}\leq d$

\item
\[  \begin{array}{lll}
     & Q R\eml{e c}^n \emr{e' c'}^m \emb{\neg e}{\neg e'} \\
=    & Q R \eml{e c} \eml{e c}^{n-1} \emr{e' c'}^m \emb{\neg e}{\neg e'} &\mbox{def iterate, LRC}\\
\leq & Q R \eml{e c} d^*  \emb{\neg e}{\neg e'} &\mbox{induction} \\
\leq & Q \eml{e c} d^*  \emb{\neg e}{\neg e'} &\mbox{test $R\leq 1$} \\
\leq & d d^*  \emb{\neg e}{\neg e'} &\mbox{$Q \eml{e c} \leq d$} \\
\leq & d^*  \emb{\neg e}{\neg e'} &\mbox{$dd^* \leq d^*$}
\end{array}
\]
For this case there is an alternate proof using using $R\emr{e' c'}\leq d$ instead of $Q \eml{e c} \leq d$.
\end{itemize}

\end{list}

\bigskip

\paragraph{Proof of Rule \rn{caWh}}

By calculation, we show that 
\rn{caWh} is sound in any BiKAT that satisfies (\ref{eq:caWhexpandCopy}).

The premises of rule \rn{caWh} are:
\[\begin{array}{rcl}
P;\leftF{e};\rightF{e'};\neg Q;\neg R; \emb{c}{c'} &\leq&
P;\leftF{e};\rightF{e'};\neg Q;\neg R; \emb{c}{c'}; P \\
P; Q; \leftF{e}; \eml{c} &\leq& P; Q; \leftF{e}; \eml{c}; P\\
P; R; \rightF{e'}; \emr{c'} &\leq& P; R; \rightF{e'}; \emr{c'}; P
\end{array}\]
The side condition is $P \leq \eqbib{e}{e'} + Q;\leftF{e} + R;\rightF{e'}$,
which implies 
\begin{equation}\label{eq:caWhSide}
P; \neg Q; \eml{e}; \neg\emr{e'} = 0
\quad\mbox{and}\quad
   P; \neg R; \neg\eml{e}; \emr{e'} = 0
\end{equation}
Observe that the premises imply $P$ is preserved by the loop body:
\[\begin{array}{lll}
 & P( Q\eml{ec} + R\emr{e' c'} + \neg Q \neg R \emb{ec}{e'c'} 
      + \neg Q \emb{e c}{\neg e'} + \neg R \emb{\neg e}{e'c'} ) \\
=  &\quad\mbox{distrib} \\
& P Q\eml{ec} + PR\emr{e' c'} + P \neg Q \neg R \emb{ec}{e'c'} 
      + P \neg Q \emb{e c}{\neg e'} + P\neg R \emb{\neg e}{e'c'} \\
= &\quad\mbox{side condition (\ref{eq:caWhSide})} \\
& P Q\eml{ec} + PR\emr{e' c'} + P \neg Q \neg R \emb{ec}{e'c'}  \\
\leq &\quad\mbox{premises, emb homo} \\
& P Q\eml{ec} P + PR\emr{e' c'} P + P \neg Q \neg R \emb{ec}{e'c'} P  \\
= &\quad\mbox{side condition (\ref{eq:caWhSide})} \\
& P Q\eml{ec} P + PR\emr{e' c'} P + P \neg Q \neg R \emb{ec}{e'c'} P 
      + P \neg Q \emb{e c}{\neg e'} P + P\neg R \emb{\neg e}{e'c'} P \\
= &\quad\mbox{distrib} \\
& P (Q\eml{ec}  + R\emr{e' c'}  +  \neg Q \neg R \emb{ec}{e'c'}  
      +  \neg Q \emb{e c}{\neg e'}  + \neg R \emb{\neg e}{e'c'} ) P \\
\end{array}\]
The conclusion of \rn{caWh} follows by
\[\begin{array}{lll}
  & P \eml{ec}^* \eml{\neg e}  \emr{e' c'}^* \emr{\neg e'} \\
= &     \quad\mbox{by (\ref{eq:caWhexpandCopy})} \\
  & P( Q\eml{ec} + R\emr{e' c'} + \neg Q \neg R \emb{ec}{e'c'} 
      + \neg Q \emb{e c}{\neg e'} + \neg R \emb{\neg e}{e'c'} )^*  \emb{\neg e}{\neg e'} \\
\leq &  \quad\mbox{invariance, observation above} \\
  & P( Q\eml{ec} + R\emr{e' c'} + \neg Q \neg R \emb{ec}{e'c'} 
      + \neg Q \emb{e c}{\neg e'} + \neg R \emb{\neg e}{e'c'} )^*  P  \emb{\neg e}{\neg e'} \\
= &     \quad\mbox{tests idem, tests commute} \\
  & P( Q\eml{ec} + R\emr{e' c'} + \neg Q \neg R \emb{ec}{e'c'} 
      + \neg Q \emb{e c}{\neg e'} + \neg R \emb{\neg e}{e'c'} )^*  \emb{\neg e}{\neg e'} P  \emb{\neg e}{\neg e'}  \\
= &     \quad\mbox{by (\ref{eq:caWhexpandCopy})} \\
  & P (\eml{ec}^* \eml{\neg e}  \emr{e' c'}^* \emr{\neg e'})  P  \emb{\neg e}{\neg e'} 
\end{array}\]

\subsection{Details for Sect.~\ref{sec:beyond}}

\subsubsection{Details for Sect.~\ref{sec:biKATsim}}

\paragraph{Proof of Theorem~\ref{thm:sim:complete}, Backward Simulation Case.}

Suppose $c\sep d:\bespec{R}{S}$ holds.  
Let witness $W$ be $\dot{R};\emb{c}{d};\dot{S}$.
We show it satisfies the conditions to be b-valid.
\begin{itemize}
\item
We have (WCb) because  $W;\dot{S} = \dot{R};W;\dot{S}$ 
by definition of $W$ and idempotence of the bitest $\dot{R}$.
\item
To show (WOb), i.e., $\eml{c};\dot{S} \leq \emr{\hav};W$,
suppose $(\sigma,\sigma')\eml{c};\dot{S}(\tau,\tau')$.
Then $\sigma c \tau$ and $\sigma'=\tau'$ and $\tau S\tau'$.
Now by $c\sep d:\bespec{R}{S}$ 
there is $\sigma''$ with $\sigma'' d \tau'$ and $\sigma R \sigma''$.
So $(\sigma,\sigma'')W(\tau,\tau')$ and thus 
$(\sigma,\sigma')\emr{\hav};W(\tau,\tau')$.
\item
To show (WUb), i.e.,  $W;\dot{S} \leq \emb{\hav}{d}$,
suppose $(\sigma,\sigma')W;\dot{S}(\tau,\tau')$.
Then by definitions $\sigma R \sigma'$, $\sigma c \tau$, $\sigma' d \tau'$, and $\tau S \tau'$.
Since $\sigma \hav \tau$, we have
$(\sigma,\sigma')\emb{\hav}{d}(\tau,\tau')$.
\end{itemize}
This concludes the proof.

Apropos Footnote~\ref{fn:sim:complete}, if the model is full then it contains 
elements that satisfy the b-validity conditions and are in some sense minimal
(whereas the expressible witness $\dot{R};\emb{c}{d};\dot{S}$ used above is in some sense maximal).
To elaborate, suppose $c\sep d:\bespec{R}{S}$ holds.  
Define the predicate
$\mathcal{P}(\sigma,\tau,\tau') \eqdef \sigma c \tau S \tau'$
and the set 
$\mathcal{X}(\sigma,\tau,\tau') \eqdef \{ \sigma' \mid \sigma R \sigma' d \tau' \}$.
For any $(\sigma,\tau,\tau')$ that satisfy $\mathcal{P}$ we have
$\mathcal{X}(\sigma,\tau,\tau')$ nonempty, owing to $c|d:\bespec{R}{S}$. 
So define 
$\mathcal{Y}(\sigma,\tau,\tau')$ to be a chosen element of $\mathcal{X}(\sigma,\tau,\tau')$ if $\mathcal{P}(\sigma,\tau,\tau')$, and undefined otherwise.
Define 
\( W \eqdef \{((\sigma,\sigma'),(\tau,\tau')) \mid 
                   \mathcal{P}(\sigma,\tau,\tau') \land 
                   \sigma' = \mathcal{Y}(\sigma,\tau,\tau') \} 
\).
By fullness, $W$ is in the BiKAT.
For b-validity we have (WCb) because $W = \dot{R} ;  W ; \dot{S}$,

For (WOb), suppose $(\sigma,\sigma')\eml{c};\dot{S}(\tau,\tau')$.
Then $\sigma c \tau$ and $\sigma'=\tau'$ and $\tau S\tau'$.
So $\mathcal{P}(\sigma,\tau,\tau')$ so there's $\rho$ with
$\sigma R \rho d \tau'$ and $(\sigma,\rho)W(\tau,\tau')$
Since $\sigma'\hav\rho$ we have 
$(\sigma,\sigma')\emr{\hav}(\tau,\tau')$.

For (WUb), suppose $(\sigma,\sigma')W;\dot{S}(\tau,\tau')$.
So $(\sigma,\sigma')W(\tau,\tau')$ and $\tau S \tau'$.  
By definition of $W$ we have $\sigma c \tau S \tau'$ 
and $\sigma R \sigma' d \tau'$.  
Since $\sigma \hav \tau$ we have
$(\sigma,\sigma')\emb{\hav}{d}(\tau,\tau')$.

\begin{figure*}

\begin{mathpar}
	
	\inferrule*[left=DisjWC]{
		R; W_1 \leq R; W_1; S \\
		R; W_2 \leq R; W_2; S
	}{
		R; (W_1 + W_2) \leq R; (W_1 + W_2); S
	}

	\inferrule*[left=DisjWU]{
		R; W_1 \leq \emb{\hav}{d_1} \\
		R; W_2 \leq \emb{\hav}{d_2}
	}{
		R; (W_1 + W_2) \leq \emb{\hav}{d_1 + d_2}
	}

	\inferrule*[left=DisjWO]{
		R; \eml{c_1} \leq W_1; \emr{\hav} \\
		R; \eml{c_2} \leq W_2; \emr{\hav}
	}{
		R; \eml{c_1 + c_2} \leq (W_1 + W_2); \emr{\hav} \\
	}
\end{mathpar}

\caption{Disjunction lemmas for use with forward simulation witness technique}\label{fig:disjLem}
\end{figure*}

\paragraph{Disjunctive Forward Witness Soundness Rules.} To support case analysis, we also have additional forward witness soundness rules with disjunctions.
Fig.~\ref{fig:disjLem} presents three lemmas for disjunctive decomposition of witnesses for forward simulation.

\rn{DisjWC} holds because
\(\begin{array}[t]{lll}
	& R; (W_1 + W_2) & \\
	= & R; W_1 + R; W_2 & \mbox{distrib} \\
	\leq & R; W_1; S + R; W_2; S & \mbox{assumptions}\\
	= & R; (W_1 + W_2); S
\end{array}\)

\rn{DisjWU} holds because
\(\begin{array}[t]{lll}
	& R; (W_1 + W_2) & \\
	= & R; W_1 + R; W_2 & \mbox{distrib} \\
	\leq & \emb{\hav}{d_1} + \emb{\hav}{d_2} & \mbox{assumptions}\\
	= & \emb{\hav}{d_1 + d_2} & \mbox{emb homo}\\
\end{array}\)

\rn{DisjWO} holds because
\(\begin{array}[t]{lll}
	& R; \eml{c_1 + c_2} & \\
	= & R;  \eml{c_1} + R;  \eml{c_2} & \mbox{assumption, emb homo} \\
	\leq & W_1; \emr{\hav} + W_2; \emr{\hav} & \mbox{assumptions}\\
	= & (W_1 + W_2); \emr{\hav}
\end{array}\)

\paragraph{Proof of Example~\ref{eg:four} in Sect.~\ref{sec:overview}}

The programs:
\begin{lstlisting}
  (c1)  x := any; y := x
  (c2)  t := any; z := t + 1
\end{lstlisting}
Goal: forward simulation judgment $c_1\sep c_2 : \aespec{true}{\eqbib{y}{z}}$.\\
Choose as witness,
$W\eqdef \emb{\kcode{x:=any}}{\kcode{t:=any}}; [x-1 \eqbi t];
\emb{\kcode{y:=x}}{\kcode{z:=t+1}}$.
\\
To show:
\begin{list}{}{}
\item[(WC)] $W \leq W; [\eqbib{y}{z}]$
\item[(WU)] $W \leq \emb{\hav}{\kcode{t:=any}; \kcode{z:=t+1}}$
\item[(WO)] $\eml{\kcode{x:=any}; \kcode{y:=x}} \leq W; \emr{\hav}$
\end{list}

\noindent $\bullet$ (WC)
\(\begin{array}[t]{lll}
    & \emb{\kcode{x:=any}}{\kcode{t:=any}}; [x-1 \eqbi t];\emb{\kcode{y:=x}}{\kcode{z:=t+1}}\\
  = & \emb{\kcode{x:=any}}{\kcode{t:=any}}; [x \eqbi t+1];\emb{\kcode{y:=x}}{\kcode{z:=t+1}}
    & \mbox{equiv test} \\    
  = & \emb{\kcode{x:=any}}{\kcode{t:=any}}; [x \eqbi t+1];\emb{\kcode{y:=x}}{\kcode{z:=t+1}}
      [y \eqbi z]
  \end{array}\)
\\
The last step is by the assignment axiom
$wp(\emb{\kcode{y:=x}}{\kcode{z:=t+1}})(y\eqbi z) = (x \eqbi t+1)$
in the KAT form
$(x \eqbi t+1); \emb{\kcode{y:=x}}{\kcode{z:=t+1}} = 
(x \eqbi t+1); \emb{\kcode{y:=x}}{\kcode{z:=t+1}}; (y\eqbi z)$.
\\
$\bullet$ (WU)
\(\begin{array}[t]{lll}
     & \emb{\kcode{x:=any}}{\kcode{t:=any}}; [x-1 \eqbi t];\emb{\kcode{y:=x}}{\kcode{z:=t+1}}\\
\leq & \emb{\kcode{x:=any}}{\kcode{t:=any}}; \emb{\kcode{y:=x}}{\kcode{z:=t+1}} &
       \mbox{tests below 1}\\
=    & \emb{\kcode{x:=any}; \kcode{y:=x}}{\kcode{t:=any}; \kcode{z:=t+1}} & \mbox{emb homo}\\
\leq & \emb{\hav}{\kcode{t:=any}; \kcode{z:=t+1}} & \mbox{$\hav$ is top, emb mono}
  \end{array}\)
\\

$\bullet$ (WO)
We use these axioms for the primitives:
\begin{itemize}
\item[(a)] $1 \leq \emr{\kcode{t:=any}};[x-1 \eqbi t];\emr{\kcode{t:=any}}$ (says that $[x-1\eqbi t]$ is left total, cf.~rule \rn{enAss})
\item[(b)] $\eml{\kcode{y:=x}};[x-1 \eqbi t] = [x-1 \eqbi t]; \eml{\kcode{y:=x}}$ (disjoint variables)
\item[(c)] $\kcode{z:=t+1};\hav = \hav$ (assignments are total)
\end{itemize}
We have (WO) because 
\[\begin{array}[t]{lll}
     & \eml{\kcode{x:=any};\kcode{y:=x}} \\
=    & \eml{\kcode{x:=any}};\eml{\kcode{y:=x}} &\mbox{emb homo} \\
\leq & \eml{\kcode{x:=any}};\eml{\kcode{y:=x}};\emr{\kcode{t:=any}};[x-1 \eqbi t];\emr{\kcode{t:=any}} &\mbox{(a)} \\
=    & \emb{\kcode{x:=any}}{\kcode{t:=any}};\eml{\kcode{y:=x}};[x-1 \eqbi t];\emr{\kcode{t:=any}} &\mbox{LRC} \\
=    & \emb{\kcode{x:=any}}{\kcode{t:=any}};[x-1 \eqbi t];\eml{\kcode{y:=x}};\emr{\kcode{t:=any}} &\mbox{(b)} \\
\leq & \emb{\kcode{x:=any}}{\kcode{t:=any}};[x-1 \eqbi t];\eml{\kcode{y:=x}};\emr{\hav} &\mbox{$\hav$ top, embed mono} \\
=    & \emb{\kcode{x:=any}}{\kcode{t:=any}};[x-1 \eqbi t];\emb{\kcode{y:=x}}{\kcode{z:=t+1}};\emr{\hav} &\mbox{(c)}
\end{array}\]

\subsubsection{Proofs for Lemma~\ref{lem:fsim:rules}, i.e., Forward Simulation Rules
in Fig.~\ref{fig:aeRHL}}

\paragraph{Proof of \rn{eAss}}
For alignment witness take $W:= \emb{v:=e}{v':=e'}$ itself.
Have (WC) in accord with the $\forall\forall$ assignment rule dAss.
For (WU), we have 
$\subst{P}{v|v'}{e|e'}; \emb{v:=e}{v':=e'} 
\leq \emb{v:=e}{v':=e'} 
\leq \emb{\hav}{v':=e'}$
by test below 1 and $\hav$ above all.
For (WO):
\[ 
\subst{P}{v|v'}{e|e'}; \emb{v:=e}{1} 
\leq \emb{v:=e}{1} 
\leq \emb{v:=e}{\hav} 
= \emb{v:=e}{v':=e';\hav} 
= \emb{v:=e}{v':=e'};\emb{1}{\hav}
\]
using homomorphism in the last step and
the law $\hav = v':=e';\hav$.
This law holds in usual operational and denotational semantics.

\paragraph{Proof of \rn{enAss}}

Let witness $W$ to be $\emb{x:=\kcode{any}}{y:=\kcode{any}};\dot{R}$. 

$\bullet$ (WC) $\bone; \emb{x:=\kcode{any}}{y:=\kcode{any}};\dot{R} \leq \bone; \emb{x:=\kcode{any}}{y:=\kcode{any}};\dot{R};\dot{R}$
holds  by idempotence of tests.

$\bullet$ (WU) $\bone ; \emb{x:=\kcode{any}}{y:=\kcode{any}}; \dot{R} \leq \emb{\hav}{y:=\kcode{any}}$
holds by $\bone$ identity, $\hav$ top, $\eml{-}$ monotonic, and $\dot{R}\leq \bone$.

$\bullet$ (WO) is $\bone; \eml{x:=\kcode{any}} \leq \emb{x:=\kcode{any}}{y:=\kcode{any}};\dot{R};\emr{\hav}$
and we have 
$\bone; \eml{x:=\kcode{any}} 
= \bone; \eml{x:=\kcode{any}}; \bone
\leq \bone; \eml{x:=\kcode{any}}; \emr{y:=\kcode{any}};\dot{R};\emr{y:=\kcode{any}}
= \emb{x:=\kcode{any}}{y:=\kcode{any}};\dot{R};\emr{y:=\kcode{any}}
\leq 
\emb{x:=\kcode{any}}{y:=\kcode{any}};\dot{R};\emr{\hav}$
using the antecedent of the rule in the second step.

\paragraph{Proof of \rn{eConseq}}

Given witness $W$ for 
$c\sep d : \aespec{P}{Q}$,
and conditions $R\imp P$ and $Q\imp S$,
we have that $W$ is also a witness for 
$ c|d : \aespec{R}{S}$.
To show (WC) for the latter, we have
$\dot{R}; W \leq 
 \dot{P}; W \leq
 \dot{P}; W; \dot{Q} \leq
 \dot{P}; W; \dot{S}$
using (WC) for $W$.
And $\dot{R}; W \leq \dot{P}; W; \dot{S}$
iff $\dot{R}; W \leq \dot{R}; W; \dot{S}$
by KAT fact $px\leq y \iff px\leq py$ and 
$\dot{R}\cdot\dot{P} = \dot{R}$ from $R\imp P$.
The proofs of (WO) and (WU) are even simpler.

\paragraph{Proof of \rn{eIf}}
Assume the side condition $P \imp \eqbib{e}{e'}$ 
and premises 
$c\sep c' : \aespec{P\land \leftF{e}\land\rightF{e'}}{Q}$ and 
$d\sep d' : \aespec{P\land \neg\leftF{e}\land\neg\rightF{e'}}{Q}$.
The encoding of the conclusion is about
$\emb{e; c + \neg e;d}{e'; c' + \neg e'; d'}$.
The side condition yields $P = P;(\emb{e}{e'}+\emb{\neg e}{\neg e'})$.

By the theorem, there is some witness $Z$ for the first premise and $W$ for the second,
such that

(WCZ)
$P;\emb{e}{e'}; Z \leq P;\emb{e}{e'}; Z; Q$ 

(WUZ)
$P;\emb{e}{e'}; Z \leq \emb{\hav}{c'}$ 

(WOZ)
$P;\emb{e}{e'}; \emb{c}{1} \leq Z ; \emb{1}{\hav}$

(WCW)
$P;\emb{\neg e}{\neg e'}; W \leq P;\emb{\neg e}{\neg e'}; W; Q$ 

(WUW)
$P;\emb{\neg e}{\neg e'}; W \leq \emb{\hav}{d'}$ 

(WOW)
$P;\emb{\neg e}{\neg e'}; \emb{d}{1} \leq W ; \emb{1}{\hav}$ 

To show the conclusion we take as witness $X := \emb{e}{e'};Z + \emb{\neg e}{\neg e'};W$
and prove each of the conditions. 

\noindent $\bullet$ (WC) is 
$P;(\emb{e}{e'};Z + \emb{\neg e}{\neg e'}; W)
\leq  
P; (\emb{e}{e'};Z + \emb{\neg e}{\neg e'}; W); Q$ 
and we have
\[\begin{array}{lll}
  & P; 
  (\emb{e}{e'};Z + \emb{\neg e}{\neg e'};W) \\
= & P; (\emb{e}{e'} + \emb{\neg e}{\neg e'});
  (\emb{e}{e'};Z + \emb{\neg e}{\neg e'};W) 
  & \mbox{side condition} \\
= &   P; \emb{e}{e'};Z 
    + P; \emb{\neg e}{\neg e'};W & \mbox{distribute, cancel, absorb} \\
\leq &   P; \emb{e}{e'};Z; Q
    + P; \emb{\neg e}{\neg e'};W; Q & \mbox{(WCZ), (WCW)} \\
= &  P; (\emb{e}{e'};Z + \emb{\neg e}{\neg e'};W) ; Q &\mbox{symmetric steps}
\end{array}\]

\noindent $\bullet$ (WU) is $P;(\emb{e}{e'};Z + \emb{\neg e}{\neg e'}; W) \leq \emb{\hav}{e';c' + \neg e';d'}$
and we have
\[\begin{array}{lll}
    & P;(\emb{e}{e'};Z + \emb{\neg e}{\neg e'}; W) \leq \emb{\hav}{e';c' + \neg e';d'} \\
\iff& P;\emb{e}{e'};Z + P\emb{\neg e}{\neg e'}; W \leq \emb{\hav}{e';c' + \neg e';d'} &\mbox{distrib}\\ 
\iff& P;\emb{e}{e'};Z + P\emb{\neg e}{\neg e'}; W \leq \emb{\hav}{e';c'} + \emb{\hav}{\neg e';d'}&\mbox{emb 
homo}\\ 
\impby&      P;\emb{e}{e'};Z \leq \emb{\hav}{e';c'} 
       \land P\emb{\neg e}{\neg e'}; W \leq \emb{\hav}{\neg e';d'}&\mbox{$+$ mono} 
\end{array}\]
Note that (WUZ) is equivalent to 
$P;\emb{e}{e'}; Z \leq \emr{e'};\emb{\hav}{c'}$ (by $px\leq y\iff px\leq py$)
and then by LRC we get $\emr{e'};\emb{\hav}{c'} = \emb{\hav}{e';c'}$,
which yields the left conjunct above.
For the right conjunct, use (WUW) similarly.

\noindent $\bullet$ (WO) is 
$P;\emb{e;c + \neg e; d}{1} 
\leq 
(\emb{e}{e'};Z + \emb{\neg e}{\neg e'}; W);\emb{1}{\hav}$
and we have
\[\begin{array}{lll}
    & P;\emb{e;c + \neg e; d}{1} 
      \leq 
      (\emb{e}{e'};Z + \emb{\neg e}{\neg e'}; W);\emb{1}{\hav} \\
\iff& \mbox{side condition}\\
    & \emb{e}{e'}+\emb{\neg e}{\neg e'});P;\emb{e;c + \neg e; d}{1} 
      \leq 
      (\emb{e}{e'};Z + \emb{\neg e}{\neg e'}; W);\emb{1}{\hav} \\
\impby&  \mbox{distrib, cancel, $+$ mono} \\
      & \emb{e}{e'};P;\emb{e;c}{1} \leq \emb{e}{e'};Z;\emb{1}{\hav}
   \land\emb{\neg e}{\neg e'};P;\emb{\neg e;d}{1} \leq \emb{\neg e}{\neg e'};W;\emb{1}{\hav} \\
\iff & \mbox{test property $px \leq y \iff px \leq py$} \\
     & \emb{e}{e'};P;\emb{e;c}{1} \leq Z;\emb{1}{\hav}
   \land\emb{\neg e}{\neg e'};P;\emb{\neg e;d}{1} \leq W;\emb{1}{\hav}
\end{array}\]
The latter conjuncts are equivalent to (WOZ) and (WOW), using emb homo on the left.

\paragraph{Proof of \rn{eWh}}
Suppose we have $P \imp \eqbib{e}{e'}$ 
and $c\sep c' : \aespec{P\land \leftF{e}\land\rightF{e'}}{P}$.
Suppose we have alignment witness $W$ for this judgment, thus 

(WCW) $P;\emb{e}{e'}; W \leq P;\emb{e}{e'}; W; P$

(WUW) $P;\emb{e}{e'}; W \leq \emb{\hav}{c'}$ 

(WOW) $P;\emb{e}{e'}; \emb{c}{1} \leq W ; \emb{1}{\hav}$ 

To show $\whilec{e}{c} \Sep \whilec{e'}{c'} : \aespec{P}{P\land \neg\leftF{e}\land\neg\rightF{e'}}$,
take as alignment witness $Z := (\emb{e}{e'};W)^*;\eml{\neg e}$.

\noindent $\bullet$ (WC) is 
$ P; ( \emb{e}{e'}; W )^{\kstar}; \eml{\neg e}
\leq P; ( \emb{e}{e'}; W )^{\kstar}; \eml{\neg e} P \emb{\neg e}{\neg e'}
=  P; ( \emb{e}{e'}; W )^{\kstar}; \emb{\neg e}{\neg e'}$
and we have
\[\begin{array}{lll}
    & P; ( \emb{e}{e'}; W )^{\kstar}; \eml{\neg e} \\
=   & P; ( \emb{e}{e'}; W )^{\kstar}; P; \eml{\neg e} & \mbox{(WCW) and invariance} \\
=   & P; ( \emb{e}{e'}; W )^{\kstar}; P; \emb{\neg e}{\neg e} &
     \mbox{by side condition, $P \imp e \eqbi e'$}
\end{array}\]

\noindent $\bullet$ (WO) is
$P; \eml{( e;c )^{\kstar}; \neg e} \leq ( \emb{e}{e'}; W )^{\kstar}; \eml{\neg e}; \emr{\hav}$.  To establish this fact, we start by calculating
\[\begin{array}{lll}
    & P; \eml{( e;c )^{\kstar}; \neg e} \leq ( \emb{e}{e'}; W )^{\kstar}; \eml{\neg e}; \eml{\hav} \\
\iff & P; \eml{e; c}^{\kstar}; \eml{\neg e} \leq ( \emb{e}{e'}; W )^{\kstar}; \emr{\hav}; \eml{\neg e}
     & \mbox{emb homo and LRC} \\
\impby & P; \eml{e; c}^{\kstar} \leq ( \emb{e}{e'}; W )^{\kstar}; \emr{\hav}
     & \mbox{monotonicity} \\
\iff & P; \eml{e;c}^{\kstar} \leq P; ( \emb{e}{e'}; W )^{\kstar}; \emr{\hav}
     & \mbox{using $a;x \leq y \iff a;x \leq a;y$} \\
\impby & P + P; ( \emb{e}{e'}; W )^{\kstar}; \emr{\hav}; \eml{e;c} \leq P; ( \emb{e}{e'}; W )^{\kstar}; \emr{\hav} & \mbox{induction} \\
\impby & P; ( \emb{e}{e'}; W )^{\kstar}; \emr{\hav}; \eml{e;c} \leq P; ( \emb{e}{e'}; W )^{\kstar}; \emr{\hav} & \mbox{join, and $P \leq \mathrm{RHS}$} \\
\iff  & P; ( \emb{e}{e'}; W )^{\kstar}; \emb{e;c}{\hav} \leq P; ( \emb{e}{e'}; W )^{\kstar}; \emr{\hav} & \mbox{rewrite LHS}
\end{array}\]
Now consider the LHS.  We have:
\[\begin{array}{lll}
     &  P; ( \emb{e}{e'}; W )^{\kstar}; \emb{e;c}{\hav}  & \\
   = &  P; ( \emb{e}{e'}; W )^{\kstar}; P; \emb{e;c}{\hav} & \mbox{(WCW) and invariance} \\
   = &  P; ( \emb{e}{e'}; W )^{\kstar}; P; \emb{e;c}{e';\hav}  & \mbox{side condition, $P \imp e \eqbi e'$} \\
   = &  P; ( \emb{e}{e'}; W )^{\kstar}; P; \emb{e}{e'}; P; \emb{e;c}{e'} \emr{\hav} & \mbox{duplicate tests, emb homo} \\
\leq &  P; ( \emb{e}{e'}; W )^{\kstar}; P; \emb{e}{e'}; W; \emr{\hav} & \mbox{(WOW)} \\
\leq &  P; ( \emb{e}{e'}; W )^{\kstar}; \emb{e}{e'}; W; \emr{\hav}  & \mbox{drop test $P$} \\
\leq &  P; ( \emb{e}{e'}; W )^{\kstar}; \emr{\hav}  & \mbox{using $x^{\kstar};x\leq x^{\kstar}$}
  \end{array}\]
This is equal to the RHS and we are done.

\noindent $\bullet$ (WU) is $P; ( \emb{e}{e'}; W )^{\kstar}; \eml{\neg e} \leq \emb{\hav}{( e';c' )^{\kstar}; \neg e'}$.  We show
\[\begin{array}{lll}
       & P; ( \emb{e}{e'}; W )^{\kstar}; \eml{\neg e} \leq \emb{\hav}{( e';c' )^{\kstar}; \neg e'} \\ 
  \iff & \mbox{using lemma $p;x = p;x;p \imp p;x^{\kstar} = p;(x;p)^{\kstar}$ and (WCW)} \\
       & P; ( \emb{e}{e'}; W; P )^{\kstar}; \eml{\neg e} \leq \emb{\hav}{( e';c' )^{\kstar}; \neg e'} \\ 
  \iff & \mbox{sliding} \\
       & ( P; \emb{e}{e'}; W )^{\kstar}; P; \eml{\neg e} \leq \emb{\hav}{( e';c' )^{\kstar}; \neg e'} \\
 \iff  & \mbox{side condition, $P \imp e \eqbi e'$} \\
       & ( P; \emb{e}{e'}; W )^{\kstar}; P; \emb{\neg e}{\neg e'} \leq \emb{\hav}{( e';c' )^{\kstar}; \neg e'} \\
\impby & \mbox{induction, $x^{\kstar};y \leq z \impby y + x;z \leq z$} \\
       & P;\emb{\neg e}{\neg e'} + P; \emb{e}{e'}; W; \emb{\hav}{( e';c' )^{\kstar}; \neg e'} \leq  \emb{\hav}{( e';c' )^{\kstar}; \neg e'} \\
\end{array}\]

Since $P;\emb{\neg e}{\neg e'}$ is less than the RHS, by join, it suffices to show:
\[\begin{array}{lll}
 & P; \emb{e}{e'}; W; \emb{\hav}{( e';c' )^{\kstar}; \neg e} \leq  \emb{\hav}{( e';c' )^{\kstar}; \neg e'} \\
\iff & P; \emb{e}{e'}; \emr{e'}; W; \emb{\hav}{( e';c' )^{\kstar}; \neg e'} \leq  \emb{\hav}{( e';c' )^{\kstar}; \neg e'} & \mbox{dup test $\emr{e'}$} \\
\iff & \emr{e'}; P; \emb{e}{e'}; W; \emb{\hav}{( e';c' )^{\kstar}; \neg e'} \leq  \emb{\hav}{( e';c' )^{\kstar}; \neg e'} & \mbox{commute tests} \\
\impby & \emr{e'}; \emb{\hav}{c'}; \emb{\hav}{( e';c' )^{\kstar}; \neg e'} \leq  \emb{\hav}{( e';c' )^{\kstar}; \neg e'} & \mbox{using (WUW)} \\
\iff &  \emb{\hav}{e';c'}; \emb{\hav}{( e';c' )^{\kstar}; \neg e'} \leq  \emb{\hav}{( e';c' )^{\kstar}; \neg e'} & \mbox{emb homo} \\
\impby & \emr{e';c';( e';c' )^{\kstar}; \neg e'} \leq \emr{( e';c' )^{\kstar}; \neg e'} & \mbox{$\hav$ is top}
\end{array}\]
But this follows from the fact that $x;x^{\kstar}\leq x^{\kstar}$.

\paragraph{Proof of \rn{eWhL}}

Suppose we have $P \imp (\rightF{e'} \imp \leftF{e})$,
$c\sep c' : \aespec{P\land \leftF{e}\land\rightF{e'}}{P}$, and
$c\sep \skipc : \aespec{P\land \leftF{e}}{P}$.  From judgment for $c\sep c'$
obtain alignment witness $B$ and from judgment for $c\sep\skipc$ obtain
witness $L$.  We have,

(WCB) $P; \emb{e}{e'}; B \leq P; \emb{e}{e'}; B; P$

(WCL) $P; \emb{e}{1}; L \leq P; \emb{e}{1}; L; P$

(WUB) $P; \emb{e}{e'}; B \leq \emb{\hav}{c'}$

(WUL) $P; \emb{e}{1}; L \leq \emb{\hav}{1}$

(WOB) $P; \emb{e}{e'}; \emb{c}{1} \leq B; \emb{1}{\hav}$

(WOL) $P; \emb{e}{1}; \emb{c}{1} \leq L; \emb{1}{\hav}$

To show
$\whilec{e}{c} \Sep \whilec{e'}{c'} : \aespec{P}{P\land
  \neg\leftF{e}\land\neg\rightF{e'}}$ we take as alignment witness
$W := (\emb{e}{e'};B)^*;(\emb{e}{1}L)^*;\emb{\neg e}{1}$.

We have to show:

(WC) $P;W \leq P;W;P;\emb{\neg e}{\neg e'}$

(WU) $P;W \leq \emb{\hav}{(e'c')^*\neg e'}$

(WO) $P;\emb{(ec)^*\neg e}{1} \leq W; \emb{1}{\hav}$

\noindent $\bullet$ For (WC), we have:
\[
\begin{array}[]{lll}
     & P; ( \emb{e}{e'}; B )^{\kstar}; ( \eml{e}; L )^{\kstar}; \eml{\neg e}  & \\
\leq & P; ( \emb{e}{e'}; B )^{\kstar}; P; ( \eml{e}; L )^{\kstar}; \eml{\neg e}
     & \mbox{using (WCB) and invariance} \\
\leq & P; ( \emb{e}{e'}; B )^{\kstar}; P; ( \eml{e}; L )^{\kstar}; P; \eml{\neg e}
     & \mbox{using (WCL) and invariance} \\
  =  & P; ( \emb{e}{e'}; B )^{\kstar}; P; ( \eml{e}; L )^{\kstar}; P; \emb{\neg e}{\neg e'}
     & \mbox{using side condition $P \imp e \eqbi e'$} \\
\leq & P; ( \emb{e}{e'}; B )^{\kstar}; ( \eml{e}; L )^{\kstar}; P; \emb{\neg e}{\neg e'}
     & \mbox{drop test $P$}
\end{array} 
\]

\noindent $\bullet$ For (WU), we have to show $P; ( \emb{e}{e'} B )^{\kstar}; ( \eml{e}; L )^{\kstar}; \eml{\neg e} \leq \emb{\hav}{( e';c' )^{\kstar};\neg e'}$.  We calculate as follows,
\[
\begin{array}{lll}
     & P; ( \emb{e}{e'} B )^{\kstar}; ( \eml{e}; L )^{\kstar}; \eml{\neg e} \leq \emb{\hav}{( e';c' )^{\kstar};\neg e'} \\ 
\impby& \mbox{(WCB) with Lemma~\ref{lem:L2}}\\
     & P; ( \emb{e}{e'} B; P )^{\kstar}; ( \eml{e}; L )^{\kstar}; \eml{\neg e} \leq \emb{\hav}{( e';c' )^{\kstar};\neg e'} \\
\iff & \mbox{sliding} \\
     & ( P; \emb{e}{e'} B; )^{\kstar}; P ; ( \eml{e}; L )^{\kstar}; \eml{\neg e} \leq \emb{\hav}{( e';c' )^{\kstar};\neg e'} \\
\iff & \mbox{(WCL) and invariance} \\
     & ( P; \emb{e}{e'} B; )^{\kstar}; P ; ( \eml{e}; L )^{\kstar}; P; \eml{\neg e} \leq \emb{\hav}{( e';c' )^{\kstar};\neg e'} \\
\iff & \mbox{by side condition $P \imp e \eqbi e'$, whence $P\eml{\neg e} = P\emb{\neg e}{\neg e'}$} \\
     & ( P; \emb{e}{e'} B; )^{\kstar}; P ; ( \eml{e}; L )^{\kstar}; P; \emb{\neg e}{\neg e'} \leq \emb{\hav}{( e';c' )^{\kstar};\neg e'} \\
\impby &  \mbox{induction, $x^{\kstar};y\leq z \impby y + x;z\leq z$} \\ 
       & P ( \eml{e}; L )^{\kstar}; P; \emb{\neg e}{\neg e'} + P; \emb{e}{e'}; B; \emb{\hav}{( e';c' )^{\kstar}; \neg e'} 
\leq \emb{\hav}{( e';c' )^{\kstar};\neg e'}
\end{array}
\]

By join, it suffices to check two conditions:
\[\begin{array}{lll}
      & P; ( \eml{e}; L )^{\kstar}; P; \emb{\neg e}{\neg e'} \leq \emb{\hav}{( e';c' )^{\kstar}; \neg e'}  \\
\impby& \mbox{(WCL) using Lemma~\ref{lem:L2}} \\
      & P; ( \eml{e}; L; P )^{\kstar}; P; \emb{\neg e}{\neg e'} \leq \emb{\hav}{( e';c' )^{\kstar}; \neg e'} \\
\iff  & \mbox{sliding and test idem}\\
      & ( P; \eml{e}; L; )^{\kstar}; P; \emb{\neg e}{\neg e'} \leq \emb{\hav}{( e';c' )^{\kstar}; \neg e'} \\
\impby& \mbox{induction}\\
      & P; \emb{\neg e}{\neg e'} + P; \eml{e}; L; \emb{\hav}{( e';c' )^{\kstar}; \neg e'} \leq \emb{\hav}{( e';c' )^{\kstar}; \neg e'} \\
\impby& \mbox{join along with $P; \emb{\neg e}{\neg e'} \leq \mathrm{RHS}$}\\
      & P; \eml{e}; L; \emb{\hav}{( e';c' )^{\kstar}; \neg e'} \leq \emb{\hav}{( e';c' )^{\kstar}; \neg e'} \\
\impby& \mbox{(WUL)} \\
      & \eml{\hav}; \emb{\hav}{( e';c' )^{\kstar}; \neg e'} \leq \emb{\hav}{( e';c' )^{\kstar}; \neg e'} 
  \end{array}\]
This follows from $\hav$ being top and reflexivity.  For the other case, we have:
\[\begin{array}{lll}
    & P; \emb{e}{e'}; B; \emb{\hav}{( e';c' )^{\kstar}; \neg e'} \leq
   \emb{\hav}{( e';c' )^{\kstar};\neg e'} \\
\iff & \emr{e'}; P; \emb{e}{e'}; B; \emb{\hav}{( e';c' )^{\kstar}; \neg e'} \leq
   \emb{\hav}{( e';c' )^{\kstar};\neg e'} & \mbox{tests idem, tests commute}\\
\impby & \emr{e'}; \emb{\hav}{c'}; \emb{\hav}{( e';c' )^{\kstar}; \neg e'} \leq
   \emb{\hav}{( e';c' )^{\kstar};\neg e'} & \mbox{(WUB)}\\
\iff & \emb{\hav}{e';c'}; \emb{\hav}{( e';c' )^{\kstar}; \neg e'} \leq
   \emb{\hav}{( e';c' )^{\kstar};\neg e'} & \mbox{LRC, emb homo}\\
\impby & e';c'; ( e';c' )^{\kstar};\neg e' \leq ( e';c' )^{\kstar};\neg e'
  \end{array}
\]
And the last fact follows from $x; x^{\kstar}\leq x^{\kstar}$.

\noindent $\bullet$ Finally, for (WO) we have to show $P;\emb{( e;c )^{\kstar};\neg e}{1} \leq ( \emb{e}{e'}; B )^{\kstar}; ( \eml{e}; L )^{\kstar}; \emb{\neg e}{\hav}$.  We argue
\[\begin{array}{lll}
\iff  & \mbox{by $a;x\leq y \iff a;x\leq a;y$} \\ 
      & P;\emb{( e;c )^{\kstar};\neg e}{1} \leq P; ( \emb{e}{e'}; B )^{\kstar}; ( \eml{e}; L )^{\kstar}; \emb{\neg e}{\hav} \\
\iff  & \mbox{emb homo and LRC} \\
      & P; \eml{e;c}^{\kstar}; \eml{\neg e} \leq P; ( \emb{e}{e'}; B )^{\kstar}; ( \eml{e}; L )^{\kstar}; \emr{\hav}; \eml{\neg e}  \\
\impby& \mbox{monotonicity} \\
& P; \eml{e;c}^{\kstar} \leq P; ( \emb{e}{e'}; B )^{\kstar}; ( \eml{e}; L )^{\kstar}; \emr{\hav} \\
\impby& \mbox{induction}\\
& P + P; ( \emb{e}{e'}; B )^{\kstar}; ( \eml{e}; L )^{\kstar}; \emr{\hav}; \eml{e;c} \leq P; ( \emb{e}{e'}; B )^{\kstar}; ( \eml{e}; L )^{\kstar}; \emr{\hav} \\
\impby& \mbox{by join and $P \leq \mathrm{RHS}$}\\
      & P; ( \emb{e}{e'}; B )^{\kstar}; ( \eml{e}; L )^{\kstar}; \emr{\hav}; \eml{e;c} \leq P; ( \emb{e}{e'}; B )^{\kstar}; ( \eml{e}; L )^{\kstar}; \emr{\hav} \\
\iff  & \mbox{$\hav$ idem, emb homo}\\
      & P; ( \emb{e}{e'}; B )^{\kstar}; ( \eml{e}; L )^{\kstar}; \emr{\hav}; \eml{e;c} \leq P; ( \emb{e}{e'}; B )^{\kstar}; ( \eml{e}; L )^{\kstar}; \emr{\hav}; \emr{\hav} \\
\iff  & \mbox{LRC}\\
      & P; ( \emb{e}{e'}; B )^{\kstar}; ( \eml{e}; L )^{\kstar}; \eml{e;c}; \emr{\hav} \leq P; ( \emb{e}{e'}; B )^{\kstar}; ( \eml{e}; L )^{\kstar}; \emr{\hav}; \emr{\hav} \\
\impby& \mbox{monotonicity}\\
& P; ( \emb{e}{e'}; B )^{\kstar}; ( \eml{e}; L )^{\kstar}; \eml{e;c} \leq P; ( \emb{e}{e'}; B )^{\kstar}; ( \eml{e}; L )^{\kstar}; \emr{\hav}
  \end{array}\]

Now, we show
\[\begin{array}{lll}
  & P; ( \emb{e}{e'}; B )^{\kstar}; ( \eml{e}; L )^{\kstar}; \eml{e;c} \\

\leq & P; ( \emb{e}{e'}; B )^{\kstar}; P; ( \eml{e}; L )^{\kstar}; \eml{e;c} & \mbox{using (WCB) and invariance}\\

\leq & P; ( \emb{e}{e'}; B )^{\kstar}; P; ( \eml{e}; L )^{\kstar}; P; \eml{e;c} & \mbox{using (WCL) and invariance}\\

\leq & P; ( \emb{e}{e'}; B )^{\kstar}; ( \eml{e}; L )^{\kstar}; P; \eml{e;c} & \mbox{dropping test $P$ in the middle}\\

= & P; ( \emb{e}{e'}; B )^{\kstar}; ( \eml{e}; L )^{\kstar}; \eml{e}; P; \eml{e;c} & \mbox{duplicating $\eml{e}$ and commuting}\\

\leq & P; ( \emb{e}{e'}; B )^{\kstar}; ( \eml{e}; L )^{\kstar}; \eml{e}; L; \emr{\hav} & \mbox{using (WOL)}\\

\leq & P; ( \emb{e}{e'}; B )^{\kstar}; ( \eml{e}; L )^{\kstar} \emr{\hav} & \mbox{since $x^{\kstar}; x \leq x^{\kstar}$}
  \end{array}\]
And now we are done since this is equal to the RHS.

\paragraph{Proof of \rn{eDisj}}

We give two proofs.  

First, here is a proof directly in terms of the TriKAT formulation (\ref{eq:bitri}).
We use that $\Lproj$ and sequence distribute over $+$, as does the embedding:
\begin{equation}\label{eq:triEmbDisj}
\triEmb{A+B}{C} = \triEmb{A}{C} + \triEmb{B}{C} \qquad\mbox{and}\qquad
\triEmb{A}{B+C} = \triEmb{A}{B} + \triEmb{B}{C} 
\end{equation}
These facts are easily proved for any relational model (not necessarily full).
These conditions are candidates for an axiomatization of TriKAT.

The premises are $  c\sep d : \aespec{P}{Q} $ and $  c\sep d : \aespec{R}{Q} $.
Here are the premises expressed in the form given by (\ref{eq:bitri}):
\[ 
\begin{array}{l}
  \dot{P} ; \emb{c}{\hav} ; \dot{id} 
  \;\leq\;
  \Lproj ( \triEmb{\dot{P}}{\dot{id}} ; \tricom{c}{\hav}{d} ; \triEmb{\dot{id}}{\dot{Q}}) 
\\ 
  \dot{R} ; \emb{c}{\hav} ; \dot{id} 
  \;\leq\;
  \Lproj ( \triEmb{\dot{R}}{\dot{id}} ; \tricom{c}{\hav}{d} ; \triEmb{\dot{id}}{\dot{Q}}) 
  \end{array}
\]
The conclusion $ c\sep d : \aespec{P\lor R}{Q} $ is expressed as 
\[ 
  (\dot{P}+\dot{R}) ; \emb{c}{\hav} ; \dot{id} 
  \;\leq\;
  \Lproj ( \triEmb{(\dot{P}+\dot{R})}{\dot{id}} ; \tricom{c}{\hav}{d} ; \triEmb{\dot{id}}{\dot{Q}}) 
\]
By disjunctivity, LHS equals 
\[ 
  \dot{P} ; \emb{c}{\hav} ; \dot{id} +   \dot{R} ; \emb{c}{\hav} ; \dot{id} 
\] 
and by disjunctivity, including the left equation of (\ref{eq:triEmbDisj}),
RHS equals 
\[ 
  \Lproj ( \triEmb{\dot{P}}{\dot{id}} ; \tricom{c}{\hav}{d} ; \triEmb{\dot{id}}{\dot{Q}}) 
+ \Lproj ( \triEmb{\dot{R}}{\dot{id}} ; \tricom{c}{\hav}{d} ; \triEmb{\dot{id}}{\dot{Q}}) 
\]
Now LHS$\leq$RHS follows directly from the premises by monotonicity of $+$.

Second, we give a proof of \rn{eDisj} using witnesses.  Let $W$ and $X$ be witnesses of 
the first and second premise, so we have

\begin{tabular}{llll}
(WCW) & $\dot{P}; W \leq \dot{P};W; \dot{Q}$ & 
(WCX) & $\dot{R}; X \leq \dot{R};X; \dot{Q}$ \\

(WUW) & $\dot{P}; W \leq \emb{\hav}{d}$ &
(WUX) & $\dot{R}; X \leq \emb{\hav}{d}$ \\

(WOW) & $\dot{P}; \eml{c} \leq W ; \emr{\hav}$ &
(WOX) & $\dot{R}; \eml{c} \leq X ; \emr{\hav}$ 
\end{tabular}
\\
Let $Z \eqdef \dot{P};W + \dot{R};X$.  We show that $Z$ satisfies the witness conditions for the conclusion.

\noindent$\bullet$ (WC) 
\[\begin{array}{lll}
    & (\dot{P}+\dot{R}) ; (\dot{P};W + \dot{R};X) \\
=   & \dot{P};W + \dot{P};\dot{R};X + \dot{R};\dot{P};W + \dot{R};X &\mbox{distrib, tests idem} \\
\leq& \dot{P};W;Q + \dot{P};\dot{R};X;Q + \dot{R};\dot{P};W;Q + \dot{R};X;Q & \mbox{(WCW) and (WCX)}\\
=  & (\dot{P}+\dot{R}) ; (\dot{P};W + \dot{R};X); \dot{Q} &\mbox{distrib, tests idem}
\end{array}\]

\noindent$\bullet$ (WU) 
The obligation $(\dot{P}+\dot{R}) ; (\dot{P};W + \dot{R};X) \leq \emb{\hav}{d} $
is equivalent, by distributivity and tests idempotent, to 
\[  \dot{P};W + \dot{P};\dot{R};X + \dot{R};\dot{P};W + \dot{R};X \leq \emb{\hav}{d}  \]
By the join property of $+$ this is equivalent to four inequalities,
all consequences of (WUW) or (WUX), for example 
$\dot{P};\dot{R};X \leq \dot{R};X \leq \emb{\hav}{d}$ using (WUX).

\noindent$\bullet$ (WO) 
By distribution, the obligation
$(\dot{P}+\dot{R}) ; \eml{c} \leq (\dot{P};W + \dot{R};X);\emr{\hav}$
is equivalent to 
\[ \dot{P};\eml{c} +\dot{R};\eml{c} \leq \dot{P};W;\emr{\hav} + \dot{R};X;\emr{\hav} \]
which follows from (WOW) and (WOX) by monotonicity of $+$.

\subsubsection{Proofs for Backward Simulation Rules}

Fig.~\ref{fig:beRHL} gives some backward simulation rules.
% only one is currently proved in the body

\paragraph{Sequential Composition \rn{bSeq}}

The rule \rn{bSeq} looks the same as \rn{dSeq} and \rn{eSeq} except that it is about the
backwards judgment.

Suppose for premise $c\sep c' : \bespec{P}{R}$ we have witness $Z$ and 
for $d\sep d' : \bespec{R}{Q}$ we have witness $W$, so the conditions are 

(WCbZ)
$Z;R \leq P;Z; R$ 

(WUbZ)
$Z;R \leq \emb{\hav}{c'}$ 

(WObZ)
$\emb{c}{1}; R \leq \emb{1}{\hav};Z$

(WCbW)
$W;Q \leq R; W; Q$ 

(WUbW)
$ W;Q \leq \emb{\hav}{d'}$ 

(WObW)
$\emb{d}{1};Q \leq \emb{1}{\hav};W$ 

To prove the conclusion $c;d \Sep c';d' : \aespec{P}{Q}$ we use $Z;W$ as witness, 
showing it satisfies the three conditions.

\noindent $\bullet$ (WCb) To show $Z;W;Q \leq P;Z;W;Q$ we have
$ Z;W;Q \leq Z;R;W;Q \leq P;Z;R;W;Q \leq P;Z;W;Q$ using (WCbZ), (WCbW), and test below 1.

\noindent $\bullet$ (WUb)
\[\begin{array}{lll}
    & Z; W; Q \leq \emb{\hav}{c;d'} \\
\iff& Z;R;W;Q \leq \emb{\hav}{c;d'} & \mbox{using (WCbW) in equality form} \\
\iff& Z;R;W;Q \leq \emb{\hav}{c};\emb{\hav}{d'} & \mbox{$\hav$ idem, emb homo} \\
\end{array}\]
and the last line follows from (WUbZ) and (WUbW) using monotonicity of sequence.

\noindent $\bullet$ (WOb)  
\[\begin{array}{lll}
    & \emb{c;d}{1};Q \\
=   & \emb{c}{1}; \emb{d}{1}; Q;Q & \mbox{emb homo, test idem} \\
\leq& \emb{c}{1}; \emb{1}{\hav};W;Q & \mbox{(WObW)} \\
=   & \emb{c}{1}; \emb{1}{\hav};R;W;Q & \mbox{(WCbW) as equality} \\
=   & \emb{1}{\hav}; \emb{c}{1};R;W;Q & \mbox{LRC} \\
\leq& \emb{1}{\hav}; \emb{1}{\hav};Z;W;Q & \mbox{(WObZ)} \\
\leq& \emb{1}{\hav}; Z;W & \mbox{$\hav$ idem, emb homo, test below 1}
\end{array}\] 

\paragraph{Proof of Loop Rule \rn{bWh}} 

Suppose we have witness $W$ for the premise
$c\sep c' : \bespec{P\land \leftF{e}\land\rightF{e'}}{P} $, with
\[ \begin{array}{lll}
(WCbW)&  W; P \leq P; \emb{e}{e'}; W; P  &\mbox{in brief: } W P \leq P e e' W P \\
(WUbW)&  W; P \leq \emb{\hav}{c'}        &\mbox{in brief: } W P \leq \eml{\hav} c' \\
(WObW)&  \eml{c}; P \leq \emr{\hav}; W; P &\mbox{in brief: } c P \leq \emr{\hav} W P
\end{array}\]

For loop witness we choose $Z := ( \emb{e}{e'}; W; P )^{\kstar}; \emb{\neg e}{\neg e'}$.

\noindent $\bullet$ (WCb) for conclusion is
$ ( \emb{e}{e'}; W; P )^{\kstar}; \emb{\neg e}{\neg e'}; P
\leq P; ( \emb{e}{e'}; W; P )^{\kstar}; \emb{\neg e}{\neg e'}; P$.
We have,
\[\begin{array}{lll}

     & ( \emb{e}{e'}; W; P )^{\kstar}; \emb{\neg e}{\neg e'}; P
       \leq P; ( \emb{e}{e'}; W; P )^{\kstar};
               \emb{\neg e}{\neg e'}; P \\

\iff & ( \emb{e}{e'}; W; P )^{\kstar}; P; \emb{\neg e}{\neg e'}
       \leq P; ( \emb{e}{e'}; W; P )^{\kstar};
               P; \emb{\neg e}{\neg e'}
     & \mbox{tests commute}\\

\impby & ( \emb{e}{e'}; W; P )^{\kstar}; P
       \leq P; ( \emb{e}{e'}; W; P )^{\kstar}; P
  \end{array}\]
which holds by backwards invariance Lemma~\ref{lem:binvar} using
$\emb{e}{e'}; W; P \leq P; \emb{e}{e'}; W; P$ which follows from (WCbW).

\noindent $\bullet$ (WUb) for conclusion is
$ ( \emb{e}{e'}; W; P )^{\kstar}; \emb{\neg e}{\neg e'}; P \leq 
  \emb{\hav}{( e';c' )^{\kstar}; \neg e'}$.  We have,
\[
\begin{array}{lll}
& ( \emb{e}{e'}; W; P )^{\kstar}; \emb{\neg e}{\neg e'}; P \leq 
  \emb{\hav}{( e';c' )^{\kstar}; \neg e'} \\

\impby & 
         \emb{\neg e}{\neg e'}; P + \emb{e}{e'}; W; P;
         \emb{\hav}{( e';c' )^{\kstar}; \neg e'} \leq
            \emb{\hav}{( e';c' )^{\kstar}; \neg e'} & 
         \mbox{induction}
\end{array}
\]
Using join, for the first conjunct
$\emb{\neg e}{\neg e'}; P \leq \emb{\hav}{( e';c' )^{\kstar}; \neg
  e'}$ holds using 1 below star and $\hav$ is top.  For the second conjunct, we have
\[
\begin{array}{lll}
& \emb{e}{e'}; W; P; \emb{\hav}{( e';c' )^{\kstar}; \neg e'} \\

\leq & \emb{e}{e'}; \emb{\hav}{c'}; \emb{\hav}{( e';c' )^{\kstar}; \neg e'} & \mbox{using WUbW} \\

=    &  \eml{e}; \eml{\hav}; \emr{e';c'}; \emr{( e';c' )^{\kstar}; \neg e'} & \mbox{LRC, $\hav$ idempotent} \\

\leq & \eml{e}; \eml{\hav}; \emr{( e';c' )^{\kstar}; \neg e'}
     & \mbox{using $x;x^{\star}\leq x^{\kstar}$} \\

\leq & \eml{\hav}; \emr{( e';c' )^{\kstar}; \neg e'}
     & \mbox{test below 1}
\end{array}
\]

\noindent $\bullet$ (WOb) for the conclusion is
$\eml{( e;c )^{\kstar}; \neg e}; \emb{\neg e}{\neg e'}; P \leq \emr{\hav}; ( \emb{e}{e'}; W; P )^{\kstar}; \emb{e}{e'}; P$.
We have,
\[\begin{array}{lll}
       & \eml{( e;c )^{\kstar}; \neg e}; \emb{\neg e}{\neg e'}; P \leq \emr{\hav}; ( \emb{e}{e'}; W; P )^{\kstar}; \emb{e}{e'}; P \\
\impby & \mbox{induction, $a^{\kstar};b\leq x$ from $b+a;x\leq x$} \\
       & \emb{\neg e}{\neg e'}; P + \eml{e;c}; \emr{\hav}; ( \emb{e}{e'}; W; P )^{\kstar}; \emb{\neg e}{\neg e'}; P \leq \emr{\hav}; ( \emb{e}{e'}; W; P )^{\kstar}; \emb{e}{e'}; P 
\end{array}\]
Using join and the fact that tests are below 1, star is above 1, and $\hav$
top, it suffices to show the second conjunct is less than or equal to the RHS.
\[\begin{array}{lll}
& \eml{e;c}; \emr{\hav}; ( \emb{e}{e'}; W; P )^{\kstar}; \emb{\neg e}{\neg e'}; P & \\

= & \eml{e;c}; \emr{\hav}; ( \emb{e}{e'}; W; P )^{\kstar}; P; \emb{\neg e}{\neg e'} & \mbox{tests commute} \\

= & \eml{e;c}; \emr{\hav}; P; ( \emb{e}{e'}; W; P )^{\kstar}; P; \emb{\neg e}{\neg e'} & \mbox{backward invariance} \\

= & \emr{\hav}; \eml{e;c}; P; ( \emb{e}{e'}; W; P )^{\kstar}; P; \emb{\neg e}{\neg e'} & \mbox{LRC} \\

\leq & \mbox{using (WObW) in equivalent form $\eml{c};P \leq \emr{\hav};W;P$} \\
     & \emr{\hav}; \eml{e}; \emr{\hav}; W; P; ( \emb{e}{e'}; W; P )^{\kstar}; P; \emb{\neg e}{\neg e'}
\\ 

= & \emr{\hav}; \eml{e}; \emr{\hav}; P; \emb{e}{e'}; W; P; ( \emb{e}{e'}; W; P )^{\kstar}; P; \emb{\neg e}{\neg e'} & \mbox{using (WCbW) as an equality} \\

= & \emr{\hav}; P; \emb{e}{e'}; W; P; ( \emb{e}{e'}; W; P )^{\kstar}; P; \emb{\neg e}{\neg e'} & \mbox{LRC, tests and $\hav$ idem} \\

\leq & \emr{\hav}; P; ( \emb{e}{e'}; W; P )^{\kstar}; P; \emb{\neg e}{\neg e'} & \mbox{star fold} \\

\leq & \emr{\hav}; ( \emb{e}{e'}; W; P )^{\kstar}; P; \emb{\neg e}{\neg e'} & \mbox{test below 1}
\end{array}\]

\bigskip

\paragraph{Proof of \rn{bIf}}

The rule is like \rn{dIf} and \rn{eIf} except no side condition is needed.

Suppose the premises 
$c\sep c' : \bespec{P\land \leftF{e}\land\rightF{e'}}{Q}$ and 
$d\sep d' : \bespec{P\land \neg\leftF{e}\land\neg\rightF{e'}}{Q}$
are witnessed by $Z$ and $W$,  so we have

(WCbZ)
$Z;Q \leq P;\emb{e}{e'}; Z; Q$ 

(WUbZ)
$Z;Q \leq \emb{\hav}{c'}$ 

(WObZ)
$\eml{c};Q \leq \emr{\hav};Z$ 

(WCbW)
$W;Q \leq P;\emb{\neg e}{\neg e'}; W; Q$ 

(WUbW)
$W;Q \leq \emb{\hav}{d'}$ 

(WObW)
$\eml{d};Q \leq \emr{\hav};W$ 

As witness we choose $\emb{e}{e'};Z + \emb{\neg e}{\neg e'};W$.

\noindent $\bullet$ To prove (WCb) for the conclusion:
\[\begin{array}{lll} 
    & (\emb{e}{e'} Z + \emb{\neg e}{\neg e'} W) Q \\
=   & \emb{e}{e'} Z Q + \emb{\neg e}{\neg e'} W Q &\mbox{distrib} \\
\leq& P \emb{e}{e'} Z Q + P \emb{\neg e}{\neg e'} W Q &\mbox{(WCbZ), (WCbW), tests idem, $+$ mono} \\

= & P (\emb{e}{e'} Z + \emb{\neg e}{\neg e'} W) Q &\mbox{distrib} 
\end{array}\]

\noindent $\bullet$ To prove (WUb) for the conclusion:
\[\begin{array}{lll}  
    & (\emb{e}{e'} Z + \emb{\neg e}{\neg e'} W) Q \leq \emb{\hav}{e' c' + \neg e' d'} \\
\iff& \emb{e}{e'} Z Q + \emb{\neg e}{\neg e'} W Q \leq \emb{\hav}{e' c'} + \emb{\hav}{\neg e' d'} &\mbox{distrib, emb homo}\\
\iff& \emb{e}{e'} Z Q + \emb{\neg e}{\neg e'} W Q \leq \emr{e'} \emb{\hav}{c'} + \emr{\neg e'} \emb{\hav}{d'} &\mbox{emb homo, LRC}\\
\impby& \emb{e}{e'} \emb{\hav}{c'} + \emb{\neg e}{\neg e'} \emb{\hav}{d'} \leq \emr{e'} \emb{\hav}{c'} + \emr{\neg e'} \emb{\hav}{d'} 
&\mbox{using (WUbZ), (WUbW)}
\end{array}\]
which holds by tests below 1 and emb homo: $\emb{e}{e'}\leq\emr{e'}$ and $\emb{\neg e}{\neg e'}\leq\eml{e}$.

\noindent $\bullet$ To prove (WOb) for the conclusion,
we use (WObZ) and (WObW) in the equivalent forms 
$\eml{c};Q \leq \emr{\hav};Z;Q$ 
and 
$\eml{d};Q \leq \emr{\hav};W;Q$ 
(by $xq\leq y \iff xq \leq yq$).
 \[\begin{array}{lll}  
     & \eml{ec + \neg ed};Q \leq \emr{\hav};(\emb{e}{e'};Z + \emb{\neg e}{\neg e'};W) \\
\iff &\mbox{distrib, emb homo}\\
& \eml{e};\eml{c};Q + \eml{\neg e};\eml{d};Q \leq \emr{\hav};P;\emb{e}{e'};Z + \emr{\hav};P;\emb{\neg e}{\neg e'};W \\
\impby &\mbox{(WObZ), (WObW)}\\
& \eml{e};\emr{\hav};Z;Q + \eml{\neg e};\emr{\hav};W;Q \leq \emr{\hav};\emb{e}{e'};Z + \emr{\hav};\emb{\neg e}{\neg e'};W \\
\iff &\mbox{(LRC)}\\
& \emr{\hav};\eml{e};Z;Q + \emr{\hav};\eml{\neg e};W;Q \leq \emr{\hav};\emb{e}{e'};Z + \emr{\hav};\emb{\neg e}{\neg e'};W \\
\iff &\mbox{(WCbZ),(WCbW)} \\
& \emr{\hav};\emb{e}{e'};Z;Q + \emr{\hav};\emb{\neg e}{\neg e'};W;Q \leq \emr{\hav};\emb{e}{e'};Z + \emr{\hav};\emb{\neg e}{\neg e'};W 
\end{array}\]
The last step uses (WCbZ) and (WCbW) in equality form.
It holds by test below 1.

\paragraph{Proof of \rn{bDisj}}

This can be proved by an argument very similar to the proof of \rn{eDisj},
using (\ref{eq:bitriBack}) and the right equation of (\ref{eq:triEmbDisj}).

\paragraph{Proof of \rn{bnAss}}
This works very similarly to the forward rule.
Let witness $W$ to be $\dot{R};\emb{x:=\kcode{any}}{y:=\kcode{any}}$ and recall $true$ is encoded as $\bone$.

$\bullet$ (WCb) 
$\dot{R};\emb{x:=\kcode{any}}{y:=\kcode{any}};\bone 
\leq \dot{R}; \dot{R}; \emb{x:=\kcode{any}}{y:=\kcode{any}};\bone$
holds  by idempotence of tests.

$\bullet$ (WUb) 
$\dot{R};\emb{x:=\kcode{any}}{y:=\kcode{any}}; \bone
\leq 
\emb{\hav}{y:=\kcode{any}}$
holds by $\bone$ identity, $\hav$ top, $\eml{-}$ monotonic, and $\dot{R}\leq \bone$.

$\bullet$ (WOb) is $\eml{x:=\kcode{any}};\bone \leq
\emr{\hav};\dot{R};\emb{x:=\kcode{any}}{y:=\kcode{any}}$ which holds because
\[\begin{array}{lll}
    & \eml{x:=\kcode{any}};\bone \\
=   & \bone; \eml{x:=\kcode{any}} \\
\leq& \emr{y:=\kcode{any}};\dot{R};\emr{y:=\kcode{any}}; \eml{x:=\kcode{any}} &\mbox{using antecedent  $\bone \leq \emr{y:=\kcode{any}};\dot{R};\emr{y:=\kcode{any}}$}\\
=   & \emr{y:=\kcode{any}};\dot{R};\emb{x:=\kcode{any}}{y:=\kcode{any}} \\
\leq& \emr{\hav};\dot{R};\emb{x:=\kcode{any}}{y:=\kcode{any}} 
\end{array}\]

\subsubsection{Details for Sect.~\ref{sec:trikat}}

\paragraph{Proof of Lemma~\ref{lem:trikat}}
Consider the inequality 
\[
  \dot{R} ; \emb{c}{\hav} ; \dot{id} 
  \;\leq\;
  \Lproj ( \triEmb{\dot{R}}{\dot{id}} ; \tricom{c}{\hav}{d} ; \triEmb{\dot{id}}{\dot{S}}) 
\]
Observe that by definitions, we have (for any $\sigma,\sigma',\sigma'',\tau,\tau',\tau''$)
\[ (\sigma,\sigma',\sigma'')(\triEmb{\dot{R}}{\dot{id}} ; \tricom{c}{\hav}{d} ; \triEmb{\dot{id}}{\dot{S}})(\tau,\tau',\tau'') \]
iff $\sigma R\sigma' \land \sigma'=\sigma'' \land \sigma c \tau \land 
\sigma'' d \tau'' \land \tau=\tau' \land \tau' S \tau''$. 
\\
Also $(\sigma,\sigma')(\dot{R} ; \emb{c}{\hav} ; \dot{id})(\tau,\tau')$
iff $\sigma R \sigma' \land \sigma c \tau \land \tau=\tau'$.  
So the displayed inequality says 
\[ \all{\sigma,\sigma',\tau,\tau'}{
   \sigma R \sigma' \land \sigma c \tau \land \tau=\tau'
  \imp\some{\sigma'',\tau''}{
  \sigma R\sigma' \land \sigma'=\sigma'' \land \sigma c \tau \land 
\sigma'' d \tau'' \land \tau=\tau' \land \tau' S \tau'' }}\] 
By the one-point rule of predicate calculus, for $\tau'$, this is equivalent to 
\[ \all{\sigma,\sigma',\tau}{
   \sigma R \sigma' \land \sigma c \tau 
  \imp\some{\sigma'',\tau''}{
  \sigma R\sigma' \land \sigma'=\sigma'' \land \sigma c \tau \land 
\sigma'' d \tau'' \land \tau S \tau'' }}\] 
Using again the one-point rule, for $\sigma''$, this is equivalent to 
\[ \all{\sigma,\sigma',\tau}{
   \sigma R \sigma' \land \sigma c \tau 
  \imp\some{\tau''}{
  \sigma R\sigma' \land \sigma c \tau \land 
\sigma' d \tau'' \land \tau S \tau'' }}\] 
By predicate calculus this is equivalent to 
\[ \all{\sigma,\sigma',\tau}{
   \sigma R \sigma' \land \sigma c \tau 
  \imp\some{\tau''}{\sigma' d \tau'' \land \tau S \tau'' }}\] 
which is the definition of $c\sep d:\aespec{R}{S}$ 
(except the identifier $\tau'$ is used in (\ref{eq:fsim}) rather than $\tau''$).

The proof that (\ref{eq:bsim}) is equivalent to (\ref{eq:bitriBack}) is similar.

\subsection{Additional Examples for Sect.~\ref{sec:beyond}}

% Note \nn used below is defined in the Examples of algebraic alignment section.
\newcommand\nx{\kcode{x}}
\newcommand\nz{\kcode{z}}
\newcommand\nt{\kcode{t}}

This section considers additional examples for $\forall\exists$ properties.

\subsubsection{Forward Simulation Example~\ref{eg:fsimProphecyEx} Using Witness Technique}

Consider the following two programs, where $\nx$, $\nt$, $\kcode{s}$, and
$\nz$ range over natural numbers:
\[\begin{array}{ll}
    C_1: & \kcode{while x>n do x:=x-1 od; t:=any+x; z:=x+t} \\
    C_2: & \kcode{s:=any; while x>n do x:=x-1 od; z:=x+s}
  \end{array}
\]
We want to show
$C_1\sep C_2:\aespec{R}{S}$, where
$R \eqdef \left(\nx\eqbi\nx \land \nn\eqbi\nn \right)$ and
$S \eqdef \left( \nz\eqbi\nz \right)$.
To do so effectively, we prophesize the nondeterministic assignment to $\nt$
in $C_1$.  Rewriting:
\[\begin{array}{ll}
    C_3: & \kcode{p:=any; while x>n do x:=x-1 od; t:=p+x; z:=x+t} \\
    C_2': & \kcode{s:=any; while x>n do x:=x-1 od; z:=x+s}
  \end{array}
\]
To show
$C_3\sep C_2':\aespec{R}{S}$,
we use the witness technique described in Theorem~\ref{thm:fsim}.  Choose witness
$W$ to be:
\[
  W \eqdef
  \emb{\kcode{p:=any}}{\kcode{s:=any}}; B;
  \emb{X}{X}^\kstar;
  \eml{\neg e}; \emb{q}{q'}
\]
where $e \eqdef [\kcode{x>n}]$, $X \eqdef \left(e; \kcode{x:=x-1}\right)$,
$B \eqdef [\kcode{p+min(x,n)}\eqbi s]$,
$q \eqdef \kcode{t:=p+x; z:=x+t}$, and $q' \eqdef \kcode{z:=x+s}$.
There are three conditions to check.
\begin{list}{}{}
  \item[(WC)] $[\nx\eqbi\nx \land \nn\eqbi\nn]; W
               \leq [\nx\eqbi\nx \land \nn\eqbi\nn]; W; [\nz\eqbi\nz]$
  \item[(WU)] $[\nx\eqbi\nx \land \nn\eqbi\nn]; W \leq \emb{\hav}{C_2'}$
  \item[(WO)] $[\nx\eqbi\nx \land \nn\eqbi\nn]; \eml{C_3}
               \leq W; \emr{\hav}$
\end{list}
Note that the witness does not contain $\emr{\neg e}$.  This is a minor
technicality that helps shorten the proof of (WO), allowing us to avoid steps
that cancel $\emr{\neg e}$ on the right hand side.  The ommision is justified
by the following observation we take as axiom:  $R;\emb{\neg e}{\neg e} = R;\eml{\neg e}$.
To prove the above inequalities, we rely on the following axioms:
\begin{enumerate}
\item[(a)] $R; B$ is preserved by $\emb{X}{X}$, i.e., $R;B;\emb{X}{X} = R;B;\emb{X}{X};R;B$.
\item[(b)] $R$ is preserved by $\emb{X}{X}$.
\item[(c)] $R; \eml{\neg e} = R; \emb{\neg e}{\neg e}$ and $R; \eml{e} = R; \emb{e}{e}$
\item[(d)] $R; B; \emb{\neg e}{\neg e} \leq [\nx\eqbi\nx];\eml{\neg e};[\kcode{p+x}\eqbi\kcode{s}]$.
\item[(e)] $1 \leq \emr{\kcode{s:=any}}; B; \emr{\kcode{s:=any}}$ which
expresses left-totality of $B$.
% \item[(e)] $\eml{\kcode{p:=any}} \leq \emb{\kcode{p:=any}}{\kcode{s:=any}};B;\emr{\hav}$.
\item[(f)] $R$ is preserved by $\emb{\kcode{p:=any}}{\kcode{s:=any}}$, i.e.,
  $R; \emb{\kcode{p:=any}}{\kcode{s:=any}} = R; \emb{\kcode{p:=any}}{\kcode{s:=any}}; R$.
\item[(g)] $[\nx\eqbi\nx];[\kcode{p+x}\eqbi\kcode{s}];\emb{\kcode{t:=p+x; z:=x+t}}{\kcode{z:=x+s}}
  \leq \emb{\kcode{t:=p+x; z:=x+t}}{\kcode{z:=x+s}};[\nz\eqbi\nz]$.
\item[(h)] $\kcode{x:=e}; \hav = \hav$, for any assignment $\kcode{x:=e}$.
\end{enumerate}

\noindent $\bullet$ To prove (WC) start by considering the LHS,

\(\begin{array}[t]{lll}
      & R; \emb{\kcode{p:=any}}{\kcode{s:=any}}; B;
        \emb{X}{X}^\kstar;
        \eml{\neg e}; \emb{q}{q'} \\
    = & R; \emb{\kcode{p:=any}}{\kcode{s:=any}}; R; B;
        \emb{X}{X}^\kstar; \eml{\neg e};
        \emb{q}{q'}
      & \mbox{using (f)} \\
    = & R; \emb{\kcode{p:=any}}{\kcode{s:=any}}; R; B;
        \emb{X}{X}^\kstar; R; B; \eml{\neg e};
        \emb{q}{q'}
      & \mbox{using (a) and lemma~\ref{lem:invar}} \\
    = & R; \emb{\kcode{p:=any}}{\kcode{s:=any}}; R; B;
        \emb{X}{X}^\kstar; R; B; \emb{\neg e}{\neg e};
        \emb{q}{q'}
      & \mbox{using (c)} \\
 \leq & R; \emb{\kcode{p:=any}}{\kcode{s:=any}}; R; B;
        \emb{X}{X}^\kstar; \eml{\neg e};
        [\nx\eqbi\nx]; [\kcode{p+x}\eqbi\kcode{s}];
        \emb{q}{q'}
      & \mbox{using (d)} \\
 \leq & R; \emb{\kcode{p:=any}}{\kcode{s:=any}}; R; B;
        \emb{X}{X}^\kstar; \eml{\neg e};
        \emb{q}{q'}; [\nz\eqbi\nz]
      & \mbox{using (g)} \\
 \leq & R; \emb{\kcode{p:=any}}{\kcode{s:=any}}; B;
        \emb{X}{X}^\kstar; \eml{\neg e};
        \emb{q}{q'}; [\nz\eqbi\nz]
      & \mbox{$R$ below $1$} \\
    = & R; W; S
  \end{array} \)

\noindent $\bullet$ To prove (WU) we have to show
$[\nx\eqbi\nx \land \nn\eqbi\nn]; W \leq \emb{\hav}{C_2'}$.
We start by proving a general lemma about $\hav$: for any $c$ and $d$,
$$\emb{c}{d}^\kstar;\eml{\hav} \leq \emr{d}^\kstar; \eml{\hav}$$
By induction and join, it suffices to show: \( \eml{\hav} \leq \emr{d}^\kstar; \eml{\hav} \)
and \( \emb{c}{d}; \emr{d}^\kstar; \eml{\hav} \leq  \emr{d}^\kstar; \eml{\hav} \).
The former follows from the fact that $1 \leq x^\kstar$ for
any $x$.  The latter:
\[\begin{array}{lll}
       & \emb{c}{d}; \emr{d}^\kstar; \eml{\hav} \leq  \emr{d}^\kstar; \eml{\hav} \\
  \iff & \emr{d}; \eml{c}; \emr{d}^\kstar; \eml{\hav} \leq  \emr{d}^\kstar; \eml{\hav} & \mbox{LRC} \\
  \iff & \emr{d}; \emr{d}^\kstar; \eml{c}; \eml{\hav} \leq \emr{d}^\kstar; \eml{\hav}
       & \mbox{by lemma~\ref{lem:item-over-star}} \\
\impby & \emr{d}^\kstar; \eml{c}; \eml{\hav} \leq \emr{d}^\kstar; \eml{\hav}
       & \mbox{emb homo, $d; d^\kstar \leq d^\kstar$} \\
\impby & \emr{d}^\kstar; \eml{\hav} \leq \emr{d}^\kstar; \eml{\hav}
       & \mbox{$\hav$ is top}
  \end{array}
\]
and this follows from reflexivity.  For (WU), we calculate starting from the LHS,
\[\begin{array}{lll}
      & R; \emb{\kcode{p:=any}}{\kcode{s:=any}}; B;
        \emb{X}{X}^\kstar;
        \eml{\neg e}; \emb{q}{q'} \\
    = & R; \emb{\kcode{p:=any}}{\kcode{s:=any}}; R; B;
        \emb{X}{X}^\kstar; R; B;
        \emb{\neg e}{\neg e}; \emb{q}{q'}
      & \mbox{using (f), (a), (c), and lemma~\ref{lem:invar}} \\
 \leq & \emb{\kcode{p:=any}}{\kcode{s:=any}};
        \emb{X}{X}^\kstar;
        \emb{\neg e}{\neg e}; \emb{q}{q'}
      & \mbox{$R$ and $B$ below 1} \\
 \leq & \emb{\hav}{\kcode{s:=any}};
        \emb{X}{X}^\kstar;
        \emb{\neg e}{\neg e}; \emb{\hav}{q'}
      & \mbox{$\hav$ is top} \\
 \leq & \emb{\hav}{\kcode{s:=any}};
        \emb{X}{X}^\kstar;
        \eml{\hav}; \emr{\neg e; q'}
      & \mbox{$\hav$ is top, LRC} \\
 \leq & \emb{\hav}{\kcode{s:=any}};
        \emr{X}^\kstar;
        \eml{\hav}; \emr{\neg e; q'}
      & \mbox{using lemma above} \\
    = & \eml{\hav}; \emb{\hav}{\kcode{s:=any}}; \emr{X}^\kstar;
        \emr{\neg e; q'}
      & \mbox{LRC and lemma~\ref{lem:item-over-star}} \\
    = & \emb{\hav}{\kcode{s:=any}}; \emr{X}^\kstar;
        \emr{\neg e; q'}
      & \mbox{$\hav = \hav;\hav$} \\
    = & \emb{\hav}{C_2'}
      & \mbox{emb homo and defs}
  \end{array}
\]

\noindent $\bullet$ Finally for (WO) we have to
show $[\nx\eqbi\nx \land \nn\eqbi\nn]; \eml{C_3} \leq W; \emr{\hav}$.
We first prove the following lemma:
\[ R;\eml{X}^\kstar \leq R;\emb{X}{X}^\kstar;\emr{\hav} \]
This follows from induction and join, provided
\( R \leq R;\emb{X}{X}^\kstar;\emr{\hav} \) and
\( R;\emb{X}{X}^\kstar;\emb{X}{\hav} \leq R;\emb{X}{X}^\kstar;\emr{\hav} \).
The first inequality is immediate.  For the second, we calculate,
\[\begin{array}{lll}
      & R;\emb{X}{X}^\kstar;\emb{X}{\hav} \\
    = & R;\emb{X}{X}^\kstar;R;\emb{X}{\hav} & \mbox{using (b) and lemma~\ref{lem:invar}} \\
    = & R;\emb{X}{X}^\kstar;R;\emb{[\nx>\nn];\kcode{x:=x-1}}{\hav} & \mbox{def of $X$} \\
    = & R;\emb{X}{X}^\kstar;R;\emb{[\nx>\nn];\kcode{x:=x-1}}{[\nx>\nn];\hav} & \mbox{using (c)} \\
    = & R;\emb{X}{X}^\kstar;R;\emb{[\nx>\nn];\kcode{x:=x-1}}{[\nx>\nn];\kcode{x:=x-1};\hav}
      & \mbox{using (h)} \\
    = & R;\emb{X}{X}^\kstar;R;\emb{[\nx>\nn];\kcode{x:=x-1}}{[\nx>\nn];\kcode{x:=x-1}}\emr{\hav}
      & \mbox{emb homo} \\
    = & R;\emb{X}{X}^\kstar;R;\emb{X}{X};\emr{\hav}
      & \mbox{def of $X$} \\
 \leq & R;\emb{X}{X}^\kstar;\emb{X}{X};\emr{\hav}
      & \mbox{$R$ below $1$} \\
 \leq & R;\emb{X}{X}^\kstar;\emr{\hav}
      & \mbox{since $x^\kstar;x\leq x^\kstar$ for any $x$}
  \end{array}
\]
Now for (WO), we start from the LHS,
\[\begin{array}{lll}
      & R; \eml{C_3}  \\
    = & R; \emb{\kcode{p:=any}}{1}; \eml{X}^\kstar; \eml{\neg e}; \eml{q}
      & \mbox{unfolding $C_3$ and emb homo} \\
 \leq & R; \emb{\kcode{p:=any}}{\kcode{s:=any}};B;\emr{\kcode{s:=any}};\eml{X}^\kstar; \eml{\neg e}; \eml{q}
      & \mbox{using (e)} \\
 \leq & R; \emb{\kcode{p:=any}}{\kcode{s:=any}};B;\emr{\hav};\eml{X}^\kstar; \eml{\neg e}; \eml{q}
      & \mbox{$\hav$ is top} \\
    = & R; \emb{\kcode{p:=any}}{\kcode{s:=any}};B;\eml{X}^\kstar;\emr{\hav};\eml{\neg e}; \eml{q}
      & \mbox{using lemma~\ref{lem:item-over-star}} \\
    = & R; \emb{\kcode{p:=any}}{\kcode{s:=any}}; B; R; \eml{X}^\kstar;\emr{\hav};\eml{\neg e}; \eml{q}
      & \mbox{using (f) and commuting $R$,$B$} \\
 \leq & R; \emb{\kcode{p:=any}}{\kcode{s:=any}}; B; R; \emb{X}{X}^\kstar; \emr{\hav}; \emr{\hav};\eml{\neg e}; \eml{q}
      & \mbox{using lemma above} \\
 \leq & R; \emb{\kcode{p:=any}}{\kcode{s:=any}}; B; R; \emb{X}{X}^\kstar; \emb{\neg e; q}{\hav}
      & \mbox{LRC, emb homo, $\hav;\hav = \hav$} \\
    = & R; \emb{\kcode{p:=any}}{\kcode{s:=any}}; B; R; \emb{X}{X}^\kstar; \emb{\neg e; q}{q'; \hav}
      & \mbox{using (h) which implies $q';\hav = \hav$} \\
    = & R; \emb{\kcode{p:=any}}{\kcode{s:=any}}; B; R; \emb{X}{X}^\kstar; \eml{\neg e}; \emb{q}{q'}; \emr{\hav}
      & \mbox{LRC, emb homo} \\
 \leq & \emb{\kcode{p:=any}}{\kcode{s:=any}}; B; \emb{X}{X}^\kstar; \eml{\neg e}; \emb{q}{q'}; \emr{\hav}
      & \mbox{$R$ below 1} \\
    = & W; \emr{\hav}
  \end{array}
\]

\subsubsection{Forward Simulation Using Deductive Rules}\label{sec:fsudr}
Example.~\ref{eg:fsimProphecyEx} describes using the witness technique from
Theorem~\ref{thm:fsim} to prove $C_3\sep C_2' : \aespec{\kcode{x}
\eqbi \kcode{x} \land \kcode{n} \eqbi \kcode{n}}{\kcode{z} \eqbi \kcode{z}}$.
An alternative is to directly apply the forward simulation proof rules
described in Fig.~\ref{fig:aeRHL}.  Start by deriving an alignment of
$C_3$ and
$C_2'$:
\[
  \emb{\kcode{p:=any}}{\kcode{s:=any}};
  \emb{\kcode{[x>n]}; \kcode{x:=x-1}}{\kcode{[x>n]}; \kcode{x:=x-1}}^\kstar;
  \emb{\neg\kcode{[x>n]}}{\neg\kcode{[x>n]}};
  \eml{\kcode{t:=p+x}}; \emb{\kcode{z:=x+t}}{\kcode{z:=x+s}},
\]
and establishing the following triples:
\begin{enumerate}
\item[(1)] $\emb{\kcode{p:=any}}{\kcode{s:=any}} : \aespec{R}{B \land R}$,
  where $B \eqdef [\kcode{p+min(x,n)}\eqbi\kcode{s}]$
  and $R \eqdef [\kcode{x}\eqbi\kcode{x} \land \kcode{n}\eqbi\kcode{n}]$.

  From \rn{enAss} conclude
  $\emb{\kcode{p:=any}}{\kcode{s:=any}} : \aespec{\textrm{true}}{B}$.  The use
  of this rule requires proving $B$ is left total, which it is.  Then use the
  fact that $\emb{\kcode{p:=any}}{\kcode{s:=any}}$ does not modify $\kcode{x}$
  or $\kcode{n}$ and so preserves $R$.
\item[(2)] $\emb{\kcode{[x>n]}; \kcode{x:=x-1}}{\kcode{[x>n]}; \kcode{x:=x-1}}^\kstar;
  \emb{\neg\kcode{[x>n]}}{\neg\kcode{[x>n]}} : \aespec{B \land R}{B \land R
    \land \neg\eml{[\kcode{x>n}]} \land \neg\emr{[\kcode{x>n}]}}$.

  Follows from an application of \rn{eWh}.  We have to show:
  \begin{enumerate}
  \item [(2.1)] $B \land R \imp \kcode{[x>n]} \eqbi \kcode{[x>n]}$.

    This holds because $R$ implies agreement on $\kcode{x}$ and $\kcode{n}$.
  \item [(2.2)] $\emb{\kcode{x:=x-1}}{\kcode{x:=x-1}} :
    \aespec{B \land R \land \emb{\kcode{[x>n]}}{\kcode{[x>n]}}}{B \land R}$.

    This follows from \rn{eAss} and \rn{eConseq}.  Note that $R$ is clearly preserved by
    $\emb{\kcode{x:=x-1}}{\kcode{x:=x-1}}$.  To reason that $B$ is preserved,
    we rely on the fact that under assumption $\eml{\kcode{[x>n]}}$, the
    bitest $[\kcode{p+min(x,n)}\eqbi \kcode{s}]$ is equivalent to the bitest
    $[\kcode{p+n}\eqbi \kcode{s}]$ which is preserved by the loop body.
  \end{enumerate}
\item[(3)] $\eml{\kcode{t:=p+x}} :
  \aespec{\kcode{p+x}\eqbi\kcode{s} \land \kcode{x}\eqbi\kcode{x}}
         {\kcode{t}\eqbi\kcode{s} \land \kcode{x}\eqbi\kcode{x}}$.

   This follows from \rn{eAss} and the fact that $\kcode{t:=p+x}$ does not modify $x$.
\item[(4)] $\emb{\kcode{z:=x+t}}{\kcode{z:=x+s}} :
  \aespec{\kcode{t}\eqbi\kcode{s} \land \kcode{x}\eqbi\kcode{x}}{\kcode{z}\eqbi\kcode{z}}$.

  This follows from \rn{eAss} and \rn{eConseq}.
\end{enumerate}

Finally conclude $C_3\sep C_2' : \aespec{R}{S}$ by using
\rn{eSeq} and \rn{eConseq} and the triples described above.  One key step is
the use of \rn{eConseq} when sequencing triples (2) and (3) above: we must
show that
$B \land R \land \neg\eml{[\kcode{x>n}]} \land \neg\emr{[\kcode{x>n}]}$
implies $\kcode{p+x} \eqbi \kcode{s}$.  This follows from observing that
$\kcode{p+min(x,N)} \eqbi s \land \neg\eml{[\kcode{x>n}]}$ implies
$\kcode{p+x} \eqbi \kcode{s}$.

\subsubsection{Backward Simulation Example~\ref{eg:bsimEx}}

We want to show $C_1 \sep C_2 : \bespec{R}{S}$ 
where $R \eqdef \nx\eqbi\nx \land \nn\eqbi\nn$ and
$S \eqdef \nx\eqbi\nx \land \nn\eqbi\nn \land \kcode{t}\eqbi\kcode{s} \land \kcode{z}\eqbi\kcode{z}$.
We use the the witness technique described in
Theorem~\ref{thm:bsim}.  Take as witness:
\[ W \eqdef
  \emr{\kcode{s:=any}}; \emb{X}{X}^\kstar;
  \eml{\neg e}; \eml{\kcode{t:=any+x}}; \emb{\kcode{z:=x+t}}{\kcode{z:=x+s}}; S
\]
where $e \eqdef [\kcode{x>n}]$, and $X \eqdef e; \kcode{x:=x-1}$.
Notice that $W$ ends with the postrelation $S$.  We have three inequalities to prove,
\begin{list}{}{}
\item[(WCb)] $W; S \leq R; W; S$
\item[(WOb)] $\eml{C_1}; S \leq \emr{\hav}; W$
\item[(WUb)] $W; S \leq \emb{\hav}{C_2}$.
\end{list}

We take the following as axioms:
\begin{enumerate}
\item[(a)] $\hav; \kcode{x:=e} = \hav; [\kcode{x}=\kcode{e}]$ for any
  primitive assignment $\kcode{x:=e}$ where $\kcode{e}$ does not depend on
  $\kcode{x}$,
\item[(b)] $\kcode{x:=e} = \kcode{x:=e}; [\kcode{x}=\kcode{e}]$ where
  $\kcode{e}$ does not depend on $\kcode{x}$,
% \item[(c)] $\eml{\kcode{x:=x-1}}; R \leq \emr{\hav};\emb{\kcode{x:=x-1}}{\kcode{x:=x-1}};R$,
\item[(c)] $\hav; \kcode{x:=x-1} = \hav$, which expresses the fact that $\kcode{x:=x-1}$ is range total,
\item[(d)] $\eml{[\kcode{z}=\kcode{x+t}]}; S = \emb{[\kcode{z}=\kcode{x+t}]}{[\kcode{z}=\kcode{x+s}]}; S$,
\item[(e)] $R$ commutes with $\emb{\kcode{z:=x+t}}{\kcode{z:=x+s}}$ since $R$ does not depend on $z$,
\item[(f)] $R$ commutes with $\eml{\kcode{t:=any+x}}$ since $R$ does not depend on $t$,
\item[(g)] $R$ commutes with $\kcode{s:=any}$ since $R$ does not depend on $s$,
\item[(h)] $\hav; \kcode{s:=any} = \hav$,
% \item[(i)] $R$ is backwards preserved by $\emb{X}{X}$ (i.e., $\emb{X}{X};R = R;\emb{X}{X};R$),
\item[(i)] $R$ is backwards preserved by $\emb{\kcode{x:=x-1}}{\kcode{x:=x-1}}$
  (i.e., $ \emb{\kcode{x:=x-1}}{\kcode{x:=x-1}};R = R;\emb{\kcode{x:=x-1}}{\kcode{x:=x-1}};R$),
\item[(j)] $R; \eml{e} = R; \emb{e}{e}$, and $R; \eml{\neg e} = R; \emb{\neg e}{\neg e}$.
\end{enumerate}

We also note the following consequences of the axioms above:
\begin{enumerate}
\item[(k)] $\eml{\kcode{x:=x-1}}; R \leq \emr{\hav};
\emb{\kcode{x:=x-1}}{\kcode{x:=x-1}}; R$, which follows from (c).
\item[(l)] $R$ is backwards preserved by $\emb{X}{X}$, which follows from (i).
\item[(m)] $\eml{\kcode{t:=any+x}};\emb{\kcode{z:=x+t}}{\kcode{z:=x+s}}; S
  = R; \eml{\kcode{t:=any+x}};\emb{\kcode{z:=x+t}}{\kcode{z:=x+s}}; S$,
  which follows from (e), (f), and the fact that $S$ contains $R$.
\end{enumerate}

\noindent $\bullet$ For (WCb), we have to show $W;S \leq R;W;S$ which is:
\[\emr{\kcode{s:=any}}; \emb{X}{X}^\kstar; \eml{\neg e};
  \eml{\kcode{t:=any+x}}; \emb{\kcode{z:=x+t}}{\kcode{z:=x+s}}; S \leq
  R; W; S \].  Consider the left hand side.  We have,
\[\begin{array}{lll}
      & \emr{\kcode{s:=any}}; \emb{X}{X}^\kstar; \eml{\neg e};
        \eml{\kcode{t:=any+x}}; \emb{\kcode{z:=x+t}}{\kcode{z:=x+s}}; S \\
    = & \emr{\kcode{s:=any}}; \emb{X}{X}^\kstar; \eml{\neg e}; R;
        \eml{\kcode{t:=any+x}}; \emb{\kcode{z:=x+t}}{\kcode{z:=x+s}}; S
      & \mbox{using (m)} \\
    = & \emr{\kcode{s:=any}}; \emb{X}{X}^\kstar; R; \eml{\neg e};
        \eml{\kcode{t:=any+x}}; \emb{\kcode{z:=x+t}}{\kcode{z:=x+s}}; S
      & \mbox{commute tests} \\
    = & \emr{\kcode{s:=any}}; R; \emb{X}{X}^\kstar; R; \eml{\neg e};
        \eml{\kcode{t:=any+x}}; \emb{\kcode{z:=x+t}}{\kcode{z:=x+s}}; S
      & \mbox{using (l), lemma~\ref{lem:binvarEq}} \\
    = & R; \emr{\kcode{s:=any}}; \emb{X}{X}^\kstar; R; \eml{\neg e};
        \eml{\kcode{t:=any+x}}; \emb{\kcode{z:=x+t}}{\kcode{z:=x+s}}; S
      & \mbox{using (g)} \\
 \leq & R; \emr{\kcode{s:=any}}; \emb{X}{X}^\kstar; \eml{\neg e};
        \eml{\kcode{t:=any+x}}; \emb{\kcode{z:=x+t}}{\kcode{z:=x+s}}; S
      & \mbox{$R$ below 1} \\
    = & R; W; S
      & \mbox{def of $W$}
  \end{array} \]

\noindent $\bullet$ For (WUb), we have to show $W;S \leq \emb{\hav}{C_2}$ which is,
\[\emr{\kcode{s:=any}}; \emb{X}{X}^\kstar;
  \eml{\neg e}; \eml{\kcode{t:=any+x}}; \emb{\kcode{z:=x+t}}{\kcode{z:=x+s}}; S \leq
  \eml{\hav}; \emr{\kcode{s:=any}}; \emr{X}^\kstar; \emr{\neg e}; \emr{\kcode{z:=x+s}}
\]
We start by proving a lemma:
$\emb{X}{X}^\kstar \leq \eml{\hav};\emr{X}^\kstar$.  By (left) induction and
join it suffices to show, $1 \leq \eml{\hav};\emr{X}^\kstar$
and $\emb{X}{X};\eml{\hav};\emr{X}^\kstar \leq \eml{\hav};\emr{X}^\kstar$.
The first inequality is immediate.  For the second, consider the left hand
side.  We have:
\[\begin{array}{lll}
      & \emb{X}{X};\eml{\hav};\emr{X}^\kstar \\
    = & \emb{e}{e};\emb{\kcode{x:=x-1}}{\kcode{x:=x-1}};\eml{\hav};\emr{X}^\kstar
      & \mbox{def of $X$} \\
 \leq & \emb{e}{e};\emb{\hav}{\kcode{x:=x-1}};\eml{\hav};\emr{X}^\kstar
      & \mbox{$\hav$ is top} \\
 \leq & \emr{e};\emb{\hav}{\kcode{x:=x-1}};\eml{\hav};\emr{X}^\kstar
      & \mbox{$e$ below 1} \\
    = & \eml{\hav;\hav}; \emr{e;\kcode{x:=x-1}};\emr{X}^\kstar
      & \mbox{LRC, emb homo} \\
    = & \eml{\hav}; \emr{e;\kcode{x:=x-1}};\emr{X}^\kstar
      & \mbox{$\hav;\hav = \hav$} \\
 \leq & \eml{\hav}; \emr{X}^\kstar
      & \mbox{since $p;p^\kstar \leq p^\kstar$ for any $p$}
  \end{array}\]
Now for (WUb), we start from the left hand side and show,
\[\begin{array}{lll}
      & \emr{\kcode{s:=any}}; \emb{X}{X}^\kstar;
        \eml{\neg e}; \eml{\kcode{t:=any+x}}; \emb{\kcode{z:=x+t}}{\kcode{z:=x+s}}; S \\
    = & \emr{\kcode{s:=any}}; \emb{X}{X}^\kstar;
        R; \eml{\neg e} \eml{\kcode{t:=any+x}}; \emb{\kcode{z:=x+t}}{\kcode{z:=x+s}}; S
      & \mbox{using (m) and commuting tests} \\
    = & \emr{\kcode{s:=any}}; \emb{X}{X}^\kstar;
        R; \emb{\neg e}{\neg e}; \eml{\kcode{t:=any+x}}; \emb{\kcode{z:=x+t}}{\kcode{z:=x+s}}; S
      & \mbox{using (j)} \\
 \leq & \emr{\kcode{s:=any}}; \emb{X}{X}^\kstar;
        \emr{\neg e}; \emb{\kcode{t:=any+x}; \kcode{z:=x+t}}{\kcode{z:=x+s}};
      & \mbox{tests below 1, LRC, emb homo} \\
 \leq & \emr{\kcode{s:=any}}; \emb{X}{X}^\kstar;
        \emr{\neg e}; \emb{\hav}{\kcode{z:=x+s}}
      & \mbox{$\hav$ is top} \\
 \leq & \emr{\kcode{s:=any}}; \eml{\hav}; \emr{X}^\kstar;
        \emr{\neg e}; \emb{\hav}{\kcode{z:=x+s}}
      & \mbox{using lemma above} \\
 \leq & \eml{\hav; \hav}; \emr{\kcode{s:=any}}; \emr{X}^\kstar;
        \emr{\neg e}; \emr{\kcode{z:=x+s}}
      & \mbox{LRC, emb homo, LRC over star} \\
    = & \eml{\hav}; \emr{\kcode{s:=any}}; \emr{X}^\kstar;
        \emr{\neg e}; \emr{\kcode{z:=x+s}}
      & \mbox{$\hav;\hav = \hav$}
  \end{array} \]

\noindent $\bullet$ Finally for (WOb), we have to show $\eml{C_1};S \leq \emr{\hav};W$ which is,
\[ \eml{X}^\kstar; \eml{\neg e}; \eml{\kcode{t:=any+x}; \kcode{z:=x+t}}; S \leq
  \emr{\hav}; \emr{\kcode{s:=any}}; \emb{X}{X}^\kstar; \eml{\neg e};
  \eml{\kcode{t:=any+x}}; \emb{\kcode{z:=x+t}}{\kcode{z:=x+s}}; S
\]
By (left) induction and join there are two cases to consider:
\begin{enumerate}
\item[(n)] $\eml{\neg e}; \eml{\kcode{t:=any+x}; \kcode{z:=x+t}}; S \leq \emr{\hav}; W$, and
\item[(o)] $\eml{X}; \emr{\hav}; W \leq \emr{\hav}; W$.
\end{enumerate}

For (n) we start with the right hand side:
\[\begin{array}{lll}
      & \emr{\hav}; \emr{\kcode{s:=any}}; \emb{X}{X}^\kstar; \eml{\neg e};
        \eml{\kcode{t:=any+x}};\emb{\kcode{z:=x+t}}{\kcode{z:=x+s}};S \\
    = & \emr{\hav}; \emb{X}{X}^\kstar; \eml{\neg e};
        \eml{\kcode{t:=any+x}};\emb{\kcode{z:=x+t}}{\kcode{z:=x+s}};S
      & \mbox{using (h)} \\
\geq  & \emr{\hav};\eml{\neg e}; \eml{\kcode{t:=any+x}}; \emb{\kcode{z:=x+t}}{\kcode{z:=x+s}}; S
      & \mbox{since $1 \leq x^\kstar$} \\
    = & \eml{\neg e}; \emb{\kcode{t:=any+x};\kcode{z:=x+t}}{\hav; \kcode{z:=x+s}}; S
      & \mbox{LRC, emb homo} \\
    = & \eml{\neg e}; \emb{\kcode{t:=any+x};\kcode{z:=x+t}}{\hav; [\kcode{z}=\kcode{x+s}]}; S
      & \mbox{using (a)} \\
    = & \eml{\neg e}; \emb{\kcode{t:=any+x};\kcode{z:=x+t};[\kcode{z}=\kcode{x+t}]}{\hav; [\kcode{z}=\kcode{x+s}]}; S
      & \mbox{using (b)} \\
    = & \eml{\neg e}; \emb{\kcode{t:=any+x};\kcode{z:=x+t}}{\hav};
        \emb{[\kcode{z}=\kcode{x+t}]}{[\kcode{z}=\kcode{x+s}]}; S
      & \mbox{LRC, emb homo} \\
\geq  & \eml{\neg e}; \emb{\kcode{t:=any+x};\kcode{z:=x+t}}{\hav};
        \eml{[\kcode{z}=\kcode{x+t}]}; S
      & \mbox{using (d)} \\
    = & \eml{\neg e}; \emb{\kcode{t:=any+x}}{\hav};
        \eml{\kcode{z:=x+t};[\kcode{z}=\kcode{x+t}]}; S
      & \mbox{LRC, emb homo} \\
    = & \eml{\neg e}; \emb{\kcode{t:=any+x}}{\hav};
        \eml{\kcode{z:=x+t}}; S
      & \mbox{using (b) in reverse} \\
\geq  & \eml{\neg e}; \eml{\kcode{t:=any+x};\kcode{z:=x+t}}; S
      & \mbox{$\hav$ is top, emb homo}
  \end{array}\]

Now for (o), (writing $q$ for $\kcode{z:=x+t}$, $q'$ for $\kcode{z:=x+s}$,
$s$ for $\kcode{s:=any}$, and $t$ for $\kcode{t:=any+x}$)
\[\begin{array}{lll}
     & \emb{X}{\hav;s}; \emb{X}{X}^\kstar; \eml{\neg e; t};
       \emb{q}{q'};S
       \leq
       \emr{\hav; s}; \emb{X}{X}^\kstar; \eml{\neg e; t};
       \emb{q}{q'};S \\
\iff & \emb{X}{\hav}; \emb{X}{X}^\kstar; \eml{\neg e; t};
       \emb{q}{q'};S
       \leq
       \emr{\hav}; \emb{X}{X}^\kstar; \eml{\neg e; t};
       \emb{q}{q'};S
     & \mbox{using (h)} \\
\iff & \emb{X}{\hav}; \emb{X}{X}^\kstar; \eml{\neg e}; R; \eml{t};
       \emb{q}{q'};S
       \leq
       \emr{\hav}; \emb{X}{X}^\kstar; \eml{\neg e; t};
       \emb{q}{q'};S
     & \mbox{using (m)} \\
\iff & \emb{X}{\hav}; \emb{X}{X}^\kstar; R; \eml{\neg e; t};
       \emb{q}{q'};S
       \leq
       \emr{\hav}; \emb{X}{X}^\kstar; \eml{\neg e; t};
       \emb{q}{q'};S
     & \mbox{commute tests} \\
\impby & \emb{X}{\hav};\emb{X}{X}^\kstar;R \leq \emr{\hav};\emb{X}{X}^\kstar
     & \mbox{monotonicity}
  \end{array}\]

For the left hand side, we show:
\[\begin{array}{lll}
       & \emb{X}{\hav};\emb{X}{X}^\kstar;R \\
    =  & \emb{e;\kcode{x:=x-1}}{\hav};\emb{X}{X}^\kstar; R
       & \mbox{def of $X$} \\
    =  & \emb{e;\kcode{x:=x-1}}{\hav}; R; \emb{X}{X}^\kstar; R
       & \mbox{using (l), lemma~\ref{lem:binvarEq}} \\
    =  & \emb{e}{\hav};\eml{\kcode{x:=x-1}};R; \emb{X}{X}^\kstar;R
       & \mbox{LRC, emb homo} \\
 \leq  & \emb{e}{\hav}; \emr{\hav}; \emb{\kcode{x:=x-1}}{\kcode{x:=x-1}}; R; \emb{X}{X}^\kstar;R
       & \mbox{using (k)} \\
    =  & \emb{e}{\hav}; \emb{\kcode{x:=x-1}}{\kcode{x:=x-1}}; R; \emb{X}{X}^\kstar;R
       & \mbox{using $\hav;\hav = hav$} \\
    =  & \emb{e}{\hav}; R; \emb{\kcode{x:=x-1}}{\kcode{x:=x-1}}; R; \emb{X}{X}^\kstar;R
       & \mbox{using (i)} \\
    =  & \emr{\hav}; R; \emb{e; \kcode{x:=x-1}}{\kcode{x:=x-1}}; R; \emb{X}{X}^\kstar;R
       & \mbox{LRC, emb homo, tests commute} \\
    =  & \emr{\hav}; R; \emb{e; \kcode{x:=x-1}}{e; \kcode{x:=x-1}}; R; \emb{X}{X}^\kstar;R
       & \mbox{using (j), emb homo} \\
 \leq  & \emr{\hav}; \emb{e; \kcode{x:=x-1}}{e; \kcode{x:=x-1}}; \emb{X}{X}^\kstar
       & \mbox{$R$ below 1} \\
 \leq  & \emr{\hav}; \emb{X}{X}^\kstar
       & \mbox{lemma $p;p^\kstar \leq p^\star$ for any $p$}
  \end{array}\]

% \ramana{Start with the following candidate alignment:}
% \[ \begin{array}{ll}
%   & \emb{\kcode{y:=x}}{\kcode{y:=x}};
%     \emb{\kcode{z:=24}}{\kcode{z:=16}};
%     \emb{\kcode{w:=0}}{\kcode{w:=0}}; \\
%   & \quad \Big( \eml{[\kcode{w\%2}\neq\kcode{0}]};
%            \eml{( [\kcode{w\%2=0}]; \kcode{z:=z*y}; \kcode{y:=y-1}  + [\kcode{w\%2}\neq\kcode{0}]  ); \kcode{w:=w+1}} \\
%   & \quad + \emr{[\kcode{w\%3}\neq\kcode{0}]};
%            \emr{( [\kcode{w\%3=0}]; \kcode{z:=z*2}; \kcode{y:=y-1}  + [\kcode{w\%3}\neq\kcode{0}]  ); \kcode{w:=w+1}} \\
%   & \quad + \emb{[\kcode{w\%2=0}]}{[\kcode{w\%3=0}]};
%     \emb{ [\kcode{w\%2=0}]; \kcode{z:=z*y}; \kcode{y:=y-1} + [\kcode{w\%2}\neq\kcode{0}]  }
%         { [\kcode{w\%3=0}]; \kcode{z:=z*2}; \kcode{y:=y-1} + [\kcode{w\%3}\neq\kcode{0}]  }; \\
%  & \quad  \emb{\kcode{w:=w+1}}{\kcode{w:=w+1}}
%     \Big)^\kstar; \\
%   & \quad \emb{\neg [\kcode{y>4}]}{\neg [\kcode{y>4}]}
%   \end{array} \]

% Alternatives (TODO)
% 1.  C1|C2 forward simulation without prophecy.  Point being we can align the
% two nondet assignments and need not align the loops.  Will require a history
% variable.
% 2.  C2|C1 forward simulation with the same spec.
% 3.  C1|C2 backward simulation with the same spec.

\subsubsection{Example~\ref{eg:paths}, $\forall\exists$ Path Alignment.}

We first prove the correctness of each witness $W_i$ under the pre-relation $R = \mkEqBi{l}$ and the post-relation $S = \mkEqBi{o}$ then prove the correctness of the overall witness $W$ under the same pre- and post-relation with the disjunction rules.	

We only show the correctness proof of $W_2$ as an example. The correctness proofs of the other $W_i$ is similar. Recall that $W_2$ corresponds to the precondition $\kcode{h} > \kcode{l} \wedge \kcode{h'} \leq \kcode{l'}$ with the chooser $\kcode{a2'} > \kcode{l'} \wedge \kcode{a2'} = \kcode{l} + \kcode{a1}$ on the nondeterministic value in the right program. 

\noindent$\bullet$ (WC) $R; W_2 \leq R; W_2; S$

\[\begin{array}[t]{lll}
	& R; W_2 & \\
= & [\kcode{l} = \kcode{l'}]; \emb{ \kcode{a1:=any; a2:=any} }{ \kcode{a1':=any; a2':=any} }; &\\
   & [\kcode{h>l; h'<=l'; a2'>l'; a2'=l+a1}]; &\\
   & \emb{\kcode{o:=l+a1}}{\kcode{x':=a2'; o':=x'}}\\
= &  [\kcode{l} = \kcode{l'}]; \emb{ \kcode{a1:=any; a2:=any} }{ \kcode{a1':=any; a2':=any} }; &\\
   & [\kcode{l=l'; h>l; h'<=l'; a2'>l'; a2'=l+a1}]; & $R is preserved$ \\ 
   & \emb{\kcode{o:=l+a1}}{\kcode{x':=a2'; o':=x')}} &\\
\leq &  [\kcode{l} = \kcode{l'}]; \emb{ \kcode{a1:=any; a2:=any} }{ \kcode{a1':=any; a2':=any} }; &\\
& [\kcode{l=l'; h>l; h'<=l'; a2'>l'; a2'=l+a1}]; & \\ 
& \emb{\kcode{o:=l+a1}}{\kcode{x':=a2'; o':=x')}}[\kcode{o} = \kcode{o'}]\\
= & R; W_2; S
\end{array}\]
Using rule \rn{DisjWC} in Fig.~\ref{fig:disjLem} we can conclude the (WC) proof for the witness $W$ from the (WC) proof of each $W_i$.

\noindent$\bullet$ (WU) $R; W_2 \leq \emb{\hav}{k'_2}$

\[\begin{array}[t]{lll}
   & R; W_2 &\\
= & [\kcode{l} = \kcode{l'}]; \emb{ \kcode{a1:=any; a2:=any} }{ \kcode{a1':=any; a2':=any} }; &\\
   & [\kcode{h>l; h'<=l'; a2'>l'; a2'=l+a1}]; &\\
   & \emb{\kcode{o:=l+a1}}{\kcode{x':=a2'; o':=x'}}\\
\leq &  \emb{ \kcode{a1:=any; a2:=any} }{ \kcode{a1':=any; a2':=any} }; &  \mbox{tests below 1}\\
       & \emb{\kcode{[h>l]; o:=l+a1}}{\kcode{[h'<=l'; a2'>l']; x':=a2'; o':=x'}} &\\
= & \multicolumn{2}{l}{\emb{\kcode{a1:=any; a2:=any; [h>l]; o:=l+a1}}{\kcode{a1:=any; a2:=any; [h'<=l'; a2'>l']; x':=a2'; o':=x'}}}\\
\leq & \emb{\hav}{\kcode{a1:=any; a2:=any; [h'<=l']; x':=a2'; [x'>l']; o':=x'}} & \\
   & &  \mbox{emb homo}\\
= & \emb{\hav}{k'_2} &\\
%\leq & \emb{\hav}{k'_1} + \emb{\hav}{k'_2} + \emb{\hav}{k'_3} & \\
%= & \emb{\hav}{k'_1 + k'_2 + k'_3} & \mbox{emb homo}\\
%= & \emb{\hav}{k'}

\end{array}\]
Similarly, we have the (WU) proofs for the other witnesses
\[\begin{array}[t]{lll}
	R; W_1 \leq \emb{\hav}{k'_1} &
	R; W_3 \leq \emb{\hav}{k'_1} &
	R; W_4 \leq \emb{\hav}{k'_1}\\
	R; W_5 \leq \emb{\hav}{k'_2} & 
	R; W_6 \leq \emb{\hav}{k'_3} &\\
\end{array}\]
From the (WU) proof of each $W_i$, applying rule \rn{DisjWU} we have 
$R; (W_1 + W_2 + W_3 + W_4 + W_5 + W_6) \leq \emb{\hav}{k'_1 + k'_2 + k'_3}$ or equivalently $R; W \leq \emb{\hav}{k'}$ .

\noindent$\bullet$ (WO) 
\[\begin{array}[t]{lll}
  & R; \eml{k_1} \leq W_2; \emr{\hav} & \\
\impby & \eml{k_1} \leq \emb{ \kcode{a1:=any; a2:=any} }{ \kcode{a1':=any; a2':=any} }; & \mbox{by } R \leq 1\\
      & \qquad \qquad [\kcode{h>l; h'<=l'; a2'>l'; a2'=l+a1}]; &\\
      &  \qquad \qquad \emb{\kcode{o:=l+a1}}{\kcode{x':=a2'; o':=x'}}; \emr{\hav} &\\
\impby & \eml{k_1} \leq \eml{\kcode{a1:=any; a2:=any; [h > l]; o:=l+a1}}; & \\
     & \qquad\qquad \emr{ \kcode{a1':=any; a2':=any;} }; \kcode{[h'<=l'; a2'>l'; a2'=l+a1]};\\ 
     & \qquad \qquad \emr{\kcode{x':=a2'; o':=x'}}; \emr{\hav} & \mbox{emb homo, LRC}\\
%\iff & \multicolumn{2}{l}{\eml{k_1} \leq \eml{k_1}; \emr{ \kcode{a1':=any; a2':=any;} }; \kcode{[h'<=l'; a2'>l'; a2'=l+a1]}; \emr{\hav}}\\
\iff & \eml{k_1} \leq \eml{k_1}; \emr{ \kcode{a1':=any; a2':=any;} }; \kcode{[h'<=l'; a2'>l'; a2'=l+a1]}; \emr{\hav} \\
\impby & 1 \leq \emr{ \kcode{a1':=any; a2':=any;} }; \kcode{[h'<=l'; a2'>l'; a2'=l+a1]}; \emr{\hav} & \rn{EnAss}
\end{array}
\]
%
%% \[\begin{array}[t]{lll}
%%   & R; \eml{k_1} \leq W_2; \emr{\hav} & \\
%% \impby & \eml{k_1} \leq \emb{ \kcode{a1:=any; a2:=any} }{ \kcode{a1':=any; a2':=any} }; & R \leq 1\\
%%       & \qquad \qquad [\kcode{h>l; h'<=l'; a2'>l'; a2'=l+a1}]; &\\
%%       &  \qquad \qquad \emb{\kcode{o:=l+a1}}{\kcode{x':=a2'; o':=x'}}; \emr{\hav} &\\
%% \impby & \eml{k_1} \leq \eml{\kcode{a1:=any; a2:=any; [h > l]; o:=l+a1}}; & \\
%%      & \multicolumn{2}{l}{\qquad \qquad \emr{ \kcode{a1':=any; a2':=any;} }; \kcode{[h'<=l'; a2'>l'; a2'=l+a1]};}\\ 
%%      & \qquad \qquad \emr{\kcode{x':=a2'; o':=x'}}; \emr{\hav} & \mbox{emb homo, LRC}\\
%% \iff & \multicolumn{2}{l}{\eml{k_1} \leq \eml{k_1}; \emr{ \kcode{a1':=any; a2':=any;} }; \kcode{[h'<=l'; a2'>l'; a2'=l+a1]}; \emr{\hav}}\\
%% \impby & \multicolumn{2}{l}{1 \leq \emr{ \kcode{a1':=any; a2':=any;} }; \kcode{[h'<=l'; a2'>l'; a2'=l+a1]}; \emr{\hav}} \\
%% & & \rn{EnAss}
%% \end{array}
%% \]
%
Similarly, we have:
\[\begin{array}[t]{lll}
	R; \eml{k_1} \leq W_1; \emr{\hav} &
	R; \eml{k_2} \leq W_3; \emr{\hav} &
	R; \eml{k_3} \leq W_4; \emr{\hav}\\
	R; \eml{k_2} \leq W_5; \emr{\hav} &
	R; \eml{k_3} \leq W_6; \emr{\hav} &
\end{array}\]

Applying \rn{DisjWO} on those proofs, we have $R; \eml{k_1 + k_2 + k_3} \leq (W_1 + W_2 + W_3 + W_4 + W_5 + W_6); \emr{\hav}$ or equivalently $R; \eml{k} \leq W; \emr{\hav}$.

\subsection{Additional Alignment Examples}%\label{apx:lit-examples}

This section considers additional examples from the literature, for $\forall\forall$ properties, showing how their alignments can be expressed in BiKAT.

\begin{example}[Lock-step alignment for strength reduction]
Sharma~\etal\citep{Sharma2013} describe methods for inferring equivalence between loops by observing and correlating executions. The following exampleis re-produced from their Figure 1.
%\begin{wrapfigure}[6]{r}[34pt]{3.9in}
\begin{center}
\begin{tabular}{ll}
\begin{minipage}{2.5in}
\begin{lstlisting}
int f(int x, int n){ int k = 0;
  for (i=0; i!=n; ++i){
    x += k*5; k += 1;
    if (i >= 5) k += 3;
  }
  return x; }
\end{lstlisting}
\end{minipage}
&
\begin{minipage}{2.5in}
\begin{lstlisting}
int ff(int x, int n){ int k = 0;
  for (i=0; i!=n; ++i){
    x += k; k += 5;
    if (i >= 5) k += 15;
  }
  return x; }
\end{lstlisting}
\end{minipage}
\end{tabular}
\end{center}
%\end{wrapfigure}
This strength-reduction optimization replaces multiplications with additions, improving performance. Their algorithm discovers a relational invariant between the aligned loops that 
\lstinline|i*5| in \lstinline|f| is equal to \lstinline|i<<2+i| in \lstinline|ff|.
This lock-step alignment can be expressed in BiKAT, as follows:
\[
\begin{array}{ll}
\mkEqBi{x,n,i,k};\\
(\embRepeat{ \kcode{i!=n} };
[5\kcode{i} \eqbi \kcode{i<<2}+\kcode{i}];
\emb{\kcode{f} \text{ loop body}}{\kcode{ff} \text{ loop body}};
)^\kstar ; \embRepeat{ \kcode{i=n} }; \\
{}[5\kcode{i} \eqbi \kcode{i<<2}+\kcode{i}]; \mkEqBi{x}
\end{array}\]
Here we have, using BiKAT structural rules instead of control-flow graphs, algebraically derived a proof of the
$\forall\forall$ property of the original program, with help from the Sharma~\etal~technique for inferring the bitest relational loop invariant
$
[5\kcode{i} \eqbi \kcode{i<<2}+\kcode{i}]
$.
\end{example}

\begin{example}[Sequential alignment and memory actions]
BiKAT is parametric on the alphabet of actions and so it is possible to work with other actions such as memory references such as assignment \kcode{*x=v} and dereference \kcode{v=*x}.
Blatter\etal\citep{Blatter22} describe a certified relational verification approach that directly constructs verification conditions and avoids self-composition. 
They provide an example comparing two implementations of swapping values in memory:
\[\begin{array}{llllll}
C_1:&\kcode{x3=*x1; *x1=*x2; *x2 = x3;}&\hspace{0.5in}&
C_2:&\kcode{*x1=*x1 + *x2; *x2 = *x1-*x2; *x1=*x1-*x2;}\\
\end{array}\]
Above $C_1$ uses a temporary variable \kcode{x3}, whereas $C_2$ uses addition/subtraction. In this case a sequential alignment is sufficient, with an intermediate relation that relates \kcode{*x1} on the left with \kcode{*x2} on the right and vice-versa:
\[\begin{array}{l}
\mkEqBi{*x1,*x2};
\eml{ C_1 };
[\kcode{*x1} \eqbi \kcode{*x2}
 \wedge \kcode{*x2} \eqbi \kcode{*x1}];
\emr{ C_2 };
\mkEqBi{*x1,*x2}
\end{array}\]
\end{example}

\begin{example}[Using relational specifications]%\label{eg:mpp}
As seen in the previous example, BiKAT permits one to use specifications such as procedure pre/post conditions as hypotheses in the logic. The previous example used a unary specification via left/right embedding, but relational specifications can also be directly via bitests and equations of the form
$P;\emb{m()}{m()};\bneg Q = 0$.

Eilers~\emph{et al.}~\citep{EilersMH18} describe modular product programs (MPP) which has a method for supporting procedure-level modularity as part of a reduction to reduce $k$-safety tasks to problems for unary solvers such as Viper. A comparison of approaches is given in Section~\ref{sec:discuss};
in brief, while MPP is automated and has full interprocedural support, it is limited to $k$-safety of single programs, limited to lock-step alignment, and does not provide an algebraic way to justify the relationship with the original program.  

By contrast, BiKAT can directly incorporate  relational specifications and can express a wider range of alignment, including the lock-step alignments used in MPPs. 
For example, the relational specification for 
the program in their Fig.~1 can be written in BiKAT 
along with calls as
$\mkEqBi{person};\emb{ \kcode{is\_female(person)} }{ \kcode{is\_female(person)} };\bneg\mkEqBi{res}=\bzero$,
and the unary specification is written
$
[\kcode{person}>0];
\kcode{is\_female(person)};
\neg \textrm{true} = 0
$, which can be embedded on the left or on the right. We can then give an algebraically justified BiKAT alignment for their Fig.~1 \kcode{main} procedure as follows:
\[\begin{array}{l}
( \embRepeat{ \kcode{i<people} }\\
\;\;\;\;\;\;\mkEqBi{i,count,people}\\
\;\;\;\;\;\;\embRepeat{ \kcode{current:=people[i]} }\\
\;\;\;\;\;\;\mkEqBi{current}
\embRepeat{ \kcode{f:=is\_female(current)} }
\mkEqBi{f}\\
\;\;\;\;\;\;\embRepeat{ \kcode{count := count + f} }\\
\;\;\;\;\;\;\embRepeat{ \kcode{i := i + 1} }
)^\bstar\\
\embRepeat{ \kcode{i>=people} }
\mkEqBi{count}
\end{array}\]
Above is the lock-step alignment used by MPP, except without the activation variables (see Section~\ref{sec:discuss} for a discussion on activation variables). The loops are lock-step aligned, the relational loop invariant is used, the relational specification of \kcode{is\_female} is employed, and the loop invariant implies the final bitest, which entails that the property holds.

\end{example}

\end{document}
\endinput
